\documentclass[article] {amsart}

\usepackage[a4paper, total={6in, 8in}]{geometry}

\usepackage[english]{babel}
\usepackage{blindtext}
\usepackage{a4wide}
\usepackage[utf8x]{inputenc}
\usepackage[T1]{fontenc}
\usepackage{rotating}
\usepackage{lscape}
\usepackage{amsaddr}
\usepackage{hyperref}
\usepackage{amsfonts,amsmath,amssymb,amsthm,amsrefs}
\usepackage{latexsym}
\usepackage{siunitx}
\usepackage{layout}
\usepackage{dsfont}
\usepackage{xcolor}
\usepackage{color,soul}
\usepackage{graphicx}
\usepackage{graphics}
\usepackage{tablefootnote}
\usepackage{multirow}
\usepackage{booktabs}
\usepackage{tikz}
\usepackage{caption}
\usepackage{subcaption}
\usepackage{comment}
\usepackage{footmisc}
\usepackage{bm}

\usetikzlibrary{shapes.geometric,}
\usetikzlibrary{trees,arrows}
\usetikzlibrary{automata, positioning}

\usepackage[linesnumbered,ruled,vlined]{algorithm2e}
\SetKwInput{KwSetup}{Setup}                
\SetKwInput{KwInput}{Input}                

\SetCommentSty{mycommfont}

\tikzstyle{arrow} = [thin,->,>=stealth]

\numberwithin{equation}{section}
\usepackage{graphicx}
\newtheorem{theorem}{Theorem}[section]
\newtheorem{proposition}[theorem]{Proposition}
\newtheorem{lemma}[theorem]{Lemma}

\newtheorem{assump}[theorem]{Assumption}
\newtheorem{remark}[theorem]{Remark}

\newtheorem{defi}[theorem]{Definition}

\newcommand{\e}{\mathbb{E}}

\newcommand{\reals}{\mathbb{R}}
\newcommand{\ind}{\mathbf{1}}

\newcommand{\diff}{\,\mathrm{d}}

\newcommand{\Cr}{\mathcal{C}} 

\DeclareMathOperator*{\esssup}{ess\,sup}

\newcommand{\prob}[1]{\mathbb{#1}}
\newcommand{\vix}{\mathrm{VIX}}
\newcommand{\FmN}{F^{(m,N)}}
\newcommand{\XmN}{X^{(m,N)}}
\newcommand{\Vm}{V^{(m)}}

\allowdisplaybreaks

\title[VIX-Linked fee incentives in variable annuities]{Analysis of VIX-linked fee incentives in variable annuities via continuous-time Markov chain approximation}

\author{Anne Mackay$^{\ast}$$^\dag$, Marie-Claude Vachon$^\ddag$ and Zhenyu Cui$^\S$}
\address{$\dag$Universit\'e de Sherbooke, Sherbrooke, Qu\'ebec, Canada\\
$\ddag$ Universit\'e du Qu\'ebec \`a Montr\'eal, Montr\'eal,  Qu\'ebec, Canada\\
$\S$Stevens Institute of Thechnology, Hoboken, NJ 07030, USA\\
}
\thanks{$^\ast$Corresponding author. Email: Anne.MacKay@USherbrooke.ca\\}

\date{\today}

\begin{document}
	

	\begin{abstract}
	We consider the pricing of variable annuities (VAs) with general fee structures under popular stochastic volatility models such as Heston, Hull-White, Scott, $\alpha$-Hypergeometric, $3/2$, and $4/2$ models. 
	In particular, we analyze the impact of different VIX-linked fee structures on the optimal surrender strategy of a VA contract with guaranteed minimum maturity benefit (GMMB). 
	Under the assumption that the VA contract can be surrendered before maturity, the pricing of a VA contract corresponds to an optimal stopping problem  with an unbounded, time-dependent, and discontinuous payoff function. 
	We develop efficient algorithms for the pricing of VA contracts using a two-layer continuous-time Markov chain approximation for the fund value process. 
	When the contract is kept until maturity and under a general fee structure, we show that the value of the contract can be approximated by a closed-form matrix expression. 
	We also provide a quick and simple way to determine the value of early surrenders via a recursive algorithm and give an easy procedure to approximate the optimal surrender surface. 
	We show numerically that the optimal surrender strategy is more robust to changes in the volatility of the account value when the fee is linked to the VIX index.
	\end{abstract}
	
	\maketitle
	
	\newpage
	\section{Introduction}
A variable annuity (or segregated fund in Canada) is a hybrid investment vehicle mainly used for retirement planning, which offers a life insurance benefit and a financial guarantee.  
It allows the policyholder to profit from potential gains resulting from an investment in financial markets, while offering a downside protection against losses. 
The real options embedded in these products are comparable to exotic options, with the following differences: the benefit may depend on the policyholder's life (survival or death), they are long-term investments (generally between 5 and 15 years, or more), and the financial guarantee is funded via a periodic fee (typically set as a percentage of the fund value) as opposed to a premium paid upfront. 
Different types of protection riders are offered, such as Guaranteed Minimum Maturity Benefit (GMMB), Guaranteed Minimum Death Benefit (GMDB), and Guaranteed Minimum Withdraw Benefit (GMWB); see \cite{hardy2003investment} or \cite{bauer2008universal} for details\footnote{There are no consensus among practitioners and scientists for these products' name, and thus, different authors may use different terminologies for the same product.}. 
This paper focuses on the GMMB rider, which guarantees the policyholder a minimum amount at the contract's maturity.  
Considering the significant size of the variable annuity market\footnote{In United-States, the total variable annuity sales were \$125 billion in 2021, representing  an increase of 25\% with respect to the total VA sales in 2020. Source: LIMRA Secure Retirement Institute, U.S. Individual Annuities survey \url{https:https://www.limra.com/siteassets/newsroom/fact-tank/sales-data/2021/q4/2012-2021-annuity-sales-updated.pdf}.}, the management of the risk associated with the guarantees embedded in variable annuities is a major concern for insurance companies, see \cite{niittuinpera2020IAA}. 
Indeed, variable annuities guarantees entail significant risks given their long-term structure and sensitivity to various financial and demographic risks as well as to policyholders' behaviour.
For the GMMB rider, this last risk is mostly due to early surrenders.\\
	\\
  Surrender risk refers to the uncertainty facing the insurer when a policyholder has the possibility to terminate her contract before its maturity.
  When she does so, she is entitled to the value accumulated in the variable annuity investment account, subject to a penalty. 
  \cite{kling2014impact} show that unexpected lapses can represent a significant risk for insurers.  
  For this reason, surrender risk has raised special attention in the literature (\cite{niittuinpera2020IAA}, Chapter 18). \cite{bacinello2011variable} provide a universal pricing framework for various riders and considers  different types of surrender behaviours: static, i.e. the contract is never surrendered; or mixed, i.e. the policyholder acts rationally and surrenders the policy as soon as it is optimal from a risk-neutral valuation perspective. 
  Pricing variable annuities under rational surrender behaviour is equivalent to solving an optimal stopping problem and corresponds to the worst-case scenarios for insurers, in the sense that it maximizes the risk-neutral value of the contract from the policyholders' perspective.
   Under this assumption, \cite{grosen1997valuation} study the valuation of interest rate guarantees by assuming that the surrender value will be the same as the benefit value. 
   \cite{milevsky2001real} assume that the policyholder will only get a certain percentage of the fund upon surrender; this hypothesis is more in line with policies seen in practice. 
   Under this assumption, they provide a closed-form analytical solution to the price of a GMDB with surrender in the Black-Scholes framework. 
   In particular, they study the interaction between the surrender charges, the fee rates, and the optimal surrender level. 
   \cite{bernard2014optimal} study a problem similar to the one of \cite{milevsky2001real}, but focus on a GMMB rather than a GMDB. 
   It is well-known that American options with finite maturity generally do not have closed-form solutions. 
   Thus, \cite{bernard2014optimal} used arbitrage-free techniques in the same vein as \cite{kim1990analytic} 
   and \cite{carr1992alternative} in the context of American call and put options to derive an analytical expression for the value of the right to surrender, which is analogous to the early exercise premium in American option terminology. 
   In particular, they study the impact of different risk factors influencing the optimal surrender boundary. 
   In that context, \cite{BernardMackay2015} provide a sufficient condition on surrender charges and fees which eliminate surrender incentives for a financially rational policyholder. 
   Recently, \cite{kang2018optimal} extended the framework of \cite{bernard2014optimal}  by analyzing how the optimal surrender boundary is affected by changes in different risk factors in a stochastic interest rate and volatility model of the Heston-Hull-White type, whereas \cite{alonso2020taxation} incorporate taxes into their analysis. \\
	\\
	Other fee structure designs have also been explored by different authors.
	The idea behind those designs is usually to reduce the insurer's exposure to various risks, such as market volatility and policyholder behaviour.
	\cite{Bernard2014} introduce a state-dependent fee structure, where the fee is only paid when the VA account value is below a certain level, and 
	present an analytical formula for a the value of a contract with GMMB rider (without early surrender) under this type of fee structure. \cite{Mackay2017} study how the fee structure and surrender charges affect the surrender region; they also design surrender charges that eliminate surrender incentives for a financially rational policyholder. 
	Other fee designs have been explored in the literature: 
	 \cite{delong2014pricing} considers a general state-dependent fee structure in a L\'{e}vy process driven market, whereas \cite{bernard2019less} study lapse-and-reentry in variable annuities with time-dependent fee structure. Finally, in a recent study, \cite{wang2021optimal} propose a stochastic control approach to determine the optimal fee structure.\\
	\\
	Recently, fee structures that are tied to the Chicago Board Options Exchange (CBOE) volatility index, the VIX\footnote{See \url{https://www.cboe.com/tradable_products/vix/}}, have gained attention in the literature, see \cite{bernard2016variable}, \cite{cui2017} and \cite{kouritzin2018vix}. 
	As mentioned in \cite{bernard2016variable}, the motivations behind this new fee design comes directly from the industry. 
	In 2010, SunAmerica issued two new variable annuities whose fees were tied to the volatility index\footnote{See Retirement Income Journal available at \url{https://retirementincomejournal.com/article/sunamerica-links-va-rider-fees-to-volatility-index/}. }. 
	More recently, America General Life Insurance Company proposed a fee structure that is linked to the VIX for its Polaris series of variable annuities, see Polaris Platinum O-Series prospectus dated May 3rd, 2021\footnote{See footnote 6 on p.9 of the prospectus (the long-form) available at \url{https://aig.onlineprospectus.net/AIG/867018103A/index.php?open=POLARIS!5fPLATINUM!5fO-SERIES!5fISP.pdf}.}. 
	By allowing the fees to move with the volatility index, the insurer expects to better match the cost of hedging with the premium collected. It also permits policyholders to profit from lower fees in low volatility, rising market, environments. 
	The CBOE published two white papers, \cite{CBOEvixVA1} and \cite{CBOEvixVA2}, illustrating how VIX-linked fee designs can be advantageous to both variable annuity insurers and policyholders. \cite{cui2017} approach the question from a theoretical perspective by analyzing variable annuities without surrender with a fee structure that is tied to the VIX under a Heston-type stochastic volatility model. 
	They provide a closed-form expression for the GMMB rider and observe that such a structure might help realign fee incomes with the value of the financial guarantee. 
	\cite{kouritzin2018vix} extend the works of \cite{cui2017} by applying the VIX fee designs to a GMWB rider and by adding jumps to the underlying index value process. 
	\cite{bernard2016variable} study fees that are tied to the volatility index by using a Gaussian copula.\\
	\\	
	 In this work, we allow fee structures to be as general as possible, i.e. the fee structure may depend on the time, the fund value, and also on the latent variance process, making it possible to link the fee to the VIX. 
	 In the constant fee case, it is well-known that the misalignment between the fees and the value of the financial guarantee creates an incentive for the risk-neutral, rational policyholder to surrender her policy early (see \cite{milevsky2001real}). 
	 Since VIX-linked fee structures allow for better alignment of the guaranteed value with the corresponding hedging cost, we expect that such fee designs can also help reduce the insurers' exposure to surrender risk.
	 For this reason, we numerically study the impact of three different VIX-linked fee designs on the optimal surrender strategy. 
	 To do so, we use a two-layer continuous-time Markov chain (CTMC) approximation for the fund dynamics inspired by \cite{cui2018general}.
	 Two-layer CTMC approximations have recently been used to price derivatives in stochastic volatility models, see \cite{cui2018general}, \cite{cui2019continuous} and \cite{MaCTMCAm}, among others. 
	 The methodology proposed by \cite{cui2018general} for approximating two-dimensional diffusions is not only theoretically appealing and applies to most stochastic volatility models, but also is simple to implement for pricing European and American options.  
	 Their approach is  especially efficient for a short/medium time horizon.
	 However, for derivatives with very long maturities, such as those involved in variable annuities pricing, 
	 the methodology proposed by \cite{cui2018general} stretches the computing resource to unacceptable levels.
	 In this paper, we adapt their method to long-maturity cases.\\
    \\
	The main contributions of this paper are as follows:
	\begin{itemize}
		\item 
		We extend the work of \cite{cui2018general}, done in the context of options pricing, by providing novel efficient algorithms to value options with very long maturities, such as variable annuities, under general stochastic volatility models. 
		\item 
		We propose a methodology to approximate the optimal surrender surface for a VA contract with GMMB rider when the underlying index follows a two-dimensional diffusion process.
		\item 
		We derive a closed-form analytical expression for the VIX index when the variance process follows a continuous-time Markov chain. The latter may also be used to price VIX derivatives.
		\item 
		We analyze the optimal surrender strategy for a VA contract with GMMB rider under the assumption that the guarantee fees are linked to the $\vix$ index. In the literature, such an analysis of surrender incentives is usually performed using a constant fee structure, or in a Black-Scholes framework when the fees are state-dependent.
	\end{itemize} 
	The remainder of the paper is organized as follows. In Section \ref{sectFinancialSett}, we introduce the market model, the VA contract, and the optimal stopping problem involved in the pricing of variable annuities with surrender. A brief introduction to CTMC approximations for a two-dimensional diffusion process is provided in Section \ref{sectionCTMCapprox}. In Section \ref{subsecVApricingCTMC}, we apply CTMC approximations to VA contract pricing and provide new efficient algorithms. Section \ref{sectNumEx} provides the numerical results and discusses how VIX-linked fees affect surrender incentives. Section \ref{conc} concludes the paper. 

	\section{Financial Setting}\label{sectFinancialSett}
	\subsection{Market Model}\label{sectionMarketModel}
	We consider a filtered probability space $(\Omega,\,\mathcal{F},\,\mathbb{F},\prob{Q})$, where $\mathbb F$ is a complete and right-continuous filtration and where 
	$\prob{Q}$ denotes the pricing measure for our market, see Remark \ref{rmkMGmeasureExistence}.
	We consider a risky asset, whose price can be described by the two-dimensional process $(S,V)=\{(S_t,V_t)\}_{t\geq 0}$ satisfying	\begin{equation}
		\begin{array}{ll}
			\diff S_t&=r S_t \diff t +\sigma_S(V_t) S_t \diff W^{(1)}_t,\\
			\diff V_t&=\mu_V(V_t)\diff t+\sigma_V(V_t) \diff W^{(2)}_t,\label{eqEDS_S_Q}
		\end{array}
	\end{equation}
	with $S_0=s_0\in\reals_+$ and $V_0=v_0\in\mathcal{S}_V$ where $\mathcal{S}_V$ denotes the state-space of $V$, with $r>0$ denoting the risk-free rate and with $W=\{(W^{(1)}_t, W^{(2)}_t)\}_{t\geq0}$ is a two-dimensional correlated Brownian motion with cross-variation $[W^{(1)},\,W^{(2)}]_t=\rho t$, where 
	$\rho\in[-1, \, 1]$. 
	For simplicity, $V$ will be referred to as the variance process. We assume that $\mu_V:\mathcal{S}_V\mapsto\reals$ is continuous and that $\sigma_S,\,\sigma_V:\mathcal{S}_V\mapsto\reals_+$ are continuously differentiable functions  with $\sigma_S(\cdot)> 0$, $\sigma_V(\cdot)> 0$ on the state-space $\mathcal{S}_V$ of $V$. 
	Further, we suppose that $\mu_V$, $\sigma_V$ and $\sigma_S$ are defined such that \eqref{eqEDS_S_Q} has a unique-in-law weak solution. 
	
	\begin{remark}\label{rmkMGmeasureExistence}In \eqref{eqEDS_S_Q}, we start directly with the dynamics under the risk-neutral measure,  
	hence the form of the market price of volatility risk is not necessary in our setting. However, as pointed out by \cite{sin1998complications}, \cite{jourdain2004loss} and \cite{cui2013martingale}, a risk-neutral measure may not always exist under stochastic volatility models; additional conditions must be added to the model parameters in order for $\{e^{-rt}S_t\}_{t\geq 0}$ to be a true martingale. 
			A list of common SV models are reported in Table \ref{tblExSVmodels} below, along with conditions for the martingale property to hold under the risk-neutral measure.
	\end{remark}

	\begin{table}[h!]
		\centering
	   \resizebox{\columnwidth}{!}{
		\begin{tabular}{c|c|c| c}
			\hline
			\textbf{Model Name} &\textbf{Dynamics} & \textbf{Parameters}& \textbf{Cond. for}\\
			                    &                   &                   &  \textbf{Martingale Measure}\\
		\hline
		\hline
		Heston (1993) & $\diff S_t=rS_t\diff t+\sqrt{V_t}S_t\diff W^{(1)}_t$ & $S_0>0$, &No additional cond. \\
		\cite{heston1993} & $\diff V_t=\kappa(\theta -V_t)\diff t+\sigma\sqrt{V_t}\diff W^{(2)}_t$  & $\kappa,\theta,\sigma, V_0>0$ & \cite{cui2013martingale}, Proposition 2.5.1\\
		\hline
		$3/2$ (1997) & $\diff S_t=rS_t\diff t+S_t/\sqrt{V_t}\diff W^{(1)}_t$ & $S_0>0$, &$\rho\leq 0$, \\
		\cite{Heston1997ASN} & $\diff V_t=\kappa(\theta -V_t)\diff t-\sigma\sqrt{V_t}\diff W^{(2)}_t$  & $\kappa,\theta,\sigma, V_0>0$ &\cite{cui2013martingale}, Proposition 2.5.4\tablefootnote{The condition stated in \cite{cui2013martingale}, Proposition 2.5.4 is automatically satisfied by requiring the correlation parameter $\rho$ to be non-positive, as pointed out by \cite{drimus2012options}.\label{footnoteMGmeasure32}} \\
		            &     & with $\kappa\theta\geq\sigma^2/2$ & \\
		\hline
		$4/2$ (2017) & $\diff S_t=rS_t\diff t+S_t\left[a\sqrt{V_t}+b/\sqrt{V_t}\right]\diff W^{(1)}_t$ &  $a,b\in\reals$, $S_0>0$, &$ \sigma^2\leq 2\kappa\theta+\min(0,2\rho\sigma b)$\\
		\cite{grasselli20174} & $\diff V_t=\kappa(\theta -V_t)\diff t+\sigma\sqrt{V_t}\diff W^{(2)}_t$  & $\kappa,\theta,\sigma, V_0>0$, & \cite{grasselli20174}, Section 2.2
		\\
		 & & with $\kappa\theta\geq\sigma^2/2$ \tablefootnote{If $b\neq 0$, then we must also impose $\kappa\theta\geq \sigma^2/2$ for the model to be well-defined.} & 
		 \\
		\hline
		Hull-White (1987) & $\diff S_t=rS_t\diff t+\sqrt{V_t} S_t\diff W^{(1)}_t$ & $S_0>0$, & $\rho\leq 0$,\\
		\cite{hull1987pricing} & $\diff V_t=\alpha V_t\diff t+\beta V_t\diff W^{(2)}_t$  & $\alpha$, $\beta,V_0>0$ &  \cite{jourdain2004loss}, Theorem 1 or \\
		& & & \cite{cui2013martingale}, Proposition 2.5.10\\
		\hline
	    Scott (1987) & $\diff S_t=rS_t\diff t+e^{V_t} S_t\diff W^{(1)}_t$ & $S_0>0$, & $\rho\leq 0$, \\
		\cite{scott1987option}, p.426 & $\diff V_t=\kappa(\theta -V_t)\diff t+\sigma\diff W^{(2)}_t$  & $\kappa,\theta,V_0 \in\reals$, $\sigma>0$ & \cite{jourdain2004loss}, Theorem 1 \\
		\hline
	     $\alpha$-Hypergeometric (2016) & $\diff S_t=rS_t+e^{V_t} S_t\diff W^{(1)}_t$ & $S_0>0$, &  If $\alpha\geq 2$ or $\alpha<2$ and either\\
		\cite{da2016alpha} & $\diff V_t=(a-b e^{\alpha V_t})\diff t+\sigma\diff W^{(2)}_t$  & $\alpha, b,\sigma>0$, $a, V_0\in\reals$ & $\rho\leq 0$, $\alpha >1$ or $\alpha =1$ and $b\geq \rho\sigma$\\
		& & & \cite{da2016alpha}, Proposition 6\\
		\hline
		\end{tabular}
	     }
		\caption{Examples of stochastic volatility models}
		\label{tblExSVmodels}
    	\end{table}

	\subsection{Variable Annuity Contract}\label{subsectionVA}
	A policyholder enters a variable annuity contract by depositing an initial premium $F_0$ into a sub-account, which is then invested in a fund tracking the financial market.
	For simplicity, we will assume that the sub-account is invested in the risky asset $S$.
	The policyholder often has the right to surrender the contract, or lapse, prior to maturity.
	This additional flexibility is often called surrender option (or surrender right)  
	 in the literature and significantly complicates the valuation of variable annuity contracts. 
	 Below we discuss the risk-neutral valuation approach for a variable annuity, under both the assumption that the policyholder makes use of her surrender right or does not.\\
	\\
	To do so, we consider a finite time horizon $T\in\reals_+$ and let $F=\left\lbrace F_t \right\rbrace_{0\leq t\leq T}$ denote the variable annuity fund (or sub-account) process. Moreover, we let $C:[0,T]\times\reals_+\times\mathcal{S}_V\rightarrow \reals_+$ denote the fee function and let the continuously compounded fee rate process $\{c_t\}_{0\leq t\leq T}$ be defined as
	\begin{equation}\label{eqFeeFct}
	c_t:=C(t, F_t, V_t),\quad 0\leq t\leq T,
    \end{equation}
	where $C$ is assumed to be continuous or bounded and such that \eqref{eqEDS_F_Q} has a unique-in-law weak solution. 
	We allow the fee structure to be as general as possible. This setting includes, among others, state-dependent fee structures (see \cite{Bernard2014}, \cite{delong2014pricing}, \cite{Mackay2017}), VIX-linked fee structures (see \cite{cui2017}, \cite{kouritzin2018vix}, \cite{bernard2016variable}), and time-dependent fee structures (see \cite{bernard2019less}).\\
	\\
	We assume that the fees are paid continuously out of the fund at a rate $c_t$, so that the fund value is given by
	\begin{eqnarray}
		F_t&=&S_te^{-\int_0^t c_u\diff u}, \quad 0\leq t\leq T, \label{eqF}
	\end{eqnarray}
	with $F_0=S_0$. Using Itô's lemma, the dynamics of $F$ under the risk-neutral measure are 
	\begin{equation}
		\begin{array}{ll}
			\diff F_t  &= (r-c_t) F_t \diff t+\sigma_S(V_t) F_t \diff W_t^{(1)},\\
			\diff V_t&=\mu_V(V_t)\diff t+\sigma_V(V_t) \diff W^{(2)}_t.\label{eqEDS_F_Q}
		\end{array}
	\end{equation}
	Throughout this paper, $\e_{t,x,y}[\cdot]$ is short-hand notation for $\e[\cdot|F_t=x,V_t=y]$ and $\e_t[\cdot]$ for $\e[\cdot|\mathcal{F}_t]$, with $x\in\reals_+$, $y\in\mathcal{S}_V$ and $t\in[0,T]$. We also use $\e_{x,y}[\cdot]$ to denote $\e_{0,x,y}[\cdot]$.\\
	\\
	We focus on a variable annuity with a guaranteed minimum maturity benefit (GMMB) whose payoff at maturity $T$ is $\max(G, F_T)$, where $G\in\reals_+$ is a predetermined guaranteed amount. Given $(F_t,V_t)=(x,y)$, the time-$t$ risk-neutral value of the variable annuity assuming that it will not be surrendered early is
	\begin{equation}\label{eqVAnoSurrender}
	v_e(t,x,y):=\e_{t,x,y}\left[e^{-r(T-t)}\max(G,F_T)\right].
	\end{equation}
	On early surrender, the policyholder receives the value of the VA sub-account, reduced by a penalty which, in our setting, can depend on time and on the value of the variance process $V$.
	When no surrender occurs, the maturity benefit is paid at $T$.

	More formally, the VA reward (or gain) function $\varphi:[0,\,T]\times \reals_+\times\mathcal{S}_V\rightarrow \reals_+$ is defined by
	\begin{equation}
		\varphi(t,x,y)=\begin{cases}
			g(t,y)\,x & \text{if $t<T$,}\\
			\max(G,\,x) & \text{if $t=T$},
		\end{cases}\label{eqRewardFct}
	\end{equation}
	where $g:[0,\,T]\times\mathcal{S}_V\rightarrow [0, 1]$ is continuous, non-decreasing in time and satisfies ${\lim_{t \rightarrow T^-} g(t,y) = 1}$ $\forall y \in \mathbb \mathcal{S}_V$. In practice, we usually consider the surrender charge (as a percentage of the account value), $1-g(\cdot,\cdot)$. 
	 A common form for the surrender charge function in the literature is $g(t,y)= e^{-k(T-t)}$ for some constant $k\geq0$,  see 
	for example \cite{Mackay2017}, \cite{shen2016valuation}, \cite{bacinello2019variable} and \cite{kang2018optimal}. It is the first time, to the best of our knowledge, that variance-dependent surrender charges are considered. 
	
	\begin{remark}\label{rmkVarphiDiscoun}
	For $x<G$, the function $t\mapsto\varphi(t,x,y)$ is discontinuous at $T$ since 
	\begin{align*}
	    \lim_{t \rightarrow T^-} \varphi(t,x,y) = g(t,y)x \leq x < G = \varphi(T,x,y).
	\end{align*}
	\end{remark}
	Under the assumption that the policyholder maximizes the risk-neutral value of her VA contract, the time-$t$ value of the variable annuity policy is given by
	\begin{eqnarray}
		v(t,x,y)= \sup_{\tau \in \mathcal{T}_{t,\, T}} \e_{t,x,y}[e^{-r(\tau -t)}\varphi(\tau, F_{\tau},  V_{\tau})],\label{eqAmOptVA}
	\end{eqnarray}
	where $\mathcal{T}_{t,\, T}$ is the (admissible) set of all stopping times taking values in the interval $[t,\,T]$.

	When the fee function is time-independent, the fund process is time-homogeneous, and \eqref{eqAmOptVA} is equivalent to
	\begin{eqnarray}
		v(t,x,y)= \sup_{\tau \in \mathcal{T}_{0,\, T-t}} \e_{x,y}[e^{-r\tau}\varphi(\tau+t, F_{\tau}, V_{\tau})].\label{eqAmOptVA2}
	\end{eqnarray}
	\newline
	Similarly to the early exercise premium in the American option literature, the value of the right to surrender, denoted by $e:[0,T]\times\reals_+\times \mathcal{S}_V\rightarrow\reals_+$,
	is defined by 
	$$e(t,x,y):=v(t,x,y)- v_e(t,x,y).$$
	\vspace{0.2cm}
	
	\section{Continuous-Time Markov Chain Approximation}\label{sectionCTMCapprox}
   The CTMC framework outlined in this section has been proposed by \cite{cui2018general} for exotic option pricing under stochastic local volatility models. 
   The general idea is to approximate the two-dimensional stock price process by a two dimensional continuous-time Markov chain.
   This is done by first approximating the variance process by a CTMC, and then by replacing the variance process by its CTMC approximation in the underlying price process. 
   The resulting regime-switching diffusion process is then further approximated by a CTMC, yielding a two-dimensional CTMC process which converges weakly to the original two-dimensional diffusion process, providing that the generator of the CTMC is chosen correctly.


   First, we shall briefly recall the basics of continuous-time Markov chains, following sections 6.9 and 6.10 of \cite{grimmet2001probability}. 
   The  stochastic process $\tilde{X}=\{\tilde{X}_t\}_{t\geq 0}$ taking values on some discrete state-space $\mathcal{S}_{\tilde{X}}$ is called a continuous-time Markov chain if it satisfies the following property (a.k.a. \textit{Markov property}):
   		\begin{equation*}
   			\prob{P}(\tilde{X}_{t_n}=\tilde{x}_j|\tilde{X}_{t_1}=\tilde{x}_{i_1},\,\ldots,\,\tilde{X}_{t_{n-1}}=\tilde{x}_{i_{n-1}})=\prob{P}(\tilde{X}_{t_n}=\tilde{x}_j|\tilde{X}_{t_{n-1}}=\tilde{x}_{i_{n-1}})	
   		\end{equation*}
   		for all $\tilde{x}_j,\,\tilde{x}_{i_1},\,\ldots,\,\tilde{x}_{i_{n-1}}\in\mathcal{S}_{\tilde{X}}$ and any time sequence $t_1<t_2<\ldots<t_n$.\\
   \\
   The \textit{transition probability} $p_{ij}(s,\,t)$, from state $\tilde{x}_i\in \mathcal{S}_{\tilde{X}}$ at time $s$ to state $\tilde{x}_j\in \mathcal{S}_{{\tilde{X}}}$ at time $t$, is defined by
   	$$p_{ij}(s,\,t)=\prob{P}({\tilde{X}}_t=\tilde{x}_j|{\tilde{X}}_s=\tilde{x}_i),\quad s\leq t.$$
   	The chain is said to be \textit{homogeneous} if $p_{ij}(s,\, t)=p_{ij}(0,\,t-s)$ for any $i,j,\, s\leq t$. In that case, we use $p_{ij}(t-s)$ to denote $p_{ij}(s,\, t).$\\
   	\\
   	Now assume that $\tilde{X}$ is time-homogeneous and $\mathcal{S}_{\tilde{X}}$ is finite.\\
   	\\
   	 The family $\{\mathbf{P}_t:=[p_{ij}(t)]_{|\mathcal{S}_{\tilde{X}}|\times|\mathcal{S}_{\tilde{X}}|}\}_{t\geq 0}$ of \textit{transition probability matrices} is referred as the \textit{transition semigroup} of the Markov chain.\\
   	 \\
   	 For an infinitesimal period of length $h>0$, it can be shown that there exist constants $\{q_{ij}\}$, also called \textit{transition rates}, such that 
   	 \begin{equation}
   	 	p_{ij}(h) =\left\{\begin{array}{ll}
   	 		q_{ij}h+c(h)& \textrm{ if } i\neq j\\
   	 		1+q_{ij}h+c(h) & \textrm{ if } i=j,
   	 	\end{array}\right.\label{eq_pijh}
   	 \end{equation}
   	 for $1\leq i,j\leq |\mathcal{S}_{\tilde{X}}|$, where $c$ is a function satisfying $\lim_{h\rightarrow 0}\frac{c(h)}{h}=0$.\\
   	 \\
   	 From the above, we can conclude that the transition rates must satisfy
   	 \begin{equation}
   	 	q_{ij}\geq 0, \textrm{ if } i\neq j\quad\textrm{and}\quad q_{ij}\leq 0, \textrm{ if } i=j,\label{eqQproperty1}
   	 \end{equation}
   	 and 
   	 \begin{equation}
   	 	\sum_{j=1}^{m}  q_{ij}=0,\quad i=1,2,\ldots,m.\label{eqQproperty2}
   	 \end{equation}

   	 The matrix $\mathbf{Q}:=[q_{ij}]_{|\mathcal{S}_{\tilde{X}}|\times|\mathcal{S}_{\tilde{X}}|}$ is called the \textit{generator} of ${\tilde{X}}$. 
   	 Under some technical conditions\footnote{More precisely, the semigroup $\{P_t\}$ must be standard---that is, $p_{ii}(t)\rightarrow 1$ and $p_{ij}(t)\rightarrow 0$ as $t\downarrow 0$---and uniform---$\sup_i -q_{ii}<\infty$, see \cite{grimmet2001probability}, Definition 6.9.4, Theorem 6.10.1 and 6.10.6 for details.}, it can be shown that the transition probability matrix $\mathbf{P}_t$ has the following matrix exponential representation:
   	 \begin{equation}
   	 	\mathbf{P}_t=\exp(\mathbf{Q} t)=\sum_{k=0}^\infty \frac{(\mathbf{Q} t)^k}{k!}.\label{Chap2EqP}
   	 \end{equation}
   	 
 \begin{assump}\label{assumpFeeFct}
 	The fee function defined in \eqref{eqFeeFct} is time-independent and denoted by $c$. That is, $C(t,x,y)=c(x,y)$ for all $0\leq t\leq T$. Moreover, we only consider functions $c$ that are continuous or bounded.
 \end{assump}
 Henceforth, we consider that Assumption \ref{assumpFeeFct} holds.
 That is, we assume that the fee function is time-independent so that the fund process is time-homogeneous. For the CTMC approximation of diffusion processes with time-dependent coefficients, see \cite{ding2021markov}.

   \subsection{Approximation of the Variance Process $\{V_t\}_{t\geq0}$}\label{subSectCTMCvariance}
   We construct a CTMC $\{V_t^{(m)}\}_{t\geq 0}$ taking values on a finite state-space $\mathcal{S}_V^{(m)}:=\{v_1,v_2,\dots\, v_m\}$, with $v_i\in\mathcal{S}_V$ and $m\in\mathbb{N}$, that converges weakly to $\{V_t\}_{t\geq 0}$ as $m\rightarrow\infty$. Weak convergence of $V^{(m)}$ to $V$, is denoted by $V^{(m)}\Rightarrow V$. 
   
   Several approaches are available in the literature to construct the finite state-space $\mathcal{S}_V^{(m)}$, from simple uniform to non-uniform grids (see \cite{tavellapricing} \cite{mijatovic2013continuously}, \cite{lo2014improved} for examples of non-uniform grids). 
   Most grid constructions require the specification of the two boundary states $v_1$, $v_m\in\mathcal{S}_V$. The lower (resp. upper) boundary $v_1$ (resp. $v_m$) must be sufficiently small (resp. large) to ensure that the grid covers a reasonably large portion of the state-space of the original process. This can be done naively by setting $v_1=\alpha V_0$ and $v_m=\beta V_0$ for a small (resp. large) constant $\alpha$ (resp. $\beta$), or by using the first two moments of the variance process to select the boundary values (see \cite{kirkby2017unified},\cite{cui2017equity}, \cite{cui2017general} and \cite{leitao2019ctmc} for details.)
   The specific grid selected for the numerical analysis performed in this paper is discussed in more details in Section \ref{sectNumEx}.

Once the state-space is chosen, the approximating CTMC $\{V_t^{(m)}\}_{t\geq 0}$ is defined via its generator $\mathbf{Q}^{(m)}=[q_{ij}]_{m\times m}$.
This generator is constructed so that the first two moments of the transition density of the variance process $\{V_t\}_{t\geq 0}$ and the approximating CTMC $\{V_t^{(m)}\}_{t\geq 0}$ coincide; these are the so-called \textit{local consistency conditions}, see \cite{kushner1990numerical} and \cite{lo2014improved}.
More precisely, the elements $q_{ij}$, $1\leq i,j\leq m$ of the generator of $V^{(m)}$ are chosen so that for a small time increment $h << T$, 
 \begin{align}
 \begin{split}
 	&\e_t\left[V^{(m)}_{t+h}-V_t^{(m)}\right]=\e_t\left[V_{t+h}-V_t\right]  \simeq\mu_V(V_t)h, \qquad\text{and}\\
 	&\e_t\left[\left(V^{(m)}_{t+h}-V_t^{(m)}\right)^2\right]=\e_t\left[\left(V_{t+h}-V_t\right)^2\right]\simeq\sigma_V^2(V_t)h,
\end{split}
 \label{eqLocalConstComd}
 \end{align}
 for all $t\geq 0$.
To ensure that the local consistency conditions are satisfied, we use the generator proposed by \cite{lo2014improved} and given by\footnote{An advised reader will notice some difference between the transition rates stated above, and the ones that appear in \cite{lo2014improved}. However, one can show that the two rate matrices are equivalent with some simple algebra.}  
 \begin{equation}
 	q_{ij}=\begin{cases}
 		\frac{\sigma_V^2(v_i)-\delta_i\mu_V(v_i)}{\delta_{i-1}(\delta_{i-1}+\delta_i)} & \text{$j=i-1$}\\
 		-q_{i,i-1}-q_{i,i+1}  & \text{$j=i$}\\
 		\frac{\sigma_V^2(v_i)+\delta_{i-1}\mu_V(v_i)}{\delta_{i}(\delta_{i-1}+\delta_i)} & \text{$j=i+1$}\\
 		0&\text{$j\neq i,\,i-1,\,i+1$},\label{eqQ}
 	\end{cases}
 \end{equation}
    for $2\leq i\leq m-1$ and $1\leq j\leq m$ and
   where $\delta_i=v_{i+1}-v_i$, $i=1,2,\ldots, m-1$. On the borders, we set $q_{12}=\frac{|\mu_V(v_1)|}{\delta_1}$, $q_{11}=-q_{12}$, $q_{m,m-1}=\frac{|\mu_V(v_m)|}{\delta_{m-1}}$, $q_{m,m}=-q_{m,m-1}$; and $0$ elsewhere.\\
   \\
  To obtain a well-defined $\mathbf{Q}^{(m)}$ matrix, the transition rates in \eqref{eqQ} must also satisfy the conditions in \eqref{eqQproperty1}. Hence, for $2\leq i\leq m-1$,
\begin{itemize}
 \item if $\mu_V(v_i)<0$, then we must have
\begin{equation}
  \delta_{i-1}\leq\frac{\sigma^2_V(v_i)}{|\mu_V(v_i)|}\quad\textrm{},\label{eqQprop1}
 \end{equation}
 in order for $q_{i,i+1}$ to be well defined (that is, $q_{i,i+1}\geq 0$), and 
 \item if $\mu_V(v_i)>0$, then we must have
 \begin{equation}
 	 \delta_{i}\leq\frac{\sigma^2_V(v_i)}{\mu_V(v_i)},\label{eqQprop2}
 \end{equation}
in order for $q_{i,i-1}\geq 0$.
\item when $\mu_V(v_i)=0$ then no additional condition needs to be added.
\end{itemize}
\begin{remark}\label{rmkQprop}
	A sufficient condition for \eqref{eqQprop1} and \eqref{eqQprop2} to hold is 
	\begin{equation}
		\max_{1\leq i\leq m-1} \delta_i\leq \min_{1\leq i\leq m-1} \frac{\sigma^2_V(v_i)}{|\mu_V(v_i)|} \label{eqQprop3}.
	\end{equation}
	 When this condition is not satisfied, more points should be added to the existing grid. However, it can sometimes be too restrictive, particularly in the case of approximating a two-dimensional process,  where adding more points to the state-space becomes too expensive computationally. In that case, condition \eqref{eqQprop3} may be replaced by a simple verification of \eqref{eqQprop1} and  \eqref{eqQprop2} which is less restrictive. However, from a numerical perspective, we observe that such conditions are not necessary to obtain good approximation results.
\end{remark}
The transition rates on the boundaries of the state-space are set so that the absolute instantaneous means are maintained at the endpoints. Other schemes could have also been employed (see \cite{chourdakis2004non},\cite{mijatovic2013continuously}), but we observed that all of these schemes are equivalent numerically. 

 \subsection{Approximation of the Fund Value Process $\{F_t\}_{0\leq t\leq T}$}\label{Chap2sectionCTMC_RS}
 The CTMC approximating $\{F_t\}_{t\geq0}$ is constructed by first replacing the variance process appearing in the drift and diffusion coefficients by their CTMC approximations, and then by further approximating the resulting regime-switching diffusion process by another CTMC.
 The resulting two-dimensional regime-switching CTMC can then be mapped to a one-dimensional CTMC on an enlarged state-space.

Lemma \ref{lemmaDecoupleBM} below allows for the  removal of the correlation between the Brownian motions in \eqref{eqEDS_F_Q}, which is necessary to construct the CTMC approximation of $\{F_t\}_{t\geq0}$.

	\begin{lemma}[Lemma 1 of \cite{cui2018general}]\label{lemmaDecoupleBM}
	\begin{sloppypar}
		Let $F$ and $V$ be defined as in \eqref{eqEDS_F_Q}. Define ${\gamma(x):=\int_{\cdot}^x \frac{\sigma_S(u)}{\sigma_V(u)} \diff u}$ and $X_t:=\ln (F_t) -\rho \gamma(V_t)$, $t\in[0,T]$. Then $X$ satisfies
	\end{sloppypar}
		\begin{equation}\label{eqEDS_X_Q}
			\begin{split}
				\diff X_t=&\mu_X(X_t,V_t)\diff t+\sigma_X(V_t)\diff W^*_t\\
				\diff V_t=&\mu_V(V_t) \diff t+\sigma_V(V_t) \diff W^{(2)}_t,
			\end{split}
		\end{equation}
		where $W^*_t$ is a standard Brownian motion $W^*_t:=\frac{W^{(1)}_t-\rho W^{(2)}_t}{\sqrt{1-\rho^2}}$ independent of $W^{(2)}_t$, $\sigma_X(y):=\sqrt{1-\rho^2}\sigma_S(y)$ and $\mu_X(x,y):=r-c(e^{x+\rho \gamma(y)},y)-\frac{\sigma^2_S(y)}{2}-\rho\psi(y)$, and
		\begin{eqnarray*}
			\psi(y)&:=&\mathcal{L}_v \gamma(y)=\mu_V(y)\gamma'(y)+\frac{1}{2}\sigma_V^2(y)\gamma''(y)\\
			&=& \mu_V(y)\frac{\sigma_S(y)}{\sigma_V(y)}+\frac{1}{2}\left[\sigma_V(y)\sigma_S'(y)-\sigma_V'(y)\sigma_S(y)\right],
		\end{eqnarray*}
	for $x\in\reals,y\in\mathcal{S}_V$.
	\end{lemma}
	The proof relies on the multidimensional Itô formula (see Lemma 1 of \cite{cui2018general} for details).
	
	Given the CTMC approximation of the process $V^{(m)}$ and its generator $\mathbf{Q}^{(m)}$, the diffusion process in \eqref{eqEDS_X_Q} can now be approximated by a regime-switching diffusion process $\{X_t^{(m)}\}_{t\geq 0}$:
	\begin{equation}\label{eqEDS_Xm}
			\diff X^{(m)}_t  = \mu_X(X_t^{(m)},V^{(m)}_t) \diff t+\sigma_X(V^{(m)}_t)\diff W^{*}_t,
	\end{equation} 
	where regimes are determined by the states of the approximated variance process, $\{v_1,v_2,\ldots,v_m\}$. 
To construct a regime-switching CTMC $(X_t^{(m,N)}, V_t^{(m)})$ approximating the regime-switching diffusion $(X_t^{(m)}, V_t^{(m)})$,
we fix a state for the variance process $V_t^{(m)}$ (or equivalently a regime) and  construct a CTMC approximation for $X_t^{(m)}$ given that $V_t^{(m)}$ is in that state.
This is done using the procedure described in Section \ref{subSectCTMCvariance} for a one-dimensional diffusion process.
The procedure is then repeated for each state in $\mathcal{S}_V^{(m)}$, and the approximating CTMCs are combined with $V^{(m)}$ to obtain the final regime-switching CTMC.
\\
More precisely, let $X_t^{(m,N)}$ be the CTMC approximation of $X_t^{(m)}$ taking values on a finite state-space $\mathcal{S}_X^{(N)}=\{x_1,x_2,\ldots,x_N\}$, $N\in\mathbb{N}$. 
For each $v_l\in\mathcal{S}_V^{(m)}$, we define the generator $\mathbf{G}^{(N)}_l=[\lambda_{ij}^l]_{N\times N}$ of $X_t^{(m,N)}$ given that the variance process is in state $v_l$  at time $t\geq 0$ by
\begin{equation}
	\lambda_{ij}^l=\begin{cases}
		\frac{\sigma_X^2(v_l)-\delta_i^x\mu_X(x_i,v_l)}{\delta_{i-1}^x(\delta_{i-1}^x+\delta_i^x)} & \text{$j=i-1$}\\
		-\lambda_{i,i-1}^l-\lambda_{i,i+1}^l  & \text{$j=i$}\\
		\frac{\sigma_X^2(v_l)+\delta_{i-1}^x\mu_X(x_i,v_l)}{\delta_{i}^x(\delta_{i-1}^x+\delta_i^x)} & \text{$j=i+1$}\\
		0&\text{$j\neq i,\,i-1,\,i+1$},\label{eqQk}
	\end{cases}
\end{equation}
for $2\leq i\leq N-1$ and $1\leq j\leq N$, where $\delta_i^x=x_{i+1}-x_i$, $i=1,2,\ldots, N-1$. On the boundaries, we set $\lambda_{12}^l=\frac{|\mu_X(x_1,v_l)|}{\delta^x_1}$, $\lambda_{11}^l=-\lambda_
{12}^l$, $\lambda_{N,N-1}^l=\frac{|\mu_X(x_N,v_l)|}{\delta^x_{N-1}}$, $\lambda_{N,N}^l=-\lambda_{N,N-1}^l$, and $0$ elsewhere.\\
\\
Using $V^{(m)}$ and the relation presented in Lemma \ref{lemmaDecoupleBM}, the approximated fund process $F^{(m,N)}$, which approximates $F$, is defined by
\begin{equation}
	F_t^{(m,N)}:=\exp\left\{X_t^{(m,N)}+\rho \gamma(V_t^{(m)})\right\},\quad{0\leq t\leq T}.\label{eqFctmcApprox}
\end{equation}
\begin{remark}[Convergence of the approximation]\label{remakWeakConvergenceCTMC}
Such a construction of the regime-switching CTMC ensures that the two-dimensional process $(X_t^{(m,N)}, V_t^{(m)})$ converges weakly to $(X_t, V_t)$ as $m,N\rightarrow\infty$. 
The main idea is to show that the generator of $(X_t^{(m,N)}, V_t^{(m)})$ is uniformly close to the infinitesimal generator of  $(X_t, V_t)$ as $m,N\rightarrow\infty$, to then conclude that $(X_t^{(m,N)}, V_t^{(m)})\Rightarrow (X_t, V_t)$ using the results of \cite{ethier2005markov} which relies on semi-group theory. 
Moreover, since the function $h:\reals\times\mathcal{S}_V\rightarrow \reals_+$ defined by $h(x,y)=e^{x+\rho \gamma(y)}$ is continuous, we have that $F^{(m,N)}\Rightarrow F$ by the continuous mapping Theorem, see \cite{billingsley1999convergence} Theorem 2.7. 
For one-dimensional processes, intuition and detailed explanations of the proof can be found in \cite{mijatovic2013continuously}, Section 5 (or in the unabridged version of the paper \cite{mijatovic2009continuously}, Section 6); for stochastic volatility models,  see \cite{cui2018general}, Section 2.4.
\end{remark}

\begin{sloppypar}
The last step is to convert the regime-switching CTMC $(X_t^{(m,N)}, V_t^{(m)})$ into a one-dimensional CTMC process $Y_t^{(m,N)}$ on an enlarged state-space ${\mathcal{S}_Y^{(m,N)}:=\{1,2,\dots,mN\}}$. This is done in Theorem 1 of \cite{cai2019unified}, reproduced below. 
\end{sloppypar}
\begin{proposition}[Theorem 1 of \cite{cai2019unified}]\label{propCTMC_XtoY}
Consider a regime-switching CTMC $(X^{(m,N)}, V^{(m)})$ taking values in $\mathcal{S}_X^{(N)}\times\mathcal{S}_V^{(m)}$ where $\mathcal{S}_X^{(N)}=\{x_1,x_2,\ldots,x_N\}$ and $\mathcal{S}_V^{(m)}=\{v_1,v_2,\ldots,v_m\}$; and another one-dimensional CTMC, $\{Y_t^{(m,N)}\}_{0\leq t\leq T}$, taking values in $\mathcal{S}_Y^{(m,N)}:=\{1,2,\ldots, mN\}$ and its transition rate matrix $\mathbf{G}^{(m,N)}$ given by
\begin{equation}
	\left(\begin{array}{cccc}
		q_{11}\mathbf{I}_N+\mathbf{G}_1^{(N)} & q_{12}\mathbf{I}_N & \cdots & q_{1m}\mathbf{I}_N\\
		q_{21}\mathbf{I}_N & q_{22}\mathbf{I}_N+\mathbf{G}_2^{(N)} & \cdots & q_{2m}\mathbf{I}_N\\
		\vdots & \vdots & \ddots & \vdots\\
		q_{m1}\mathbf{I}_N & q_{m2}\mathbf{I}_N & \cdots & q_{mm}\mathbf{I}_N+\mathbf{G}_m^{(N)}\\
	\end{array}\right),\label{Chap2eqQ_Y}
\end{equation}
where $\mathbf{I}_N$ is the $N\times N$ identity matrix, $\mathbf{G}_l^{(N)}=[\lambda_{ij}^l]_{N\times N}$, $l=1,2,\ldots, m$ and $\mathbf{Q}^{(m)}=[q_{ij}]_{m\times m}$ are the generators defined in \eqref{eqQk} and \eqref{eqQ}, respectively. Define the function $\psi:\mathcal{S}_X^{(N)}\times\mathcal{S}_V^{(m)}\rightarrow \mathcal{S}_Y^{(m,N)}$ by $\psi(x_n,v_l)=(l-1)N+n$ and its inverse $\psi^{-1}:\mathcal{S}_Y^{(m,N)}\mapsto \mathcal{S}_X^{(N)}\times \mathcal{S}_V^{(m)}$ by $\psi^{-1}(n_y)= (x_n,v_l)$ for $n_y\in\mathcal{S}_Y^{(m,N)}$, where $n\leq N$ is the unique integer such that $n_y=(l-1)N+n$ for some $l\in\{1,2,\ldots,m\}$. Then, we have
$$
\e\left[\Psi(X^{(m,N)}, V^{(m)})|X_0^{(m,N)}=x_i,V_0^{(m)}=v_k\right]=\e\left[\Psi(\psi^{-1}(Y^{(m,N)}))|Y_0=(k-1)N+i\right],
$$
for any path-dependent function $\Psi$ such that the expectation on the left-hand side is finite.
\end{proposition}

\section{Variable Annuity Pricing via CTMC Approximation}\label{subsecVApricingCTMC}
In this section, we use the CTMC approximation of the fund value process to price variable annuities under different surrender strategies. 
We provide a simple way to approximate the optimal surrender surface, which is the extension in three dimensions of the exercise boundary for two-dimensional processes. The present section is based on \cite{cui2018general} and \cite{cui2019continuous}, in which CTMC approximation are used for option pricing.\\
\\
Recall that  $(X^{(m,N)}, V^{(m)})$ is the regime-switching CTMC approximation of $(X,V)$ (see Lemma \ref{lemmaDecoupleBM}) taking values in a finite state-space $\mathcal{S}_X^{(N)}\times\mathcal{S}_V^{(m)}$ where $\mathcal{S}_X^{(N)}=\{x_1,x_2,\ldots,x_N\}$, $N\in\mathbb{N}$; and $\mathcal{S}_V^{(m)}=\{v_1,v_2,\ldots,v_m\}$, $m\in\mathbb{N}$. We have also defined $F^{(m,N)}$ in terms of $(X^{(m,N)}, V^{(m)})$ in \eqref{eqFctmcApprox}.\\
\\
Throughout this section, we denote by $\{\mathbf{e}_{ik}\}_{i,k=1}^{N,m}$ the standard basis in $\reals^{mN}$, i.e. $\mathbf{e}_{ik}$ represents a row vector of size $1\times mN$ having a value of $1$ in the $(k-1)N+i$-th entry and $0$ elsewhere.

\subsection{Variable Annuity without Early Surrenders}

Consider an initial premium $F_0>0$ and let $V_0\in\mathcal{S}_V$. 
The risk-neutral value of a variable annuity contract assuming no early surrenders can be approximated by

\begin{equation}
	\begin{array}{ll}
v_e(0,F_0,V_0)&=\e\left[e^{-rT}\max(G,F_T)|F_0,V_0\right]\\
\quad&\approx \e\left[e^{-rT}\max(G,F_T^{(m,N)})|X_0^{(m,N)}=x_i,V_0^{(m)}=v_k\right].
\end{array}\label{Chap2eqVActmc}
\end{equation}

Here, we assume\footnote{If $X_0^{(m,N)}$ and $V_0^{(m)}$ are not part of their respective grids, then the two points can be added to the grids, or the option price must be linearly interpolated between grid points, see Remark \ref{remarkGripPoint} for details.} that $x_i\in\mathcal{S}_X^{(N)}$ and $v_k\in\mathcal{S}_V^{(m)}$, with $x_i=\ln(F_0)-\rho \gamma(V_0)$.
\begin{proposition}\label{propVAwoSurrenders}
	 Let $F_0>0$ be the initial premium, with $X_0^{(m,N)}=\ln(F_0)-\rho \gamma(V_0)=x_i \in\mathcal{S}_X^{(N)}$ and $V_0=v_k\in\mathcal{S}_V$. The risk-neutral value at time $0$ of a variable annuity contract held until maturity $T$ and with guaranteed amount $G>0$ can be approximated by
	\begin{eqnarray}
		v_e^{(m,N)}(0,F_0,V_0)&:=&\e[e^{-rT}\max(G,F_T^{(m,N)})|F_0^{(m,N)}=F_0,V_0^{(m)}=V_0]\nonumber\\
		&= &e^{-rT}\mathbf{e}_{ik}\exp\{T\mathbf{G}^{(m,N)}\}\mathbf{H},\label{eq_CTMC_VA}
	\end{eqnarray}
	where $\mathbf{G}^{(m,N)}$ is defined in \eqref{Chap2eqQ_Y} and $\mathbf{H}$ is a column vector of size $mN\times 1$ whose $(l-1)N+n$-th entry is given by
	\begin{equation}
	 h_{(l-1)N+n}=\max\left(G, e^{x_n+\rho \gamma( v_l)}\right), \quad 1\leq l\leq m,\,1\leq n\leq N.\label{eqCTMChj}
	\end{equation}
\end{proposition}
\begin{sloppypar}
The proof is straightforward by noticing that \eqref{eq_CTMC_VA} is the matrix representation of the conditional expectation of a (function of a) discrete one-dimensional random variable whose conditional probability mass function is given by the transitional probabilities ${p_{(k-1)N+i,j} (T),\,1\leq j\leq mN}$, with $\mathbf{P}(T)=[p_{ij} (T)]_{mN\times mN}=\exp\{T\mathbf{G}^{(m,N)}\}$ as per \eqref{Chap2EqP}.
\end{sloppypar}
The algorithm to approximate the value of a variable annuity without early surrenders is straightforward from the last proposition. 

\begin{remark} [Convergence of VA prices without early surrenders]\label{remarkVAwoSurrendersConv}
	As per Remark \ref{remakWeakConvergenceCTMC}, we know that $\FmN \Rightarrow F$, as $m,N\rightarrow\infty$. From the continuous mapping theorem, c.f. Theorem 2.7 of \cite{billingsley1999convergence}, we also know that $\varphi(T,\FmN_T,\Vm_T)\Rightarrow \varphi(T,F_T,V_T)$; implying the convergence of the corresponding VA prices $\e[\varphi(T,\FmN_T,\Vm_T)]\rightarrow\e[\varphi(T,F_T,V_T)]$, Theorem 3.6 of \cite{billingsley1999convergence}. Convergence analysis of the two-layers CTMC approximation in the context of European option pricing is performed in \cite{ma2022convergence}.
\end{remark}
 
\subsubsection{Fast Algorithm}\label{subsectFastAlgo}
When we consider a typical time-horizon of $10$, $15$ or $20$ years as is often the case in VA pricing, the probability for the fund or the volatility processes to reach high (resp. low) value is higher than when shorter maturities are concerned, and so the grid's upper (resp. lower) bound must be set to a higher (resp. lower) value. 
This complicates the pricing of VAs compared to financial option pricing (generally written for short or medium time-horizon), since we need more discretization points $m$ and $N$ in order to capture the distribution of the variance and the fund value process correctly. 
Theoretically, this is not a problem; however several numerical issues can be encountered when computing the price. 
First, the generator matrix $\mathbf{G}^{(m,N)}$ can become very large and thus require a large amount of storage space, which may cause memory problems. 
Second, calculating the exponential of a large sparse $mN \times mN$ matrix over a long time horizon is time-consuming.\\
\\
Using the tower property of conditional expectations and an approximation based on the assumption that the variance process remains constant over small time periods, we propose a new algorithm that speeds up the pricing of the VA contract without early surrenders (see Algorithm \ref{algoVA_WOsurrendersFast}). 
The main idea behind the new algorithm resides in the use of nested conditional expectations. 
More precisely, for $h= T/M$, $M\in\mathbb{N}$, we use
$$\e\left[e^{-rT}\varphi\left(T, e^{\XmN_T+\rho \gamma(\Vm_T)}, \Vm_T\right)\right]=\e\left[\tilde{\varphi}\left(kh\right)\right],\quad k=1,2,\ldots,M-2$$

with $\tilde{\varphi}\left(kh\right)=\e\left[\tilde{\varphi}\left((k+1)h\right)\big|\mathcal{F}_{kh}\right]$ and
${\tilde{\varphi}\left((M-1)h\right)=\e\left[\varphi\left(T,e^{\XmN_{T}+\rho \gamma(\Vm_{T})}, \Vm_{T}\right)\big|\mathcal{F}_{(M-1)h}\right]}$.\\
\\
The approximation used in Algorithm \ref{algoVA_WOsurrendersFast} follows from Proposition \ref{propExpectionApprox} presented below.
\begin{proposition}\label{propExpectionApprox}
Let $h>0$ with $h\ll T$ and $0\leq t\leq T-h$. For any function $\phi$ such that the expectation on the left-hand side of \eqref{eqCondExpapprox} is finite, we have that
\begin{multline}
\e\left[\phi\left(t+h,\XmN_{t+h},\Vm_{t+h}\right)|\XmN_t=x_i,\Vm_t=v_k\right]\\
=\sum_{j=1}^m\e\left[\phi(t+h,\XmN_{t+h},\Vm_{t+h})|\Vm_t=\Vm_{t+h}=v_j,\XmN_t=x_i\right]\\
\times\prob{P}\left(\Vm_{t+h}=v_j |\Vm_{t}=v_{k}\right)+c(h),\label{eqCondExpapprox}
\end{multline}
where $c(h)$ is a function satisfying $\lim_{h\rightarrow 0}\frac{c(h)}{h}=0$.
\end{proposition}
The proof of Proposition \ref{propExpectionApprox} is reported in Appendix \ref{appendixProofPropCondExpApprox}.\\
\\
Using the last proposition allows us to work with the matrices $\{\mathbf{G}^{(N)}_j\}_{j=1}^m$ and $\mathbf{Q}^{(m)}$ separately, so that the new algorithm now requires $m$ times the calculation of the exponential of a $N\times N$ matrix and one time the exponential of a $m\times m$ matrix over small time-intervals. Hence, by reducing the size of the matrix in the matrix exponential and the length of the time interval over which the exponential is calculated allows to reduce the computation time significantly and to manage the memory space more effectively. Also, note that the added cost of computing $m+1$ matrix exponentials (rather than only one) is counterbalanced by the reduced cost of computing the exponential of smaller matrices over a short time interval $h$.\\
\\
The following notation is used in Algorithm \ref{algoVA_WOsurrendersFast} below.
\begin{enumerate}
	\item We use $M\in\mathbb{N}$ time steps of length $\Delta_M=T/M.$
	\item $\mathbf{B}=[b_{jn}]_{j,n=1}^{m,N}$ denotes a matrix of size $m\times N$, containing the value of the VA contract. More precisely, the matrix $\mathbf{B}$ is updated at each time step, so that after the first iteration, $b_{jn}\approx\e\left[e^{-rT}\varphi(T, e^{\XmN_T+\rho \gamma(\Vm_T)}, \Vm_T)|\XmN_{T-\Delta_M}=x_n,\Vm_{T-\Delta_M}=v_j\right]$; after the second iteration, $b_{jn}\approx\e\left[e^{-rT}\varphi(T, e^{\XmN_T+\rho \gamma(\Vm_T)}, \Vm_T)|\XmN_{T-2\Delta_M}=x_n,\Vm_{T-2\Delta_M}=v_j\right]$, and so on. 
	\item $\mathbf{B}_{*,n}=[b_{jn}]_{j=1}^{m}$ denotes the $n$-th column of $\mathbf{B}$, $n=1,2,\ldots, N$,
	\item  $\mathbf{B}_{j,*}=[b_{jn}]_{n=1}^{N}$ denotes the $j$-th row of $\mathbf{B}$, $j=1,2,\ldots, m$.
	\item The symbol $\top$ indicates the matrix (vector) transpose operation.
\end{enumerate}  
\begin{algorithm}[h!]
	\caption{Variable Annuity without Early Surrenders via CTMC Approximation- Fast Algorithm}
	\label{algoVA_WOsurrendersFast}
	\DontPrintSemicolon
	\KwInput{Initialize $\mathbf{Q}^{(m)}$ as in \eqref{eqQ} and $\mathbf{G}^{(N)}_j$ for $j=1,2,\ldots,m$, as in \eqref{eqQk}\;
		$M\in\mathbb{N}$, the number of time steps, \;
		$\Delta_M\leftarrow T/M$, the size of a time step}
		Set $\mathbf{B}_{j,*}\leftarrow\left[\varphi(T, e^{x_n+\rho \gamma(v_j)},v_j)\right]_{n=1}^{N}$,  $j=1,2,\ldots,m$\;    
	\tcc{Calculate the transition probability matrices}
	
	\For{$j=1,2,\ldots, m$}
	{\tcc{Transition probability matrix of $\XmN$ given $\Vm=v_j$ over a period of length $\Delta_M$}
		$\mathbf{P}^X_j\leftarrow e^{\mathbf{G}_{j}^{(N)} \Delta_M}$}
	\tcc{Transition probability matrix of $\Vm$ over a period of length $\Delta_M$}
	$\mathbf{P}^V\leftarrow e^{\mathbf{Q}^{(m)} \Delta_M}$\;
	\tcc{VA valuation}
	\For{$z=M-1,\ldots,0$}{
		\For{$j=1,2,\ldots, m$}
		{$\tilde{\mathbf{H}}_{*,j}\leftarrow \mathbf{P}^X_j\mathbf{B}_{j,*}^{\top}$}
		\For{$n=1,2,\ldots, N$}
		{$\mathbf{B}_{*,n}\leftarrow \mathbf{P}^V\tilde{\mathbf{H}}_{n,*}^{\top}$}
	}	
	\KwRet $b_{ki}$\;
\end{algorithm}
Note that  at the end of the last iteration (i.e. when $z=0$), we have that
\begin{equation}
b_{ki}\approx\e\left[e^{-rT}\varphi(T, e^{\XmN_T+\rho \gamma(\Vm_T)}, \Vm_T)|\XmN_{0}=x_i,\Vm_{0}=v_k\right].\label{eqCTMCeuroVAapprox}
\end{equation}
\begin{sloppypar}
The computational gain of using Algorithm \ref{algoVA_WOsurrendersFast} over the previous algorithm  comes at the cost of a loss of accuracy since the conditional expectations are approximated over small time intervals (refer to Proposition \ref{propExpectionApprox}). Numerical experiments below show that highly accurate results are obtained in seconds when the time step is small, but the algorithm can perform poorly when the time step is not small enough.
\end{sloppypar}
\begin{remark}
 Since at each time step, Algorithm \ref{algoVA_WOsurrendersFast}, or the  ``Fast Algorithm'', takes advantage of the tower property of conditional expectations over short time periods of the same length $\Delta_M$, the transition probability matrix can be pre-computed at the beginning of the procedure and stored, accelerating the numerical process greatly. Moreover, as discussed in \cite{cui2018general} in the context of option pricing, the method can be used to simultaneously price variable annuities with different guaranteed amounts since the main computational effort in the algorithm resides in the calculation of the matrix exponentials.
\end{remark}
\subsection{Variable Annuity with Early Surrenders}
We approximate the value of the VA contract (including surrender rights) by its Bermudan counterpart for a large number of monitoring dates. 
Bermudan options can be exercised early, but only at predetermined dates $R \subset [0,T]$. Thus, Bermudan options are similar to American options, 
but the region of the permitted exercise times is a subset of $[0,T]$ containing a finite number of exercise dates, $\{t_0,t_1,\ldots,t_M\}$ with $t_z\in[0,T]$, $z=0,1,2,\ldots,M$ for some $M\in\mathbb{N}$.\\
\\
In this paper, we use the term Bermudan (resp. American) contract to refer to a variable annuity under which the policyholder has the right to surrender her contract prior to maturity on predetermined dates (resp. at any time prior to maturity). In the same vein, a variable annuity without surrender rights is also called an European contract. Note that these terms do not refer to existing contracts, and they are used to simplify explanations. Naturally, as $M\rightarrow\infty$, we expect the price of the Bermudan contract to converge to the one of a variable annuity with surrender rights as defined in \eqref{eqAmOptVA}. 
The latter is formalized in the following.\\
\\
Let $\Delta_M = T/M$ for some $M \in \mathbb N$ and define the set $\mathcal{H}_M=\{t_0,t_1,\ldots,t_M\}$ where $t_z=z\Delta_M$, ${z=0,1,2,\ldots, M}$, so that $t_0=0$ and $t_M=T$. The time-$t$ risk-neutral value of the Bermudan contract with permitted exercise dates $\mathcal{H}_M$ is defined by 
\begin{eqnarray}
b_M(t,x,y)& =& \sup_{\tau \in\mathcal{T}_{\Delta_M},\tau\geq t}\e_{t,x,y}[e^{-r(\tau-t)}\varphi(\tau,F_{\tau},V_{\tau})] \label{eqBermVA}
\end{eqnarray}

where $\mathcal{T}_{\Delta_M}$ is the set of stopping times taking values in $\mathcal{H}_M$. Proposition \ref{propBermAmConv} below shows that $b_M(t,x,y)\rightarrow v(t,x,y)$ as $M\rightarrow\infty$.
\begin{proposition}\label{propBermAmConv}
	As $M\rightarrow\infty$, the value function of the Bermudan variable annuity contract \eqref{eqBermVA} converges to its American counterpart \eqref{eqAmOptVA}, that is,
	$$\lim_{M\rightarrow\infty} b_M(t,x,y)=v(t,x,v).$$
\end{proposition}
The proof can be found Appendix \ref{Proof-4.5}.

\begin{remark}
	Using a Bermudan option with a large number of monitoring times to approximate the value of an American option is common in finance; see, for example, \cite{amin1994convergence} and \cite{lamberton2002brownian}. 
	In fact, since the value of an American option generally does not have closed-form expression, such a time-discretisation is often necessary to implement many numerical techniques. 
	\\
	In the practical context of variables annuities, surrenders are often only allowed at specific times (such as on the policy anniversary dates). In that case, the Bermudan contract is thus more realistic than its American counterpart in \eqref{eqAmOptVA}. Based on this idea, some authors propose numerical procedures which allow for surrenders only at specific points in time, see \cite{moenig2018lapse} and \cite{bernard2019less}. 
\end{remark}
We denote the Bermudan contract value process by $B:=\{B_z:=b_M(t_z,F_{t_z},V_{t_z})\}_{0\leq z\leq M}$. From the principle of dynamic programming (see for example \cite{lamberton1998}, Theorem 10.1.3), it is well-known that the discretized problem admits the following representation:
\begin{equation}
\left\{	\begin{array}{lll}
	   B_M&=\varphi(T,F_T, V_T) & \\
	   B_z& =\max\left(\varphi(t_z,F_{t_z}, V_{t_z}),e^{-r\Delta_M}\e_{t_z}[B_{z+1}]\right),&\quad{0\leq z\leq M-1}.
	\end{array}\right.
\end{equation}
Using CTMCs, we can define an approximation for the time-$t$ risk-neutral value of the Bermudan contract by
\begin{eqnarray*}
	b_M^{(m,N)}(t,x,y)& =& \sup_{\tau \in\mathcal{T}_{\Delta_M},\tau\geq t}\e\left[e^{-r\tau}\varphi(\tau,\FmN_{\tau},\Vm_{\tau})|\FmN_t=x, \Vm_t=y\right].\\ 
\end{eqnarray*}
The approximation of the Bermudan contract value  process, denoted by\newline $B^{(m,N)}:=\{B_z^{(m,N)}:=b_M^{(m,N)}(t_z,\FmN_{t_z}, \Vm_{t_z})\}_{0\leq z\leq M}$, is thus given by 
 \begin{equation}
 	\left\{	\begin{array}{lll}
 		B_M^{(m,N)}&=\varphi\left(T,\FmN_T, \Vm_T\right) & \\
 		B_z^{(m,N)}& =\max\left(\varphi\left(t_z,\FmN_{t_z}, \Vm_{t_z}\right) ,e^{-r\Delta_M}\e_{t_z}[B_{z+1}^{(m,N)}]\right),&\quad{0\leq z\leq M-1},
 	\end{array}\right.\label{Chap2eqBermCTMC}
 \end{equation}
which can also be written as
\begin{equation}
	\left\{	\begin{array}{lll}
		B_M^{(m,N)}&=\max\left(G,\FmN_T\right) & \\
		B_z^{(m,N)}& =\max\left(g(t,\Vm_t)\FmN_{t_z} ,e^{-r\Delta_M}\e_{t_z}[B_{z+1}^{(m,N)}]\right),&\quad{0\leq z\leq M-1}.
	\end{array}\right.\label{Chap2eqBermCTMC2}
\end{equation}
Or equivalently in terms of the process $Y^{(m,N)}$, we have that
\begin{equation}
	\left\{	\begin{array}{lll}
		B_M^{(m,N)}&=\varphi(T,\psi^{-1}(Y_T^{(m,N)})) & \\
		B_z^{(m,N)}& =\max\left(\tilde{\varphi}(t_z,\psi^{-1}(Y_{t_z}^{(m,N)})) ,e^{-r\Delta_M}\e_{t_z}[B_{z+1}^{(m,N)}]\right),&\quad{0\leq z\leq M-1},
	\end{array}\right.\label{Chap2eqBermCTMC3}
\end{equation}
where $\psi^{-1}$ is defined in Proposition \ref{propCTMC_XtoY}\footnote{Recall that $\psi^{-1}(n_y)=(x_n,v_l)$ for $n_y\in\mathcal{S}_Y^{(m,N)}$, where $n\leq N$ is the unique integer such that $n_y=(l-1)N+n$  for some $l\in\{1,2,\ldots,m\}$.}. 
\\
\\
Finally, we have $b_M(0,F_0,V_0)=B_0\approx B^{(m,N)}_0=b^{(m,N)}_M(0, F_0, V_0)$.\\
\\
Based on the above, an approximation for the value of the Bermudan contract can be obtained as described in the proposition below.
\begin{proposition}\label{propVAvalue_surreder}
	Let $F_0>0$ , $V_0\in\mathcal{S}_V$ and $\mathbf{G}^{(m,N)}$ be the generator defined in \eqref{Chap2eqQ_Y}. The risk-neutral value of a variable annuity with maturity $T>0$ and guaranteed amount $G>0$ can be approximated recursively by
    \begin{equation}
	\begin{array}{lll}
		\mathbf{B}^{(m,N)}_M & =\mathbf{H}^{(1)}, & \quad\\
		\mathbf{B}^{(m,N)}_z & = \max\{\mathbf{H}_z^{(2)}, e^{-r\Delta_M}\exp\{\Delta_M\mathbf{G}^{(m,N)}\}\mathbf{B}^{(m,N)}_{z+1}\} &0\leq z\leq M-1,
	\end{array}\label{eqAMoption2}
\end{equation}
for $M \in \mathbb N$ sufficiently large and where the maximum is taken element by element (also known as the parallel maxima). $\mathbf{H}^{(1)}$ and  $\mathbf{H}_z^{(2)}$, $z=0,1,\ldots,M-1$ are column vectors of size $mN\times 1$ whose $(l-1)N+n$-th entries, $h_{(l-1)N+n}^{(1)}$ and $h_{z,(l-1)N+n}^{(2)}$, are respectively given by 
\begin{equation}
\begin{array}{ll}
    h_{(l-1)N+n}^{(1)} & =\max(G,e^{x_n+\rho \gamma(v_l)}),\textrm{ and}\\ 
    h_{z,(l-1)N+n}^{(2)} & =g(t_z,v_l)e^{x_n+\rho \gamma( v_l)},
\end{array}\label{Chap2eqH}
\end{equation}
$1\leq l\leq m$ and $1\leq n\leq N$.\\
\\
Specifically, given $X_0^{(m,N)}=x_i=\ln(F_0)-\rho \gamma(V_0)$ and $V_0^{(m)}=V_0=v_k$, the approximated value of the Bermudan contract is given by
$$b^{(m,N)}_M(0,F_0,V_0)=\mathbf{e}_{ik}\mathbf{B}^{(m,N)}_0.$$
\end{proposition}
Hence, based on the last proposition, Algorithm \ref{algoVAsurrenders} below provides a CTMC approximation for the value of a variable annuity contract (including early surrenders).

\begin{algorithm}
	\caption{Variable Annuity with Early Surrenders via CTMC Approximation}
	\label{algoVAsurrenders}
	\DontPrintSemicolon
	\KwInput{Initialize $\mathbf{G}^{(m,N)}$ as in \eqref{Chap2eqQ_Y}, $\mathbf{H}^{(1)}$ and $\mathbf{H}^{(2)}_z$, for $z=0,1,\ldots,M-1$, as in \eqref{Chap2eqH}\;
	$M\in\mathbb{N}$, the number of time steps, \;
	$\Delta_M\leftarrow T/M$, the size of a time step}
	Set $\mathbf{B}_M^{(m,N)}\leftarrow\mathbf{H}^{(1)}$ and $\mathbf{A}_{\Delta_M}\leftarrow\exp\{\Delta_M \mathbf{G}^{(m,N)}\}e^{-r\Delta_M}$\;
	\For{$z=M-1,M-2,\ldots,0$}
	{$\mathbf{B}_z^{(m,N)}\leftarrow\max\left\{\mathbf{H}_z^{(2)}, \mathbf{A}_{\Delta_M}\mathbf{B}_{z+1}^{(m,N)}\right\}$}
	$b_M^{(m,N)}(0,F_0, V_0)\leftarrow\mathbf{e}_{ik} \mathbf{B}_0^{(m,N)}$\;
	\KwRet $b_M^{(m,N)}(0,F_0, V_0)$\;
\end{algorithm}

\begin{remark}[Convergence of VA prices with early surrenders]\label{remarkVAwithSurrendersConv}
	 Recall, from Remark \ref{remakWeakConvergenceCTMC} that $\FmN \Rightarrow F$ as $m,N\rightarrow\infty$.  The convergence of the price of the Bermudan contract written on $F^{(m,N)}$ to the price of the Bermudan contract written on $F$, that is $B_0^{(m,N)}\rightarrow B_0$ as $m,N\rightarrow \infty$, follows from \cite{song2013weak}, Theorem 9, and the results of \cite{palczewski2010finite}, Theorem 3.5, on the alternative continuous reward representation of the  value function $v$\footnote{Many convergence results, such as the one in \cite{song2013weak} and \cite{palczewski2010finite}, require the reward function to be bounded. However, as mentioned in \cite{mijatovic2013continuously}, Remark 5.4 and \cite{cui2018general}, Remark 5, the original payoff $\varphi$ can be replaced by the truncated payoff $\varphi\wedge L$ with a constant $L$ sufficiently large without altering the accuracy of the numerical results.}. Finally, the convergence of the price of the Bermudan contract to its American counterpart as $M$ goes to infinity follows from Proposition \ref{propBermAmConv}.\\
	To our knowledge, detailed error and convergence analysis for the two-layers CTMC approximation of early-exercise options have not yet been performed in the literature. However, \cite{cui2018general} demonstrate the accuracy of the approximation numerically in the context of American put option pricing.
\end{remark}

\subsubsection{Fast Algorithms}
Algorithm \ref{algoVAsurrenders}, although theoretically correct, might stretch computing resources to unacceptable levels as $m,N$ increase because of the size of the exponent $\mathbf{G}^{(m,N)}$ in the calculation of the matrix exponential (even if the matrix exponential is calculated only once at the beginning of the procedure). 
However, similarly to the no surrender case, the efficiency of Algorithm \ref{algoVAsurrenders} can be improved significantly by using an approximation based on the assumption that the variance process remains constant over small time periods (see Proposition \ref{propExpectionApprox}). 
Indeed, Algorithm \ref{algoVA_WOsurrendersFast} can be easily adapted to the valuation of the Bermudan contract since the continuation value in Algorithm \ref{algoVAsurrenders} is already calculated on a small time interval (thus, it is not necessary to use the tower property of conditional expectations in order to apply the fast methodology.).\\
\\
Recall that $\mathbf{Q}^{(m)}$ is the generator of $\Vm$, defined in \eqref{eqQ}, and $\mathbf{G}^{(N)}_j$ is the generator of $X^{(m,N)}$ given the variance process is in state $v_j$, defined in \eqref{eqQk}. Moreover, let $\bm{\varphi}(t):=[\varphi(t,e^{x_n+\rho \gamma(v_j)}, v_j)]_{j,n=1}^{m,N}$ be a $m\times N$ matrix representing the payoff at time $t$ for each state in $\mathcal{S}_V^{(m)}\times\mathcal{S}_X^{(N)}$, and $\varphi_{j,*}(t)$ be the $j$-th row of  $\bm{\varphi}(t)$. 
We denote the matrix (vector) transpose operation by $\top$.\\
The Fast Algorithm to value VA with surrender rights is given in Algorithm \ref{algoVA_WITHsurrendersFast}.
\begin{algorithm}[h!]
	\caption{Variable Annuity with Early Surrenders via CTMC Approximation - Fast Algorithm}
	\label{algoVA_WITHsurrendersFast}
	\DontPrintSemicolon
		\KwInput{Initialize $\mathbf{Q}^{(m)}$ as in \eqref{eqQ} and $\mathbf{G}^{(N)}_j$ for $j=1,2,\ldots,m$, as in \eqref{eqQk}\;
		$M\in\mathbb{N}$, the number of time steps, \;
		$\Delta_M\leftarrow T/M$, the size of a time step}
	Set $\bm{\varphi}(t_z)\leftarrow[\varphi(t_z,e^{x_n+\rho \gamma(v_j)}, v_j)]_{j,n=1}^{m,N}$ for $z=0,1,\ldots,M$\;
	Set $\mathbf{B}_{j,*}\leftarrow\bm{\varphi}_{j,*}(t_M)$ for $j=1,2,\ldots,m$\;
	\tcc{Calculate the transition probability matrices}
	
	\For{$j=1,2,\ldots, m$}{\label{algoFastline3}
	\tcc{Transition probability  matrix of $\XmN$ given $\Vm=v_j$ over a period of length $\Delta_M$}
		$\mathbf{P}^X_j\leftarrow e^{\mathbf{G}_{j}^{(N)} \Delta_M}$}
	\tcc{Transition probability matrix of $\Vm$ over a period of length $\Delta_M$}
	$\mathbf{P}^V\leftarrow e^{\mathbf{Q}^{(m)} \Delta_M}$\label{algoFastline5}\;
	\tcc{VA valuation}
	\For{$z=M-1,\ldots,0$}{\label{algoFastline6}
		\For{$j=1,2,\ldots, m$}
		{$\mathbf{E}_{*,j}\leftarrow \mathbf{P}^X_j\mathbf{B}_{j,*}^{\top}$}
		\For{$n=1,2,\ldots, N$}{\label{algoFastline9}
		$\mathbf{B}_{*,n}\leftarrow \mathbf{P}^V\mathbf{E}_{n,*}^{\top}$}\label{algoFastline10}
		$\mathbf{B}=\max(\mathbf{B},\bm{\varphi}(t_z))$\label{lineAlgoFast10}
	}	
	\KwRet $b_{ki}$\;
\end{algorithm}

The Fast Algorithms to price VA contracts with and without early surrenders are very similar. 
The only difference is the additional line \ref{lineAlgoFast10} in Algorithm \ref{algoVA_WITHsurrendersFast}. 
In fact, at a given time $t_z$ (that is, we fix one $z$ in the loop line \ref{algoFastline6} to \ref{lineAlgoFast10}), we can observe that at the end of the inner loop (line \ref{algoFastline9} and \ref{algoFastline10}) , the matrix $\mathbf{B}$ contains the continuation value of the Bermudan contract at $t_z$. 
Since Bermudan contracts can be surrendered at any time in $\mathcal{H}_M$, we simply need to calculate the maximum between the continuation value and the payoff at $t_z$ to obtain the value of the Bermudan contract at that time (line \ref{lineAlgoFast10}). 
Therefore, the only difference between the Fast Algorithms for VA pricing with and without early surrenders stems from the fact that the latter contract cannot be surrendered prior to maturity, and thus,  only the continuation value needs to be calculated at each time step (that is, line \ref{lineAlgoFast10} is not used to price VA contracts without early surrenders).\\
\\
The computational effort in Algorithms \ref{algoVA_WOsurrendersFast} and  \ref{algoVA_WITHsurrendersFast} resides in the calculation of the matrix exponentials (line \ref{algoFastline3} to line \ref{algoFastline5}). Hence, once they are (pre-)computed, one can price variable annuity contracts with and without surrender rights simultaneously at almost no additional cost. This also holds true for any other VA contracts with different guaranteed amount, that is, a large variety of contracts with different guarantee structure and surrender rights can be priced simultaneously for almost the same computational effort as a single contract. 
Numerical experiments below demonstrate the efficiency and the accuracy of the Fast Algorithm.

\subsection{Optimal Surrender Surface}\label{sectOptSurrSurf}
In this section, we provide an algorithm to approximate the optimal surrender strategy for variable annuities with a general fee structure depending on the fund value and the variance process.  Policyholder behaviour may significantly impact pricing and hedging of variable annuities, \cite{kling2014impact}. Thus, analyzing optimal surrender behaviour is crucial for insurers when developing risk management strategies for variable annuities, \cite{niittuinpera2020IAA}, \cite{bauer2017policyholder}. Optimal surrender strategies have been studied in the literature in different contexts, see for instance \cite{Mackay2014}, \cite{bernard2014optimal}, \cite{BernardMackay2015}, \cite{shen2016valuation} and \cite{kang2018optimal}.\\
\\
The objective is now to approximate the (optimal) surrender surface using the CTMC approximation. To this end, we first introduce some additional definitions and notations.\\

\begin{defi}\label{def:surrender_region}
	Let $E=[0,T]\times \reals_+\times \mathcal{S}_V$. The continuation region $\mathcal{C}\subset E$ is defined as
	$$
	\mathcal{C}=\lbrace (t,x,y)\in E:v(t,x,y)>\varphi(t,x,y)\rbrace,
	$$
	and the surrender region $\mathcal{D}\subseteq E$, as
	$$
	\mathcal{D}=\lbrace (t,x,y)\in E:v(t,x,y)=\varphi(t,x,y)\rbrace.
	$$
\end{defi} 

\begin{remark}
    It follows from Definition \ref{def:surrender_region} that $E=\mathcal{C}\cup\mathcal{D}$, since $v(t,x,y)\geq\varphi(t,x,y)$ for all $(t,x,y)\in E$.
\end{remark}

 If the value function $v$ is continuous, then $\mathcal{C}$ is an open set and $\mathcal{D}$ is closed; and the optimal surrender surface is the boundary $\partial\Cr$ of $\Cr$. \\
 \\
Definition \ref{def:surrender_region} provides a simple way of approximating the optimal surrender surface via CTMC approximation. 
To do so, denote by $f_{nl}$ the approximated fund process associated to $(x_n, v_l)\in\mathcal{S}^{(N)}_X \times\mathcal{S}_V^{(m)}$ such that $f_{nl}=e^{x_n+\rho \gamma(v_l)}$ and let $\mathcal{S}^{(m,N)}_F=\{f_{nl}\}_{n,l=1}^{N,m}$ be the state-space of $F^{(m,N)}$. 
We also denote by $\mathcal{S}^{(m,N)}_{F,l}=\{f_{nl}\}_{n=1}^{N}$ the $l$-section of $\mathcal{S}^{(m,N)}_F$.\\
\\
Moreover, let $\tilde{\mathcal{H}}_M=\{t_0,t_1,\ldots,t_{m-1}\}$, that is $\tilde{\mathcal{H}}_M=\mathcal{H}_M\setminus t_M$, and $\tilde{E}_l=\tilde{\mathcal{H}}_M\times\mathcal{S}^{(m,N)}_{F,l} \times\{v_l\}$.\\
\\
Using the previous definitions, the $l$-section, $l\in\{1,2,\ldots,m\}$, of the continuation and the surrender regions can be approximated using the CTMC processes via
$$
\mathcal{C}^{(m,N)}_l=\left\lbrace (t_z,f_{nl},v_l) \in \tilde{E}_l\Big|b^{(m,N)}_M(t_z,f_{nl},v_l)>g(t_z,v_l) f_{nl}\right\rbrace,
$$
and 
$$
\mathcal{D}^{(m,N)}_l=\left\lbrace  (t_z,f_{nl},v_l)\in \tilde{E}_l\Big| b^{(m,N)}_M(t_z,f_{nl},v_l)=g(t_z,v_l) f_{nl}\right\rbrace\cup\{T\}\times\mathcal{S}^{(m,N)}_{F,l}\times\{v_l\},
$$
respectively. Hence, the approximated continuation and surrender regions are given by
$$\mathcal{C}^{(m,N)}=\cup_{l=1}^m \mathcal{C}^{(m,N)}_l,\quad\textrm{and}\quad\mathcal{D}^{(m,N)}=\cup_{l=1}^m \mathcal{D}^{(m,N)}_l.$$
We use the notation of Proposition \ref{propVAvalue_surreder}, with  $b_{z,(l-1)N+n}$ denoting the $(l-1)N+n$-th entry of $\mathbf{B}_z^{(m,N)}$. 
For $(t_z,f_{nl}, v_l)$, $0\leq  z\leq M-1$, $1\leq n\leq N$ and $1\leq l\leq m$,  we set $(t_z,f_{nl},v_l)\in \mathcal{C}^{(m,N)}$ if $b_{z,(l-1)N+n}>g(t_z,v_l)f_{nl}$, and  $(t_z,f_{nl},v_l)\in \mathcal{D}^{(m,N)}$ otherwise. 
The approximated optimal surrender surface can then be obtained by analyzing the shape of $\mathcal{C}^{(m,N)}$.\\
\\
We are now interested in studying the shape of the surrender region. To do so, we fix $t\in[0,T)$ and  $y\in\mathcal{S}_V$ and consider the set of points $\mathcal{D}_{t,y}\subseteq \reals_+$ for which it is optimal to surrender the contract. More precisely, we define $\mathcal{D}_{t,y}$ by

$$\mathcal{D}_{t,y}=\left\{f\in \reals_+\Big| (t,f,y)\in\mathcal{D}\right\}$$

Suppose that, for all couples $(t,y)\in[0,T)\times\mathcal{S}_V$, the set $\mathcal{D}_{t,y}$ is of the form
$[f^\star(t,y),\infty)$ for some $f^\star(t,y)\in\reals_+$. That is, $f^\star(t,y)$ is the smallest fund value for which it is optimal to surrender the contract at time $t$ for a volatility level $y$, and for any fund value greater than $f^\star(t,y)$, it is also optimal to surrender. Mathematically, this  may be expressed by
\begin{align}
f^\star(t,y) := \inf\left\{f\in\reals_+\Big|f\in \mathcal{D}_{t,y}\right\}=\inf\{ \mathcal{D}_{t,y}\}.
\label{eqfstar}
\end{align}
Under this assumption, the continuation and the surrender regions can be expressed as 
$$
\mathcal{C}=\left\lbrace (t,f,y) \in E \Big|f< f^\star(t,y)\right\rbrace,
$$
and
$$
\mathcal{D}=\left\lbrace (t,f,y) \in E\Big|f\geq f^{*}(t,y)\right\rbrace\cup\{T\}\times\reals_+\times\mathcal{S}_V,
$$
respectively.\\
\\
Hence, under this assumption, the optimal surrender surface $f^{*}$ splits $E$ in two regions: at or above the surface is the surrender region, and below, the continuation region. 
That is, the set $\mathcal{D}_{t,y}$ is connected\footnote{A set $X$ is connected if it cannot be divided into two disjoint non-empty open sets.}. 
In this paper, we say that the surrender region is of ``threshold type'' if for any $(t,y)\in[0,T)\times\mathcal{S}_V$, the set $\mathcal{D}_{t,y}$ is connected.
There is a financial interpretation for such a form for the surrender region. 
As explained in \cite{milevsky2001real}, it is advantageous for the policyholder to hold on to the contract when the fund value is low since there is a higher chance that the guarantee will be triggered at maturity.
\begin{remark}
	The surrender region can take any shape; see for examples \cite{Mackay2017} Figure 4 and Figure 5. However, for specific fee and surrender charge structures, it can be shown that the surrender region is of threshold type, see for instance \cite{Mackay2014
	}, Appendix 2.A, when the index value process is modelled by a geometric Brownian motion. Other authors take this form for the surrender region as an initial assumption, see for example \cite{kang2018optimal}. In the context of financial derivatives pricing, \cite{jacka1991optimal}, Proposition 2.1.3, shows that the continuation region of American put options has a threshold type shape under the Black-Scholes setting whereas \cite{touzi1999american}, Section 2, proves it for some stochastic volatility models, and \cite{de2019lipschitz}, Proposition 4.1, in a very general setting.
\end{remark}

When the surrender region is of threshold type, a simple algorithm can be developed to approximate the optimal surrender surface. 
The idea is based on the definition of $f^\star(t,y)$ in \eqref{eqfstar} above: for each $t_z\in\mathcal{H}_M$ and $v_l\in\mathcal{S}_V^{(m,N)}$, we identify the smallest fund value $f^{(m,N)}(t_z,v_l)$ for which it is optimal to surrender. 
Algorithms \ref{algoOptSurrSurf} returns the approximated optimal surrender surface $f^{(m,N)}$ (under the assumption that the surrender region is of threshold type) and the approximated value of a variable annuity with early surrenders given  $X^{(m,N)}_0=x_i=\ln(F_0)-\rho \gamma(V_0)$ and $V_0^{(m)}=V_0=v_k$.

\begin{algorithm}[h!]
\caption{Optimal Surrender Surface (of threshold type) via CTMC Approximation}
\label{algoOptSurrSurf}
\DontPrintSemicolon
\KwInput{Initialize $\mathbf{G}^{(m,N)}$ as in \eqref{Chap2eqQ_Y}, $\mathbf{H}^{(1)}$ and $\mathbf{H}^{(2)}_z$, for $z=0,1,\ldots,M-1$, as in \eqref{Chap2eqH}\;
	$M\in\mathbb{N}$, the number of time steps, \;
	$\Delta_M\leftarrow T/M$, the size of a time step}
Set $\mathbf{B}_M^{(m,N)}\leftarrow\mathbf{H}^{(1)}$ and $\mathbf{A}_{\Delta_M}\leftarrow\exp\{\Delta_M \mathbf{G}^{(m,N)}\}e^{-r\Delta_M}$\;
	\For{$z=M-1,M-2,\ldots,0$}
	    {$\mathbf{B}_z^{(m,N)}\leftarrow\max\left\{\mathbf{H}_z^{(2)}, \mathbf{A}_{\Delta_M}\mathbf{B}_{z+1}^{(m,N)}\right\}$ \;
	    \For{$l=1,2\ldots,m$}
	    {$n\leftarrow 1$\;
	    \While{ $\left(b_{z,(l-1)N+n}>g(t_z,v_l)e^{x_n+\rho \gamma(v_l)}\right)$ and $(n<N)$}
	    {$n\leftarrow n+1$}
        $f^{(m,N)}(t_z,v_l)\leftarrow e^{x_{n}+\rho \gamma( v_l)}$
         }
        }
  $b_M^{(m,N)}(0,F_0, V_0)\leftarrow\mathbf{e}_{ik} \mathbf{B}_0^{(m,N)}$\; 
 \KwRet $f^{(m,N)}$ and $b_M^{(m,N)}(0,F_0, V_0)$\;
\end{algorithm}

We note that the derivation of the optimal surrender surface is not mandatory to obtain the value of the Bermudan contract, as observed from Algorithm \ref{algoVAsurrenders} or Algorithm  \ref{algoVA_WITHsurrendersFast} in the previous subsection.\\
\\
Similarly as above, Algorithm \ref{algoOptSurrSurfFast} is the fast version of Algorithm \ref{algoOptSurrSurf}. Recall that $\mathbf{Q}^{(m)}$ is the generator of $\Vm$ defined in \eqref{eqQ}, $\mathbf{G}^{(N)}_j$ is the generator of $X^{(m,N)}$ given the variance process is in state $v_j$ defined in \eqref{eqQk}, $\bm{\varphi}(t):=[\varphi(t,e^{x_n+\rho \gamma(v_j)})]_{j,n=1}^{m,N}$ is a $m\times N$ matrix representing the payoff at time $t$ for each state in $\mathcal{S}_V^{(m)}\times\mathcal{S}_X^{(N)}$, and $\bm{\varphi}_{j,*}(t)$ (resp $\mathbf{B}_{j,*}$) is the $j$-th row of  $\bm{\varphi}(t)$ (resp. $\mathbf{B}$). We also denote by $b_{jn}$, the $(j,n)$ entry of the matrix $\mathbf{B}$.
\begin{algorithm}[h!]
	\caption{Optimal Surrender Surface (of threshold type) via CTMC Approximation - Fast Algorithm}
	\label{algoOptSurrSurfFast}
	\DontPrintSemicolon
		\KwInput{Initialize $\mathbf{Q}^{(m)}$ as in \eqref{eqQ} and $\mathbf{G}^{(N)}_j$ for $j=1,2,\ldots,m$, as in \eqref{eqQk}\;
		$M\in\mathbb{N}$, the number of time steps, \;
		$\Delta_M\leftarrow T/M$, the size of a time step}
	Set $\bm{\varphi}(t_z)\leftarrow[\varphi(t_z,e^{x_n+\rho \gamma(v_j)}, v_j)]_{j,n=1}^{m,N}$ for $z=0,1,\ldots,M$\;
	Set $\mathbf{B}_{j,*}\leftarrow\bm{\varphi}_{j,*}(t_M)$ for $j=1,2,\ldots,m$\;
	\tcc{Calculate the transition probability matrices}
	
	\For{$j=1,2,\ldots, m$}
	{\tcc{Transition probability matrix of $\XmN$ given $\Vm=v_j$ over a period of length $\Delta_M$}
		$\mathbf{P}^X_j\leftarrow e^{\mathbf{G}_{j}^{(N)} \Delta_M}$}
	\tcc{Transition probability matrix of $\Vm$ over a period of length $\Delta_M$}
	$\mathbf{P}^V\leftarrow e^{\mathbf{Q}^{(m)} \Delta_M}$\;
	\tcc{VA valuation}
	\For{$z=M,M-1,\ldots,0$}{
		\For{$j=1,2,\ldots, m$}
		{$\mathbf{E}_{*,j}\leftarrow \mathbf{P}^X_j\mathbf{B}_{j,*}^{\top}$}
		\For{$n=1,2,\ldots, N$}
		{$\mathbf{B}_{*,n}\leftarrow \mathbf{P}^V\mathbf{E}_{n,*}^{\top}$}
		$\mathbf{B}=\max(\mathbf{B},\bm{\varphi}(t_z))$\;
	    \For{$j=1,2\ldots,m$}
		{$n\leftarrow 1$\;
			\While{ $\left(b_{jn}>g(t_z,v_j)e^{x_n+\rho \gamma(v_j)}\right)$ and $(n<N)$}
			{$n\leftarrow n+1$}
			$f^{(m,N)}(t_z,v_j)\leftarrow e^{x_{n}+\rho \gamma( v_j)}$
		}
	}

	\KwRet $f^{(m,N)}$ and $b_{ki}$\;
\end{algorithm}
\begin{remark}
	Algorithms \ref{algoVAsurrenders} and \ref{algoVA_WITHsurrendersFast} do not require the specification of any particular form for the surrender region, which is not the case for many of the numerical procedures presented in the literature (\cite{bernard2014optimal},\cite{shen2016valuation} and \cite{kang2018optimal}). Hence their scope is more general. 
\end{remark}
The accuracy of the approximated surrender boundary is demonstrated numerically in Appendix \ref{appendixSupplMaterial} (available online as supplemental material).
\subsection{CTMC Approximation of the VIX}\label{subsectionCTMCapproxVIX}
In Section \ref{sectNumEx}, we analyze numerically the impact of various VIX-linked fee structure on the optimal surrender strategy. 
Since analytical formulas for the $\vix$ are not always known for all models listed in Table \ref{tblExSVmodels}, we use a CTMC approach to approximate the value of the volatility index.
This is the case for numerical experiments performed under the 3/2 model whose results are available online in Appendix \ref{appendixSupplMaterial}\footnote{For the 3/2 model, a closed-form expression for the $\vix$ may be found in \cite{carr2007new}, Theorem 4. However, as pointed out by \cite{drimus2012options}, the integral that appears in the analytical formula is difficult to implement and is not suited for fast and accurate numerical methods. For this reason, the CTMC approximation of the VIX is used in the numerical examples under the 3/2 model.\label{footnoteApproxVIX32}}.
In this section, we propose an approximation for the $\vix$ when the variance process is approximated by a CTMC.\\
\\
Let $\left\{\vix_t^2\right\}_{t\geq 0}$ be the process representing the square of the VIX, which is given by
$$\vix_t^2=\e_t\left[\frac{1}{\tau}\int_t^{t+\tau}\sigma^2_S(V_{s})\diff s\right],$$ 
with $\tau=30/365$, see \cite{cui2021valuation} Equation (6) for details. 
\begin{sloppypar}
Recall that $V^{(m)}$ is the CTMC approximation of $V$ taking values on a finite state-space ${\mathcal{S}_V^{(m)}:=\{v_1,v_2,\ldots,v_m\}}$.
\end{sloppypar}
When the variance process is a CTMC, a closed-form expression can be obtained for the value of the volatility index. The CTMC approximation of the VIX, denoted by $\vix^{(m)}=\{\vix_t^{(m)}\}_{t\geq 0}$, is given in the proposition below.
\begin{proposition}\label{propCTMCapproxVIX}
	Given $V_t^{(m)}=v_k$, the square of the VIX index at time-$t$ can be approximated by  
	\begin{eqnarray}
		\left(\vix_{t}^{(m),k}\right)^2&:=&\e\left[\frac{1}{\tau}\int_t^{t+\tau}\sigma^2_S(V^{(m)}_{s})\diff s\big | V_t^{(m)}=v_k\right]\nonumber\\
	&=&\frac{1}{\tau}\int_0^{\tau} \mathbf{e}_k e^{\mathbf{Q}^{(m)} s} \mathbf{H}\diff s,\label{eqVIX2}
	\end{eqnarray}
	where $\tau=30/365$, $\mathbf{Q}^{(m)}$ is the generator of $V^{(m)}$ defined in \eqref{eqQ}, $\mathbf{e}_k$ is a row vector of size $1\times m$ with all entries equal to $0$ except the $k$-th entry and $\mathbf{H}$ is a $m\times 1$ vector whose $j$-th entry $h_j$ is given by $h_j=\sigma_S^2(v_j)$, $j=1,2,\ldots,m$.
\end{proposition}
The proof is a direct consequence of Fubini's Theorem and the CTMC representation for conditional expectations.
\begin{remark} Because of the time-homogeneity of $V^{(m)}$ the approximation does not depend on $t$ and thus, it needs to be calculated only once for all $t\geq 0$.
\end{remark}
The integral part in \eqref{eqVIX2} can be approximated via a quadrature rule. Practically, this is done by dividing the interval $[0,\,\tau]$ into $n>0$ equidistant sub-intervals for $z=0, 1,\,2,\ldots, n$ with
$$t_z:=z\Delta_n,\quad \Delta_n:=\tau/n.$$
The approximation then becomes
\begin{equation}
	\left(\vix_{t}^{(m),k}\right)^2\approx\frac{\Delta_n}{\tau}\mathbf{e}_k\sum_{z=1}^{n}e^{\mathbf{Q}^{(m)} t_z}\mathbf{H}.\label{eqVIX_CTMCapprox1}
\end{equation}

Equation \eqref{eqVIX_CTMCapprox1} can be implemented in a straightforward manner. 
However, it requires calculating the exponential of a matrix $n$ times, which can be computationally inefficient. 
By making use of the tower property of conditional expectations, Algorithm \ref{algoCTMCapproxVIX} can be used to speed up the calculation of \eqref{eqVIX_CTMCapprox1}.
\begin{algorithm}
	\caption{Efficient Algorithm for the calculation of the VIX using CTMC approximation}
	\label{algoCTMCapproxVIX}
	\DontPrintSemicolon
	
	\KwInput{Initialize $\mathbf{Q}^{(m)}$ as in \eqref{eqQ} and $\mathbf{H}$ as in Proposition \ref{propCTMCapproxVIX}\;
		$n\in\mathbb{N}$, the number of time steps, \;
		$\Delta_n\leftarrow \tau/n$, the size of a time step}
	Set $\mathbf{A}_{\Delta_n}\leftarrow\exp\{\Delta_n \mathbf{Q}^{(m)}\}$, $\mathbf{S}\leftarrow \mathbf{0}_{m\times1}$ and $\mathbf{E}\leftarrow \mathbf{H}$\;
	\For{$z=n,n-1,\ldots,1$}
	{$\mathbf{E}\leftarrow \mathbf{A}_{\Delta_n}\mathbf{E}$ \;
	 $\mathbf{S}\leftarrow \mathbf{S}+ \mathbf{E}$		
	}
	$\vix_{t}^{(m),k}\leftarrow\sqrt{\mathbf{e}_k\mathbf{S}\frac{\Delta_n}{\tau}}$\;
\end{algorithm}

The vector $\mathbf{0}_{m\times1}$ represents the null column vector of size $m$. Note that Algorithm \ref{algoCTMCapproxVIX} requires the calculation of a matrix exponential only once at the beginning of the procedure, which makes the algorithm very efficient.
\begin{remark}
	The convergence of the approximation follows easily from the weak convergence of the CTMC approximation (with a reasoning similar to that of Remark \ref{remarkVAwoSurrendersConv}), and the convergence of the time step discretisation, from the monotone convergence theorem. 
\end{remark}
Numerical experiments in the next section also demonstrate the accuracy and the efficiency of Algorithm \ref{algoCTMCapproxVIX} empirically.\\

\section{Numerical Analysis}\label{sectNumEx}
In the constant fee case, the misalignment between the fees and the value of the financial guarantee creates an incentive for the policyholder to surrender her policy prematurely (see \cite{milevsky2001real} for details). 
Indeed, when the fund value is high, the amount of the fees paid is also high, but the put option embedded in a GMMB is out-of-the-money and worth very little since the probability for the guarantee to be triggered at maturity is low. 
Thus, the policyholder pays high fees for a financial guarantee that has a low value. 
This is clearly an incentive for a policyholder to surrender his policy prior to the term. \cite{Mackay2017} considered state-dependent fee structures where the fee is paid when the fund value is under a certain level, and showed that this particular type of fee structure reduces insurer's exposure to policyholder behaviour under risk-neutral value maximization assumption. 
Fees that are tied to the  S\&P volatility index, VIX, are also studied in the literature, \cite{cui2017}, \cite{bernard2016variable} and \cite{kouritzin2018vix}. 
Since the volatility is negatively correlated with the stock price (see for instance \cite{Rebonato2005}), we expect $\vix$-linked fees to be low when the fund value is high and to be higher when the fund value is low. \cite{cui2017} and \cite{kouritzin2018vix} showed numerically that linking the fee to the volatility index $\vix$ may help  realign revenues with variable annuity liabilities (for VA without surrender rights).  Hence, it is normal to believe that linking the fees to the $\vix$ may help to reduce insurers' exposure to surrender risk. 
This will be explored in greater details in the numerical experiments conducted below.\\
	\\
This section first discusses the market, the VA, and the CTMC parameters used in all numerical experiments performed below.  Then, we investigate the efficiency of 
the Fast Algorithms (Algorithms \ref{algoVA_WOsurrendersFast} and \ref{algoVA_WITHsurrendersFast}). In the third subsection, we discuss different structures for the $\vix$-linked fee, that is, different ways to link the fee rate to the $\vix$ index. Finally, we analyze the impact of different $\vix$-linked fee structures on the value of VAs and their optimal surrender strategy. We restrict our analysis to the classical Heston model. Using our framework, any model listed in Table \ref{tblExSVmodels} could have been used.
\\
As supplementary material (available online), we investigate the numerical accuracy of the approximated optimal surrender surface derived in Algorithm \ref{algoOptSurrSurf} (or equivalently \ref{algoOptSurrSurfFast}), and the CTMC approximation of the VIX (Algorithm \ref{algoCTMCapproxVIX}). Finally, we explore numerically the impact of VIX-linked fee structures under the $3/2$ model.  

\subsection{Market, VA and CTMC Parameters}\label{subsecMarketVACTMCparam}
We consider a market under regular conditions: low initial variance $v_0$, low long-term variance $\theta$, moderate volatility of volatility $\sigma$, and moderate speed reversion $\kappa$\footnote{Bloomberg provides historical Heston calibrated parameters to market data on a daily basis via its Option Pricing template (OVME). These parameters are currently used in practice for over-the-counter option pricing. Bloomberg's Heston calibrated speed reversion parameter is $\kappa=3.6881$ as of December 31, 2019, $\kappa=5$ as of March 31, 2020 and $\kappa=1.1397$ as of September 30, 2020. The parameter selected for our numerical experiments falls approximately in the middle of those of December 2019 and September 2020. In the financial literature, \cite{ai2007maximum} obtain $\kappa=5.07$ whereas \cite{garcia2011estimation} obtain $\kappa=0.173$, and again our values fall between these two values.}.  The initial value of the variance is set to $V_0=0.03$, the correlation to $\rho=-0.75$ and the risk-free rate to $r=0.03$. The selected market parameters are summarized in Table \ref{tblMarketParam} and the model dynamic is given in Table \ref{tblExSVmodels}.
\begin{table}[h!]
	\begin{tabular}{ccccccc}
		\hline
		Parameter & ${V_0}$ & ${\kappa}$ & ${\theta}$ & ${\sigma}$ & $\mathbf{\rho}$ & ${r}$\\
		Value  & $0.03$ & $2.00$ & $0.04$ & $0.20$ & $-0.75$ & $0.03$\\
		\hline
	\end{tabular}
	\caption{Market Parameters}
	\label{tblMarketParam}
\end{table}

The variable annuity parameters are set to $F_0=S_0=100$, $T=10$ (years), and $G=100$. We assume that the payoff when the contract is surrendered early is given by $g(t,V_t)F_t$, with 
$$g(t,y)=e^{-k(T-t)},\quad y\in\mathcal{S}_V,$$
and $k=0.2\%$. As mentioned previously, such a form for the early surrender payoff is common in the actuarial literature,  see \cite{bernard2014optimal}, \cite{Mackay2014}, \cite{BernardMackay2015},\cite{Mackay2017}, \cite{shen2016valuation}, \cite{bacinello2019variable} and \cite{kang2018optimal}. The choices for the fee function $c(x,y)$ are discussed in greater details in the next section.\\
\\
Note that a numerical analysis under the Heston model with $G=F_0e^{\tilde{g} T}$, $\tilde{g}=2\%$ (rather than $G=100$) is also performed in Appendix \ref{appendixSupplMaterial}, available online as supplemental material. 
\\

For all numerical examples in this paper, we use the non-uniform grid proposed by Tavella and Randall (\cite{tavellapricing}, Chapter 5.3). For example, suppose that $\tilde{X}$ is a one-dimensional diffusion process approximated by a continuous-time Markov chain $\tilde{X}^{(n)}$ taking values on a finite state-space  $\mathcal{S}_{\tilde{X}}=\{\tilde{x}_1,\,\tilde{x}_2,\ldots,\tilde{x}_n\}$, $n\in\mathbb{N}$. The state-space of the approximated process can be determined as follows:
\begin{equation}
	\tilde{x}_i=\tilde{X}_0+\tilde{\alpha}\sinh\left(c_2 \frac{i}{m}+c_1\left[1-\frac{i}{m}\right]\right),\quad i=2,\,\ldots,\,m-1,
\end{equation}\label{eqGridTavRand}
where
$$c_1=\sinh^{-1}\left(\frac{\tilde{x}_1-\tilde{X}_0}{\tilde{\alpha}}\right),\quad\textrm{and } c_2=\sinh^{-1}\left(\frac{\tilde{x}_m-\tilde{X}_0}{\tilde{\alpha}}\right).$$
 Here the constant, $\tilde{\alpha}\geq 0$, controls the degree of non-uniformity of the grid. When $\tilde{\alpha}$ is small, we obtain a highly non-uniform grid with more points concentrated around $X_0$, whereas the grid is uniform for high values of $\tilde{\alpha}$.
The choices for the two boundary states and the non-uniformity parameter $\tilde{\alpha}_v$ (resp. $\tilde{\alpha}_X$) for $V^{(m)}$ (resp. $X^{(m,N)}$) are discussed in more details below.
\begin{remark}\label{remarkGripPoint}
	When the initial values of the auxiliary and variance processes are not in the grid,  they can be inserted (see for example \cite{cui2019continuous}, Section 2.3 for details), or the value of the variable annuity must be interpolated between the appropriate grid points.
\end{remark}

Non-uniforms schemes have been used frequently in the literature for options pricing via CTMC approximation methods, see \cite{mijatovic2013continuously}, \cite{lo2014improved}, \cite{cai2015general}, \cite{kirkby2017unified}, \cite{leitao2019ctmc},\cite{cui2018general}, \cite{cui2019continuous} and \cite{MaCTMCAm}, among others. It has also been used for solving partial differential equations (PDEs) numerically, see for instance \cite{tavellapricing}. 
\cite{lo2014improved} showed numerically that non-uniform grids are more stable and can converge faster than uniform grids. 
In the equidistant grid setting, the choice of boundary values significantly impacts the convergence to the true option price.
Also, appropriate bounds depend on the parameters of the underlying process and the maturity of the option, and thus, may be difficult to find. \cite{leitao2019ctmc} tested four different schemes to approximate the volatility process in the Heston model. 
They conclude that the grid proposed by Tavella and Randall \eqref{eqGridTavRand} is the most robust and precise in general, whereas the method proposed by \cite{mijatovic2013continuously}, Section 4, requires few Markov states to achieve good accuracy.
For a deep analysis of grid designs and how they can affect convergence, the reader is referred to \cite{zhang2019analysis}.\\
\\
Unless stated otherwise, all numerical experiments are performed using the CTMC parameters listed in Table \ref{tblCTMCparam}. Recall that $m$ is the number of grid points for the variance process whereas $N$ represents the number of grid points of the fund process. $\tilde{\alpha}_v$ (resp. $\tilde{\alpha}_X$) is the grid non-uniformity parameter of the variance (resp. the auxiliary) process. The grid's upper and lower bounds are respectively $v_1$ and $v_m$ for the variance process and $x_1$ and $x_N$ for the auxiliary process with $X_0= \ln(F_0)-\rho \gamma(V_0)$. The values of $V_0$ and $X_0$ are inserted in their respective grid as per Remark \ref{remarkGripPoint}. Finally, we use $M=500\times 10$ times steps, which corresponds to a computing frequency of approximately twice daily (since there are approximately $250$ trading days per year). 
\begin{table}[h!]
	\begin{tabular}{cccccccccc}
		\hline
		 Parameter & ${m}$ & ${N}$ & ${v_1}$ & ${v_m}$ & $\tilde{\alpha}_v$& $x_1$ & $x_N$ & $\tilde{\alpha}_X$ & $M$\\
		Value  &$50$ & $2,000$ & $V_0/100$ & $7 v_0$ & $0.6571$ & $X_0/10^6$ & $1.95 X_0$ &$2/100$ &$5,000$  \\
		\hline
	\end{tabular}
	\caption{CTMC Parameters}
	\label{tblCTMCparam}
\end{table}

Note that, under the Heston model, good approximations of the transition density of the variance process can be obtained with a small number of grid points, see for example \cite{cui2019continuous} Figure 3. \\
All the numerical experiments are carried out with Matlab R2015a on a Core i7 desktop with 16GM RAM and speed 2.40 GHz.

\subsection{Efficiency of the Fast Algorithms}

The valuation of options (or variable annuities) using CTMC requires the calculation of a matrix exponential to obtain the  transition probability matrix. 
In Algorithm  \ref{algoVAsurrenders}, the two-dimensional process is mapped onto a one-dimensional process resulting in a generator of size $m N\times m N$. 
Thus, we need to calculate the exponential of a $m N\times m N$ matrix to obtain the transition probability matrix. For large values of $mN$, this procedure might stretch computing resources to unacceptable levels. 
Algorithms  \ref{algoVA_WOsurrendersFast} and \ref{algoVA_WITHsurrendersFast}, proposed in Section \ref{subsecVApricingCTMC}, require $m$ times the calculation of the exponential of a $N\times N$ matrix and one time the exponential of a $m\times m$ matrix.
As demonstrated below, reducing the size of the matrix in the exponent allows to significantly increase the efficiency of the procedure.\\
\\
 When the size of the exponent in the matrix exponential is greater than $200\times 200$, we observe that the function \textit{fastExpm} for Matlab, see \cite{matlabFastExpm}\footnote{The function \textit{fastExpm} is based on \cite{hogben2011spinach} and \cite{kuprov2011diagonalization}.}, which is designed for the fast calculation of matrix exponentials of large sparse matrices,  can further speed up the calculation. Combining the function fastExpm and the Fast Algorithms can speed up the code by up to $100$ times (for European and Bermudan contracts).
For the European contract, the Fast Algorithm can decrease the running time by up to $12$ times whereas the function \textit{fastExpm}, by up to $7$ times. 
For the Bermudan contract, the computation time is decreased by up to $4$ times with the Fast Algorithm and by up to $40$ times with the function \textit{fastExpm}. 
When $m=50$ and $N=100$ the running time is approximately 6 seconds for both the European and the Bermudan contracts, confirming the high efficiency of the new algorithms. 
Figure \ref{figCalcTimeAlgoFast} illustrates the computation time in seconds of the ``Fast Algorithm'', Algorithm \ref{algoVA_WOsurrendersFast} for the European contract and Algorithm \ref{algoVA_WITHsurrendersFast} for the Bermudan contract, combined with the function \textit{fastExpm} of \cite{matlabFastExpm} (when the size of the generator in the matrix exponential is greater than $200\times 200$), and the computation time of the ``Regular Algorithm'' for the European contract \eqref{eq_CTMC_VA} and Algorithm \ref{algoVAsurrenders} for the Bermudan one. The running time in Figure \ref{figCalcTimeAlgoFast} are recorded using market, VA and CTMC parameters of Subsection \ref{subsecMarketVACTMCparam}; except for the number of grid points for the auxiliary process $N$ which is set to $N=100, 200,300$ and $500$, respectively. Moreover, we use a constant fee structure, that is we fix $c(x,y)=1.5338\%$ for all $(x,y)\in \reals_+\times\mathcal{S}_V$.\\
\begin{figure}[!h]
	\centering
	\begin{tabular}{cc}
		\includegraphics[scale=0.4]{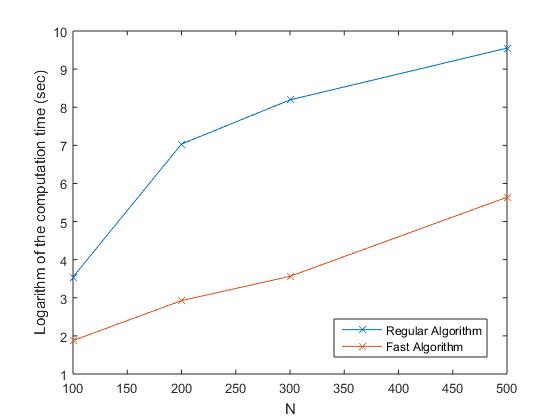}  & \includegraphics[scale=0.4]{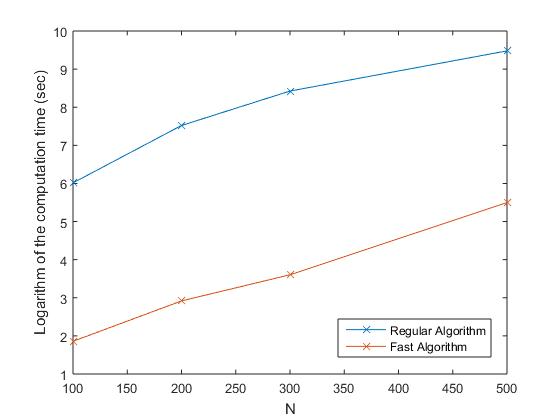}   \\
		(a) European contract & (b) Bermudan contract\\
		\phantom{(a) }Computation time (in seconds) &	\phantom{(b) } Computation time (in seconds)
	\end{tabular}
	\caption{\small{Fast Algorithms: Algorithm \ref{algoVA_WOsurrendersFast} for the European contract and Algorithm \ref{algoVA_WITHsurrendersFast} for the Bermudan contract. Regular Algorithms: Equation \eqref{eq_CTMC_VA} for the European contract and Algorithm \ref{algoVAsurrenders} for the Bermudan contract.}}\label{figCalcTimeAlgoFast}
\end{figure}

We also compare the accuracy of the new algorithms. For comparison, we use the regular algorithms for the European contract and for the Bermudan contract. The absolute difference in the VA prices between the regular and the fast algorithms is around $10^{-3}$ for both European and Bermudan contracts, whereas the relative difference is around $10^{-5}$, confirming the accuracy of the fast algorithms. See also Appendix \eqref{appendixSupplMaterial} available online for more details.
\begin{remark}\label{rmkFastAlgoEfficiency}
	When valuing a European contract, the use of the Expokit of \cite{sidje1998expokit} based on Krylov subspace projection methods accelerates the computation time considerably. Moreover, software packages in Matlab and Fortran can be downloaded for free at \url{https://www.maths.uq.edu.au/expokit}. When $N=500$, we observe that the function \textit{expv} in the Expokit can accelerate the running time by up to $7$ times (compared to the fast algorithm).\\
	However, since the function takes advantage of the product of a matrix exponential with a vector, the fast algorithm is more efficient when valuing Bermudan contract. Indeed, when using the fast algorithm, matrix exponentials are calculated only once at the beginning of the procedure; whereas when using  the Expokit, it needs to be calculated at each time step (since we need to calculate the matrix exponential multiplied by a vector in order to make use of the procedure) which slows down the execution considerably. Hence, when valuing a Bermudan contract, Algorithm \ref{algoVA_WITHsurrendersFast} is up to $9$ times faster than the regular algorithm using Expokit procedures of \cite{sidje1998expokit}.\end{remark}

	As mentioned previously, the computational effort in Algorithms \ref{algoVA_WOsurrendersFast} and  \ref{algoVA_WITHsurrendersFast} resides in the calculation of the matrix exponentials at the beginning of the two procedures. Hence,  once they are (pre-)computed, one can obtain the value of variable annuities with and without surrender rights simultaneously at almost no additional cost. For instance, when $N=500$, we simultaneously obtain the prices of variable annuities with and without surrender rights in $270$ seconds; whereas the values of the Bermudan and the European contracts can be obtained separately in approximately $250$ seconds for each. This also holds true for any other VA contracts with different guaranteed amounts; that is, a variety of contracts can be priced for almost the same computational effort as a single contract.

\subsection{Fee Structures and Fair Fee Parameters}\label{subsectFeeStruct}
First, recall from Subsection \ref{subsectionCTMCapproxVIX} that 
$$\vix_t^2=\e_t\left[\frac{1}{\tau}\int_t^{t+\tau}\sigma^2_S(V_{s})\diff s\right],\quad\quad \text{with $\tau=30/365$.}$$ 
\\

For all the numerical experiments conducted below, we consider three types of VIX-linked fee structures. 
As in \cite{cui2017}, we use an uncapped $\vix^2$-linked fee structure (the ``Uncapped $\vix^2$'' fee structure) of the form 
$$c_t=\tilde{c}+\tilde{m}\vix^2_t,\,0\leq t\leq T,$$
with $\tilde{c}, \tilde{m}\geq 0$.\\

When the volatility is high (e.g. financial turmoil), the Uncapped $\vix^2$ fee rate can become excessive (see Table \ref{tblFeeExample} for examples). Thus, as mentioned in \cite{cui2017}, it may be interesting to introduce a cap to the actual $\vix^2$-linked fee structure in order to keep the product marketable and to reduce the negative net return on the VA account resulting from a high fee. This motivates our choice of imposing a cap to this type of fee.\\
\\
Thus, for the second fee structure, we consider capped $\vix^2$-linked fees (or simply the ``Capped $\vix^2$'' fee structure or Capped fee structure) of the form 
$$c_t=\min(\tilde{c}+\tilde{m}\vix^2_t, K),\,0\leq t\leq T,$$ 
where $\tilde{c},\tilde{m}\geq 0$ and $K>0$.\\
\\
Finally, we look at an uncapped fee structure linked to the VIX (rather than the VIX squared), also called the ``Unccaped VIX'' fee structure. Such a structure can help to keep the fee rates reasonable during high volatility periods. More precisely, we consider a third fee process of the form
$$c_t=\tilde{c}+\tilde{m}\vix_t,\,0\leq t\leq T,$$
with $\tilde{c}, \tilde{m}\geq 0.$ \\
\\
Since the volatility process is Markovian, the volatility index at $t$ will depend only on $V_t$. 
Therefore, the three fee functions are of the form $c(y)$, $y\in\mathcal{S}_V$.
\\
Analytical formulas for the $\vix$ are not known for all the models stated in Table \ref{tblExSVmodels}. 
In that case, CTMCs can be used to approximate the value of the volatility index as discussed in Subsection \ref{subsectionCTMCapproxVIX}. 
\\
The fee structure parameters $\tilde{c}$ and $\tilde{m}$ are set in such a way that the contract is fair at inception. That is, the initial amount invested by the policyholder, $F_0$, is equal to the expected discounted value of the future benefit (without early surrenders). Those parameters are called the fair parameters and are henceforth denoted by a star: $(\tilde{c}^*, \tilde{m}^*)$.  We set the fee in this manner to calculate the value added by the right to surrender. To identify the fair fee vector $(\tilde{c}^*, \tilde{m}^*)$, we first fix a multiplier $\tilde{m}^*$, and then solve for the corresponding fair base fee $\tilde{c}^*$. Note that such a fair base fee $\tilde{c}^*$ does not exist for all values of $\tilde{m}^*>0$. When $c_t=\tilde{c}^*+\tilde{m}^*\vix^2$, the fair parameters are obtained using the exact formula of \cite{cui2017}. For the other fee structures, the fair parameters are calibrated using a CTMC approximation with $N=100$ (all other CTMC parameters are the same as in Table \ref{tblCTMCparam}), to reduce the computation time. 
Note that very accurate VA prices are obtained extremely fast with $N=100$ under the Heston model, see Appendix \ref{appendixSupplMaterial}, available online, for numerical details. Table \ref{tblFairFeeHeston} presents the fair fee vectors  $(\tilde{c}^*, \tilde{m}^*)$, and Table \ref{tblFeeExample} shows examples of fair fee rates produced by each fair fee vector at different volatility levels ($\sqrt{V_t}$)\footnote{Under the Heston model there is a closed-form expression for the VIX given by $\vix^2_t=B+ A V_t$
	with $A=\frac{1-e^{-\kappa\tau}}{\kappa\tau}$ and $B=\frac{\theta(\kappa\tau-1+e^{-\kappa\tau})}{\kappa\tau}$, see \cite{ZhuZang2007}.}. 
\begin{table}[h!]
	\begin{tabular}{c cccc}
		\hline
		& \multicolumn{4}{c}{$c_t=\tilde{c}^*+\tilde{m}^*\vix^2_t$}\\
		\hline
			$\tilde{m}^*$ & $0.0000$ & $0.1500$ & $0.3000$ & $0.4345$\\
		$\tilde{c}^*$ & $1.5338\%$  & $1.0036\%$ & $0.4741\%$ & $0.000\%$\\
		
		\hline
		\hline
		& \multicolumn{4}{c}{$c_t=\min\{\tilde{c}^*+\tilde{m}^*\vix^2_t,K\}$, $K=2\%$}\\
		\hline
		$\tilde{m}^*=$ & $0.0000$ & $0.1500$ & $0.3000$ & $0.4927$\\
		$\tilde{c}^*\,\,=$ & $1.5338\%$  & $1.0112\%$ & $0.5415\%$ & $0.000\%$\\
		\hline
		\hline
		& \multicolumn{4}{c}{$c_t=\tilde{c}^*+\tilde{m}^*\vix_t$}\\
		\hline
		$\tilde{m}^*=$ & $0.0000$ & $0.0250$ & $0.0500$ & $0.0836$\\

	    $\tilde{c}^*\,\,=$ & $1.5338\%$  & $1.0750\%$ & $0.6164\%$ & $0.000\%$\\
		\hline
	\end{tabular}
    \caption{Fair fee vectors $(\tilde{c}^*,\tilde{m}^*)$}
    \label{tblFairFeeHeston}
\end{table} 
\begin{table}[h!]
   \begin{subtable}[c]{0.5\linewidth}
	\begin{tabular}{|c|cccc|}
		\hline
		$\mathbf{\sqrt{V_t}}$ (\%)$\backslash\mathbf{\tilde{m}^*}$ &  $\mathbf{0.0000}$ & $\mathbf{0.1500}$ & $\mathbf{0.3000}$ & $\mathbf{0.4345}$\\
		\hline
		\hline
		$\mathbf{8.702}$ & $1.5338$ & $1.1550$ & $0.7770$ & $0.4387$\\
		$\mathbf{13.882}$ & $1.5338$ & $1.3168$  & $1.1007$ & $0.9075$\\
		$\mathbf{30.434}$ & $1.5338$ & $2.3314$  & $3.1299$ & $3.8464$\\
		$\mathbf{42.772}$  & $1.5338$  & $3.5808$  & $5.6285$ & $7.4653$\\
		\hline
	\end{tabular}
	\subcaption{$c_t=\tilde{c}^*+\tilde{m}^*\vix_t^2$ (\%)}
	\end{subtable}
	\begin{subtable}[c]{0.5\linewidth}
			\begin{tabular}{|c|cccc|}
			\hline
			$\mathbf{\sqrt{V_t}}$ (\%)$\backslash\mathbf{\tilde{m}^* }$ &  $\mathbf{0.0000}$ & $\mathbf{0.1500}$ & $\mathbf{0.3000}$ & $\mathbf{0.4927}$\\
			\hline
			\hline
			$\mathbf{8.702}$ & $1.5338$ & $1.1627$ & $0.8444$ & $0.4975$\\
			$\mathbf{13.882}$ & $1.5338$ & $1.3245$  & $1.1681$ & $1.0291$\\
			$\mathbf{30.434}$ & $1.5338$ & $2.0000$  & $2.0000$ & $2.0000$\\
			$\mathbf{42.772}$  & $1.5338$  & $2.0000$  & $2.0000$ & $2.0000$\\
			\hline
		\end{tabular}
		\subcaption{$c_t=\min\{\tilde{c}^*+\tilde{m}^*\vix_t^2 (\%),K\}$, $K=2\%$}
	\end{subtable}
    \begin{subtable}[c]{0.5\linewidth}
    \begin{tabular}{|c|cccc|}
    	\hline
    	$\mathbf{\sqrt{y}}$ (\%)$\backslash\mathbf{\tilde{m}^*}$ &  $\mathbf{0.0000}$ & $\mathbf{0.0250}$ & $\mathbf{0.0500}$ & $\mathbf{0.0836}$\\
    	\hline
    	\hline
    	$\mathbf{8.702}$ & $1.5338$ & $1.3262$ & $1.1188$ & $0.8402$\\
    	$\mathbf{13.882}$ & $1.5338$ & $1.4363$ & $1.3390$  & $1.2083$\\
    	$\mathbf{30.434}$ & $1.5338$ & $1.8188$ & $2.1041$  & $2.4877$ \\
    	$\mathbf{42.772}$  & $1.5338$  & $2.1113$  & $2.6889$ & $3.4657$\\
    	\hline
    \end{tabular}
    	\subcaption{$c_t=\tilde{c}^*+\tilde{m}^*\vix_t$}
    \end{subtable}
  \caption{Fair fee rates in \% }
	\label{tblFeeExample}
\end{table}

The uncapped $\vix^2$-linked fee rate can get very high as the volatility increases, reaching levels as high as $7.4653\%$ when the volatility is $42.772\%$. During the last COVID-19 financial crisis, volatility reached levels as high as $80\%$ in March 2020\footnote{See VIX historical data at \url{https://www.cboe.com/tradable_products/vix/vix_historical_data/}.}. 
This motivates the two other fee structures we propose. 
The introduction of a cap does not significantly affect the calibrated fair fee parameters. 
For instance, when $\tilde{m}^*=0.3$, the calibrated base fee rate jumps slightly from $\tilde{c}^*=0.4741\%$ to $\tilde{c}^*=0.5415\%$ when the cap is included, and when $\tilde{c}^*=0.0000\%$, the fair multiplier goes from $\tilde{m}^*=0.4345$ to $\tilde{m}^*=0.4927$. 
The other fair fee vectors are almost the same under both structures. 
However, when we compare Tables \ref{tblFeeExample} \textsc{(a)} and \textsc{(b)}, we observe that the introduction of such a cap allows to keep the product marketable by avoiding excessive fee rates during financial turmoil. 

Comparing Tables \ref{tblFeeExample} \textsc{(a)} and \textsc{(c)}, we note that fee rates of the VIX fee structure are slightly higher during low/regular volatility periods than the ones of the $\vix^2$ fee structure. 
This slight increase in fee rates in low volatility regimes allows the fees to be kept at reasonable levels during high volatility periods compared to the Uncapped $\vix^2$ fee structure. 
Thus, this third uncapped fee function is also clearly more attractive from the policyholder's perspective.

\subsection{Effect of VIX-Linked Fees on Surrender Incentives}\label{subsectNumAnalHeston}
	Recall that under the Heston model, \cite{heston1993}, the price of the risky asset satisfies
		\begin{equation}
		\begin{array}{ll}
			\diff S_t&=r S_t \diff t +\sqrt{V_t} S_t \diff W^{(1)}_t\\
			\diff V_t&=\kappa(\theta-V_t)\diff t+\sigma\sqrt{V_t} \diff W^{(2)}_t,
		\end{array}\label{eqEDSheston}
	\end{equation}
     with $S_0$ and $V_0$ are deterministic, and where $W=(W^{(1)}, W^{(2)})^T$ is a bi-dimensional correlated Brownian motion under  $\prob{Q}$ and such that $[W^{(1)},\,W^{(2)}]_t=\rho t$ with $\rho\in[-1, \, 1]$, the speed of the mean-reversion $\kappa>0$, the long term variance $\theta>0$ and the  volatility of the variance $\sigma>0$ (also called the volatility of the volatility).\\
     \\
     Moreover, when the market is modeled by \eqref{eqEDSheston}, the VIX has a closed-form expression given by
     \begin{equation}
     	\vix^2_t=B+ A V_t\label{eqVIX2heston}
     \end{equation}
     with $A=\frac{1-e^{-\kappa\tau}}{\kappa\tau}$ and $B=\frac{\theta(\kappa\tau-1+e^{-\kappa\tau})}{\kappa\tau}$, see \cite{ZhuZang2007} for details.\\
     \\
     The three fee structures exposed in the Subsection \ref{subsectFeeStruct} can thus be obtained explicitly in terms of the current volatility using \eqref{eqVIX2heston} as follows:
     \begin{table}[h!]
     	\begin{tabular}{lc}
     		\hline
     		\textbf{Fee Structure} & $c_t$, $0\leq t\leq T$\\
     		\hline\hline
     		\textbf{Uncapped $\vix^2$}  & $c_t=\tilde{c}^*+\tilde{m}^*(A+BV_t)$\\
     		\textbf{Capped $\vix^2$}  & $c_t=\min\{K,\tilde{c}^*+\tilde{m}^*(A+BV_t)\}$\\
     	\textbf{Uncapped $\vix$}  & $c_t=\tilde{c}^*+\tilde{m}^*\sqrt{(A+BV_t)}$\\
     	\hline
     \end{tabular}
     \caption{Fair fee process under the Heston model}
     \label{tblFeeProcessHeston}
 \end{table}

   Now from Lemma \ref{lemmaDecoupleBM}, we find that $\gamma(x)=x/\sigma$. Thus, given a certain fee process $c_t$ (listed in Table \ref{tblFeeProcessHeston}), the dynamics of the auxiliary process can be derived as 
     \begin{equation}
     	\begin{array}{ll}
     		\diff X_t  &= \mu_X(X_t,V_t) \diff t+\sigma_X(V_t)\diff W^{*}_t,\\
     		\diff V_t&=\mu_V(V_t)\diff t+\sigma_V(V_t)\diff W^{(2)}_t,\label{eqEDS_X_Qheston}
     	\end{array}
     \end{equation} 
     where $\mu_X(X_t,V_t)=r-\frac{\rho\kappa\theta}{\sigma}- c_t+V_t\left(\frac{\rho\kappa}{\sigma}-\frac{1}{2}\right)$, and $\sigma_X(V_t)=\sqrt{(1-\rho^2)V_t}$, $0\leq t\leq T$.\\
\\
  Using the CTMC technique outlined in Section \ref{sectionCTMCapprox} and the market, VA, and CTMC parameters of subsection \ref{subsecMarketVACTMCparam}, we perform the valuation of a variable annuity with and without early surrenders ( ``VA with ES'' and `` VA without ES'', respectively). The results are reported in Table \ref{tblVApricesHeston} below\footnote{Numerical experiments under the Heston model have been performed using Equation \eqref{eq_CTMC_VA}, and Algorithms  \ref{algoVAsurrenders} and \ref{algoOptSurrSurf}. Note however that similar results are obtained when using the Fast Algorithms, see Appendix \ref{appendixSupplMaterial}, available online for details.}.
 \begin{table}[h!]
 \begin{subtable}[c]{0.7\linewidth}
 	\centering
 	\begin{tabular}{l| cccc}
 		\hline
 		$\mathbf{\tilde{m}^*=}$ & $\mathbf{0.0000}$ & $\mathbf{0.1500}$ & $\mathbf{0.3000}$ & $\mathbf{0.4345}$\\
 		$\mathbf{\tilde{c}^*\,\,=}$ & $\mathbf{1.5338\%}$  & $\mathbf{1.0036\%}$ & $\mathbf{0.4741\%}$ & $\mathbf{0.000\%}$\\
 		\hline
 		\hline
 		\textbf{VA without ES} & $100.00090$ & $100.00091$& $100.00092$& $100.00093$\\
 		\textbf{VA with ES} & $103.01785$ & $103.00823$ & $103.00330$ & $103.00367$\\
 		\hline 
 		\textbf{Value of ES} & $3.01695$ & $3.00732$ & $3.00238$ & $3.00274$\\
 		\hline
 	\end{tabular}
 	\subcaption{$c_t=\tilde{c}^*+\tilde{m}^*\vix_t^2$}
\end{subtable}
 \begin{subtable}[c]{0.7\linewidth}
 	\centering
     	\begin{tabular}{l| cccc}
     	\hline
     	$\mathbf{\tilde{m}^*=}$ & $\mathbf{0.0000}$ & $\mathbf{0.1500}$ & $\mathbf{0.3000}$ & $\mathbf{0.4927}$\\
     	$\mathbf{c^*\,\,=}$ & $\mathbf{1.5338\%}$  & $\mathbf{1.0112\%}$ & $\mathbf{0.5415\%}$ & $\mathbf{0.000\%}$\\
     	\hline
     	\hline
     	\textbf{VA without ES} & $100.00090$ & $100.00075$& $100.00070$& $100.00036$\\
     	\textbf{VA with ES} & $103.01785$ & $103.00596$ & $102.99137$ & $102.97560$\\
     	\hline 
     	\textbf{Value of ES} & $3.01695$ & $3.00521$ & $2.99067$ & $2.97524$\\
     	\hline
     \end{tabular}
  	\subcaption{$c_t=\min\{K,\tilde{c}^*+\tilde{m}^*\vix_t^2\}$, with $K=2\%$}
\end{subtable}
 \begin{subtable}[c]{0.7\linewidth}
 	\centering
     \begin{tabular}{l| cccc}
     	\hline
     	$\mathbf{\tilde{m}^*=}$ & $\mathbf{0.0000}$ & $\mathbf{0.0250}$ & $\mathbf{0.0500}$ & $\mathbf{0.0.0836}$\\
     	$\mathbf{\tilde{c}^*\,\,=}$ & $\mathbf{1.5338\%}$  & $\mathbf{1.0750\%}$ & $\mathbf{0.6164\%}$ & $\mathbf{0.000\%}$\\
     	\hline
     	\hline
     	\textbf{VA without ES} & $100.00090$ & $100.00099$& $100.00057$& $100.000167$\\
     	\textbf{VA with ES} & $103.01785$ & $103.01080$ & $103.00446$ & $102.99853$\\
     	\hline 
     	\textbf{Value of ES} & $3.01695$ & $3.00981$ & $3.00389$ & $2.99686$\\
     	\hline
     \end{tabular}
  	\subcaption{$c_t=\tilde{c}^*+\tilde{m}^*\vix_t$}
\end{subtable}
 	\caption{Variable annuity with and without early surrender (ES).}
 	\label{tblVApricesHeston}
 \end{table} 
\\
First, we observe that the fair value of the variable annuity without early surrenders is approximately $F_0=100$ for all fair fee vectors. 
This is because fair fee parameters are calibrated such that the value of the VA without surrender rights at time $t=0$ is equal to the initial premium ($F_0=100$). 
Moreover, under the Uncapped $\vix^2$-linked structure, fair fee parameters are obtained using the exact pricing formula of \cite{cui2017}, confirming the accuracy of the approximated model.  
The absolute error is around $10^{-4}$ for all fee vectors for this fee structure. 
An accuracy around  $10^{-3}$ can be obtained for each value (the value of the VA with and without surrender rights) with fewer grid points with significantly less computational effort.
Detailed results with $N=100$ and $N=1,000$ are given in Appendix \ref{appendixSupplMaterial} (available online as supplemental material) with their respective computation time. The value of the surrender right is calculated as the difference between the values of the variable annuity with and without early surrenders. \\
\\
As $\tilde{m}^*$ increases, the risk-neutral value of early surrenders remains almost the same for all fee structures. We were expecting a VIX-linked fee structure to reduce the risk-neutral value of the surrender rights since this type of fee structure realigns income and liability, \cite{cui2017}, but this is not what we observe here. 
However, VIX-linked fee structures have an impact on optimal surrender strategies, as shown below.\\
\\
In Figure \ref{figBoundary3D}, the shape of the approximated optimal surrender surface associated to each of the VIX-linked fee structure is illustrated for different values of the fair multiplier. We observe that the surrender region is of threshold type for all $\tilde{m}^*$. We also note that, regardless of the value of $\tilde{m}^*$, the boundary is a concave function in the time variable. It slowly increases to a maximum and decreases rapidly to the guarantee level $G$. This is consistent with earlier findings of \cite{bernard2014optimal} and \cite{kang2018optimal}.
\begin{figure}[!t]
	\centering
	\begin{tabular}{cccc}
		\multicolumn{4}{c}{\small{\textsc{(a)} $c_t=\tilde{c}^*+\tilde{m}^*\vix^2_t$}}\\
		\includegraphics[scale=0.2]{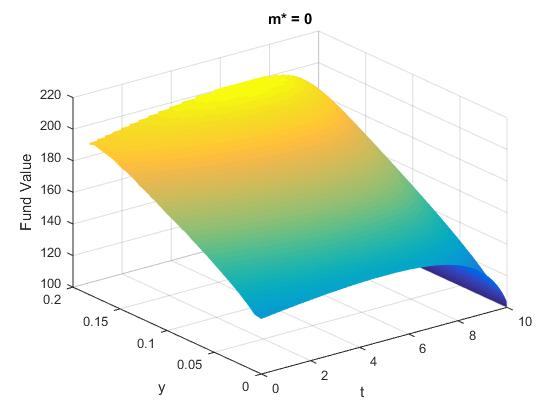}  & \includegraphics[scale=0.2]{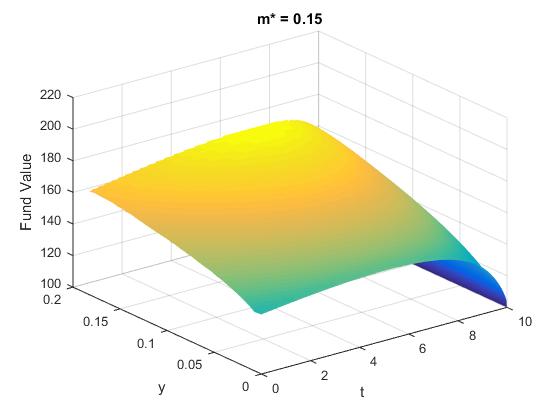}   & \includegraphics[scale=0.2]{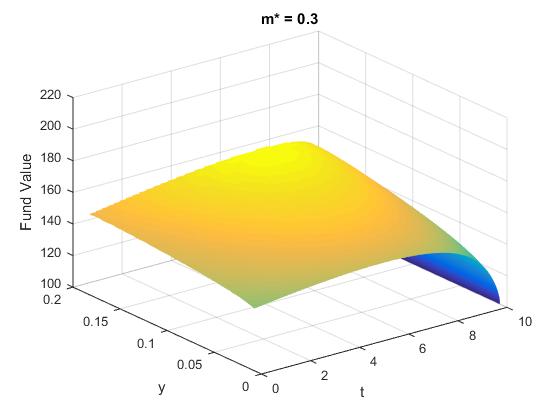}  & \includegraphics[scale=0.2]{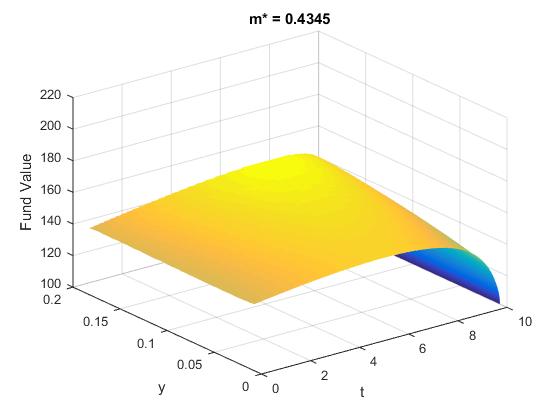} \\
		\multicolumn{4}{c}{\small{\textsc{(b)} $c_t=\min\{\tilde{c}^*+\tilde{m}^*\vix^2_t,K\}$, $K=2\%$}}\\
		\includegraphics[scale=0.2]{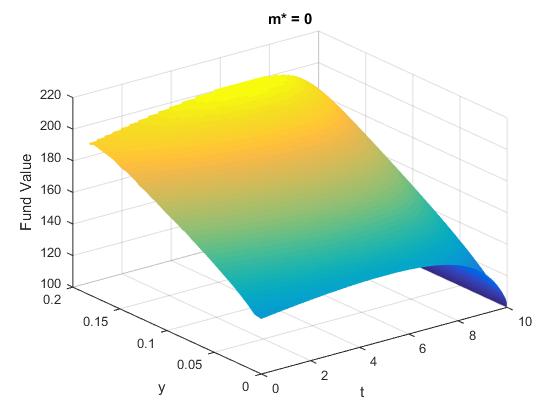} & \includegraphics[scale=0.2]{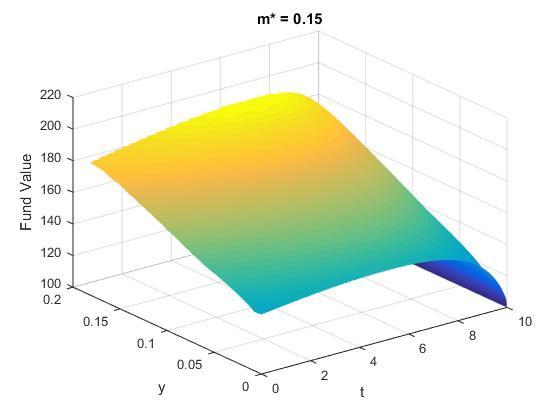}&	\includegraphics[scale=0.2]{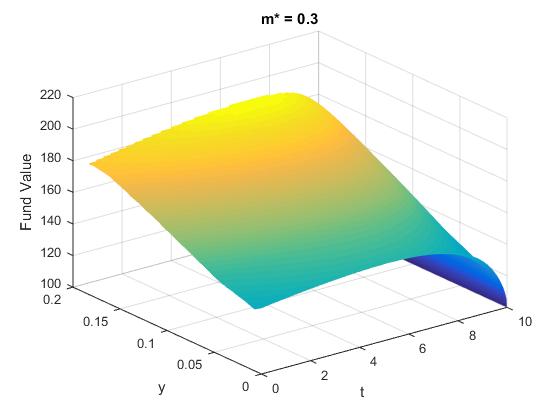} &	\includegraphics[scale=0.2]{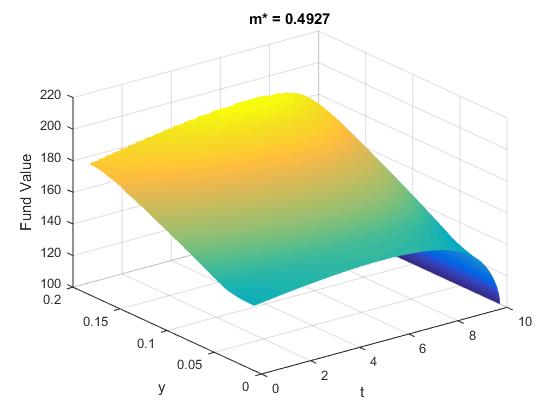} \\
		\multicolumn{4}{c}{\small{\textsc{(c)} $c_t=\tilde{c}^*+\tilde{m}^*\vix_t$}}\\
		\includegraphics[scale=0.2]{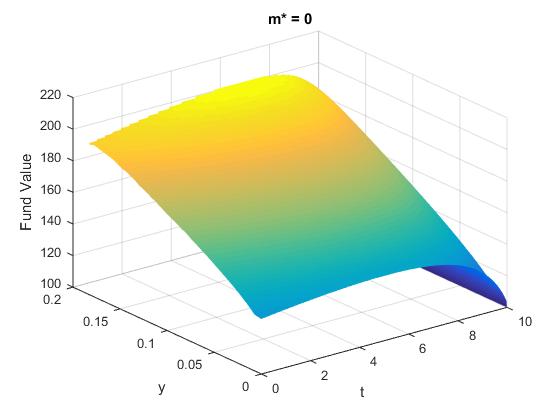} & \includegraphics[scale=0.2]{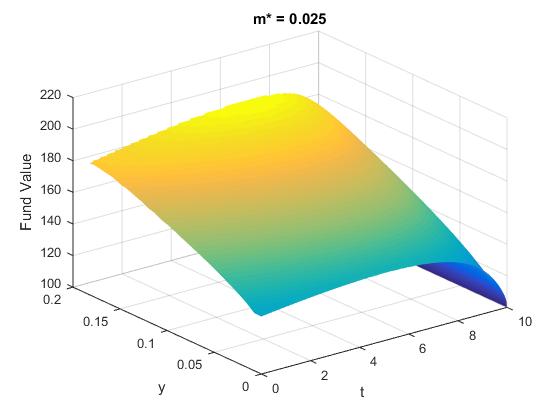}  & \includegraphics[scale=0.2]{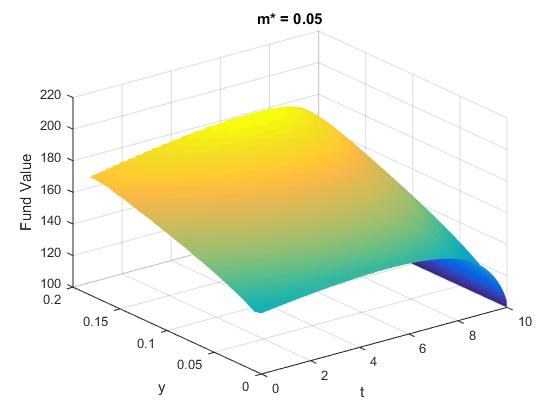}  &
		\includegraphics[scale=0.2]{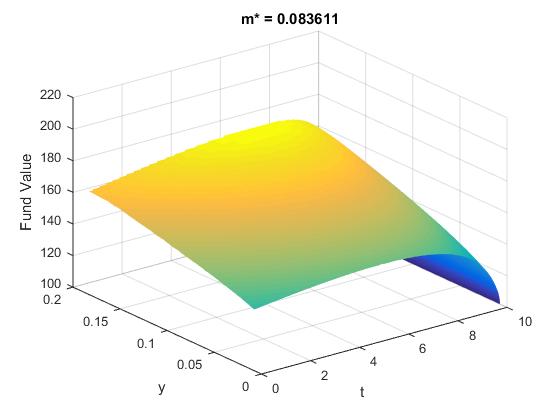}  
	\end{tabular}
	\caption{\small{Approximated optimal surrender surface of VIX-linked fees VAs for different values of fair multiplier $\tilde{m}^*$ under the Heston model. The $x$-axis represents the time and the $y$-axis the variance.}}\label{figBoundary3D}
\end{figure}

Now fix $t\in[0,T]$ and observe, in Figures \ref{figBoundary3D} and \ref{figBoundary}, that the $t$ section of $f^{(m,N)}$, the function $y\mapsto f^{(m,N)}(t,y)$, denoted by $f^{(m,N)}_t$,  is increasing for all $\tilde{m}^*$.  
Hence, when the volatility is high, variable annuities are surrendered for higher fund values in contrast to low volatility levels. 
This is in line with findings of \cite{kang2018optimal}.
However, as $\tilde{m}^*$ increases, the function $f^{(m,N)}_t$ increases more slowly (and particularly for the uncapped structures, panel \textsc{(a)} and \textsc{(c)} of Figures \ref{figBoundary3D}  and \ref{figBoundary}). 
For instance, fix $\tilde{m}^*\in\{0,0.15,0.3,0.4345\}$, $y\in\mathcal{S}^{(m)}_V$ and note, from Figure \ref{figBoundary} (panel \textsc{(a)} or \textsc{(c)}), that the $y$ section of $f^{(m,N)}$, the function $t\mapsto f^{(m,N)}(t,y)$, denoted by $f^{(m,N)}_y$, is pushed upwards as $y$ increases. 
However, the difference between the low and the high volatility $y$ section is less significant as $\tilde{m}^*$ grows. 
This means that the optimal surrender decision for uncapped VIX-linked fees is less sensitive to volatility fluctuations when fees are tied to the volatility index. 
When we look at the capped structure (panel \textsc{(b)} of Figures \ref{figBoundary3D} and \ref{figBoundary}), we also note that approximated optimal surrender surfaces are gradually increasing as volatility grows; however, they now increase at essentially almost all the same pace for all $\tilde{m}^*$.  This can be interpreted from a financial perspective as pointed out below.
\begin{figure}[!t]
	\centering
	\begin{tabular}{cccc}
		\multicolumn{4}{c}{\small{\textsc{(a)} $c_t=\tilde{c}^*+\tilde{m}^*\vix^2_t$}}\\
		\includegraphics[scale=0.2]{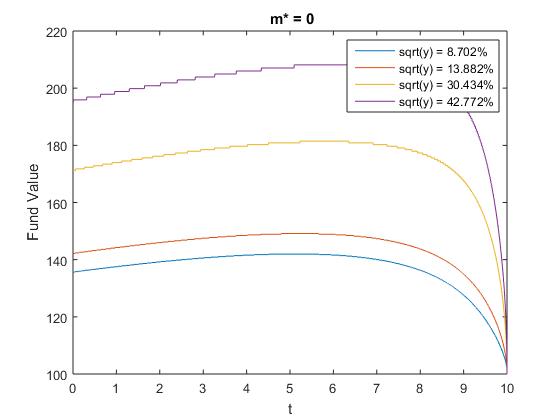}   &\includegraphics[scale=0.2]{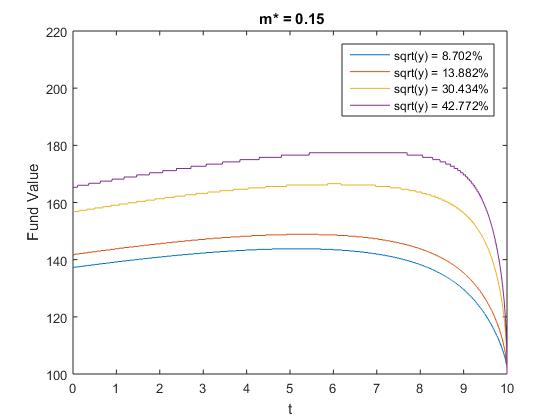}   & \includegraphics[scale=0.2]{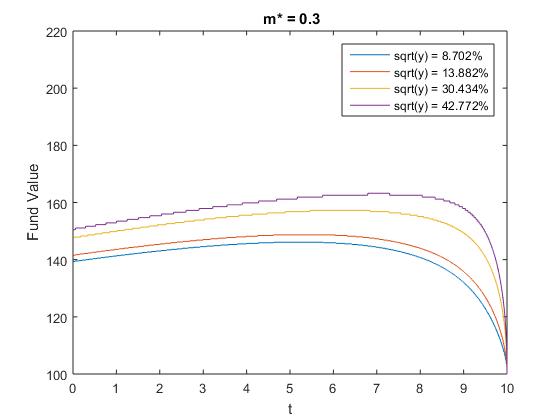}  & \includegraphics[scale=0.2]{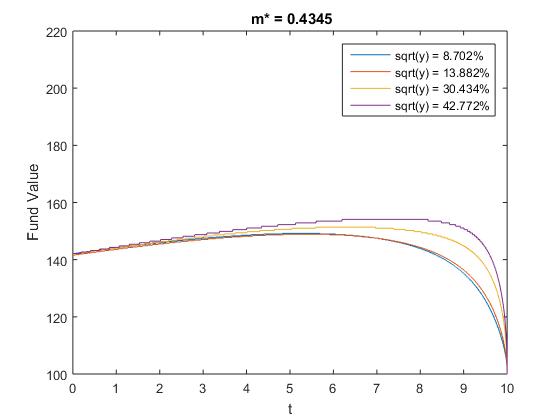} \\
		\multicolumn{4}{c}{\small{\textsc{(b)} $c_t=\min\{\tilde{c}^*+\tilde{m}^*\vix^2_t,K\}$, $K=2\%$}}\\
		\includegraphics[scale=0.2]{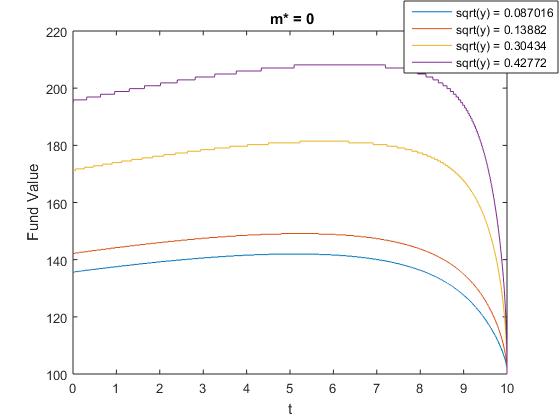}  & 	\includegraphics[scale=0.2]{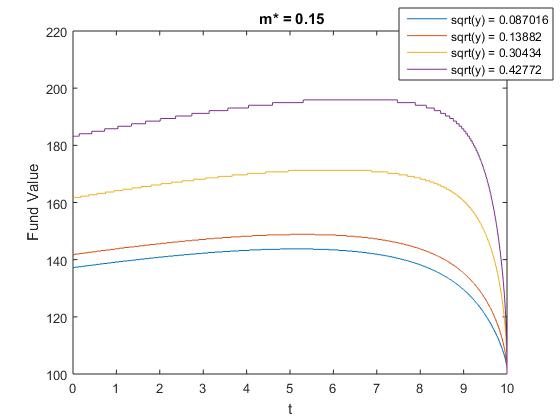} &		\includegraphics[scale=0.2]{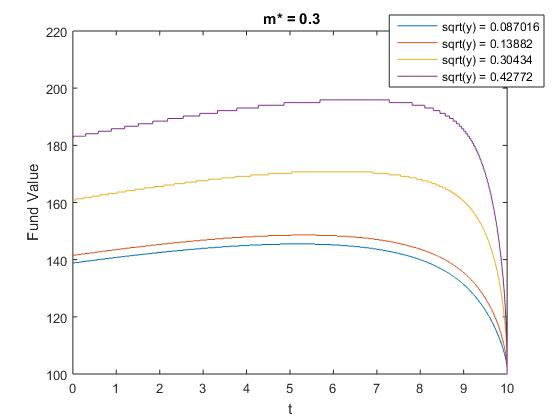}   &		\includegraphics[scale=0.2]{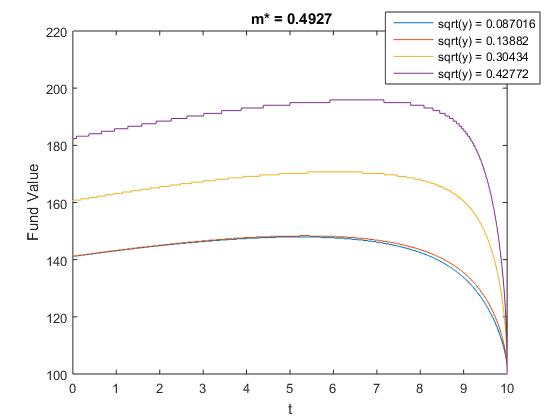}   \\
		\multicolumn{4}{c}{\small{\textsc{(c)} $c_t=\tilde{c}^*+\tilde{m}^*\vix_t$}}\\
		\includegraphics[scale=0.2]{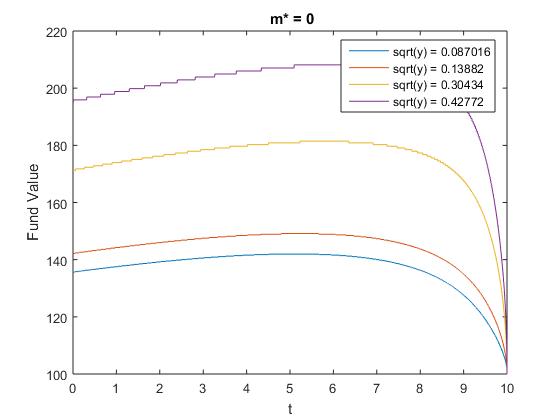}   &\includegraphics[scale=0.2]{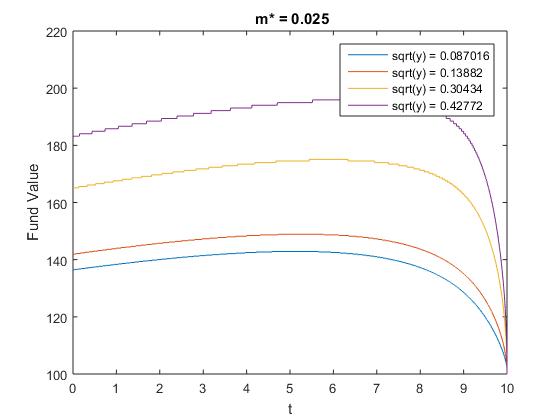}    &\includegraphics[scale=0.2]{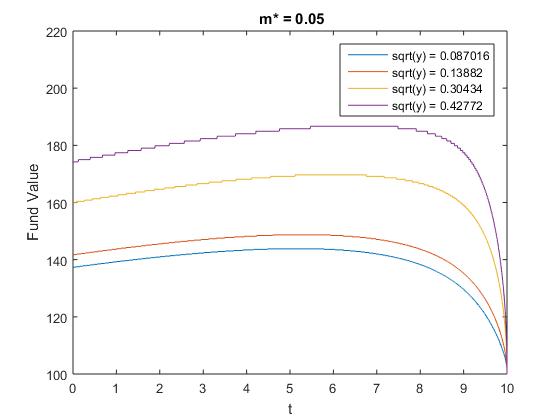}    &
		\includegraphics[scale=0.2]{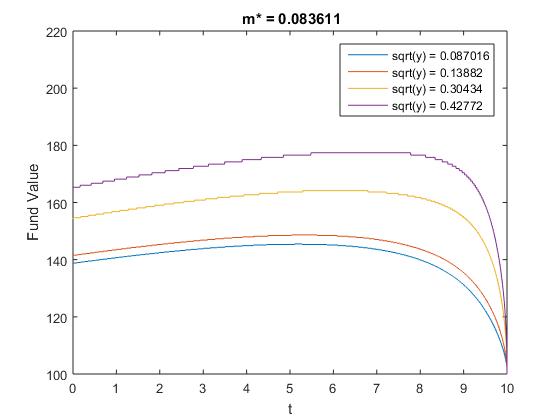}   
	\end{tabular}
	\caption{\small{ The $y$ section of the approximated optimal surrender surface, $f^{(m,N)}_y$, for different volatility levels $\sqrt{y}$ and fair multipliers $\tilde{m}^*$.}}\label{figBoundary}
\end{figure}

In Figure \ref{figBoundaryAll}, we fix a volatility level $y$ and we compare the $y$ sections of $f^{(m,N)}$, $f_y^{(m,N)}$, for different fair multipliers $\tilde{m}^*$. 
When the volatility is low ($\sqrt{y}=8.702\%$), see for instance the first graph of Figure \ref{figBoundaryAll} \textsc{(a)}, we observe that $f_y^{(m,N)}$ is pushed upwards with increasing values of $\tilde{m}^*$. 
This means that a variable annuity contract with a fully dependent uncapped $\vix^2$-linked fee structure ($\tilde{m}^*=0.4345$) is surrendered at higher fund values than the constant fee one ($\tilde{m}^*=0$). 
Indeed, when the volatility is low, the fee paid under a variable annuity contract with a $\vix^2$-linked fee structure is also low (see Table \ref{tblFeeExample} for examples), making $\vix$-linked fee contracts more attractive than the constant fee ones. 
However, as the volatility rises, $\vix^2$-linked fees also rise and so, the relation between the optimal surrender decision and $\tilde{m}^*$ reverts. 
The second graph of Figure \ref{figBoundaryAll} \textsc{(a)} shows that variable annuity contracts are surrendered at almost all the same fund value levels when the volatility equals to $13.882\%$, regardless of $m^*$. 
The latter may be explained by the fact that fee rates are all around the same level when $\sqrt{y}=13.882\%$, that is $\pm 1\%$ as per Table \ref{tblFeeExample}. 
However, when the volatility increases (see for instance the third and the last graphs of Figure \ref{figBoundaryAll} \textsc{(a)}), $\vix^2$-linked fees are also high (refer again to Table \ref{tblFeeExample} for examples), making $\vix^2$-linked fee variable annuity contracts less attractive than the constant fee ones. 
And so, when volatility is high, we observe that $f_y^{(m,N)}$ is pushed downward with increasing values of $\tilde{m}^*$. 
Or, in other words, for high volatility levels, variable annuity contracts with $\vix^2$-linked fee structures are surrendered for fund values that are less than the ones with constant fee structures. 
This is financially intuitive, as pointed out by \cite{bernard2014optimal} under the Black-Scholes setting with a constant fee function, since when the fee gets higher, the policyholder has to pay more for the guarantee, and so, the mismatch between the premium for guarantee and its value is even greater; resulting in earlier exercise time. 
The analysis above shows that the findings of \cite{bernard2014optimal} also extends to stochastic fee structures. 
Similar conclusion can be drawn for the Uncapped VIX-linked fees, Figure \ref{figBoundaryAll}\textsc{(c)}.\\
\begin{figure}[!t]
	\centering
	\begin{tabular}{cc}
		\small{\textsc{(a)} $c_t=\tilde{c}^*+\tilde{m}^*\vix^2_t$} & \small{\textsc{(b)} $c_t=\min\{\tilde{c}^*+\tilde{m}^*\vix^2_t,K\}$, $K=2\%$} \\
		\includegraphics[scale=0.15]{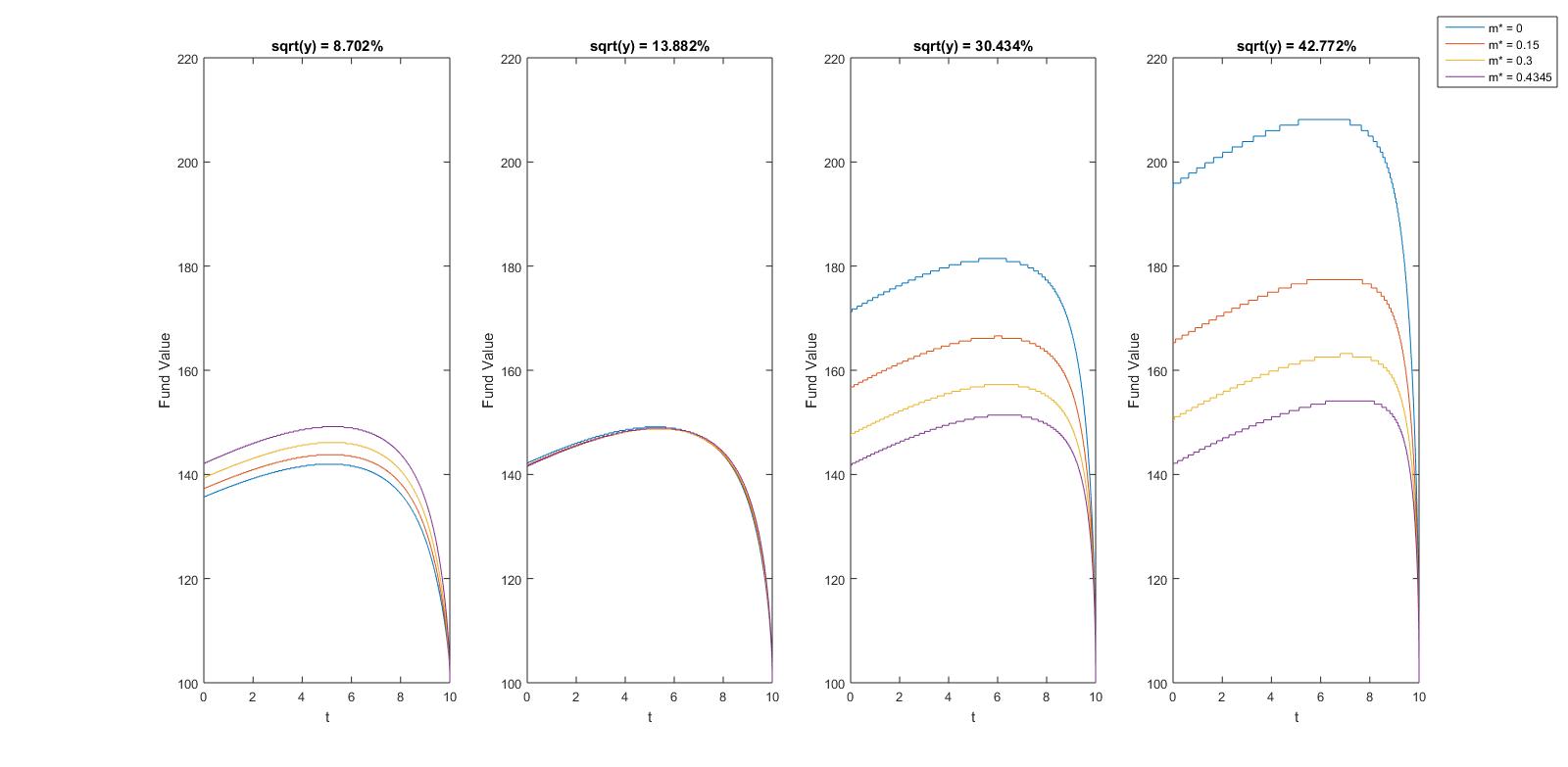}   &\includegraphics[scale=0.15]{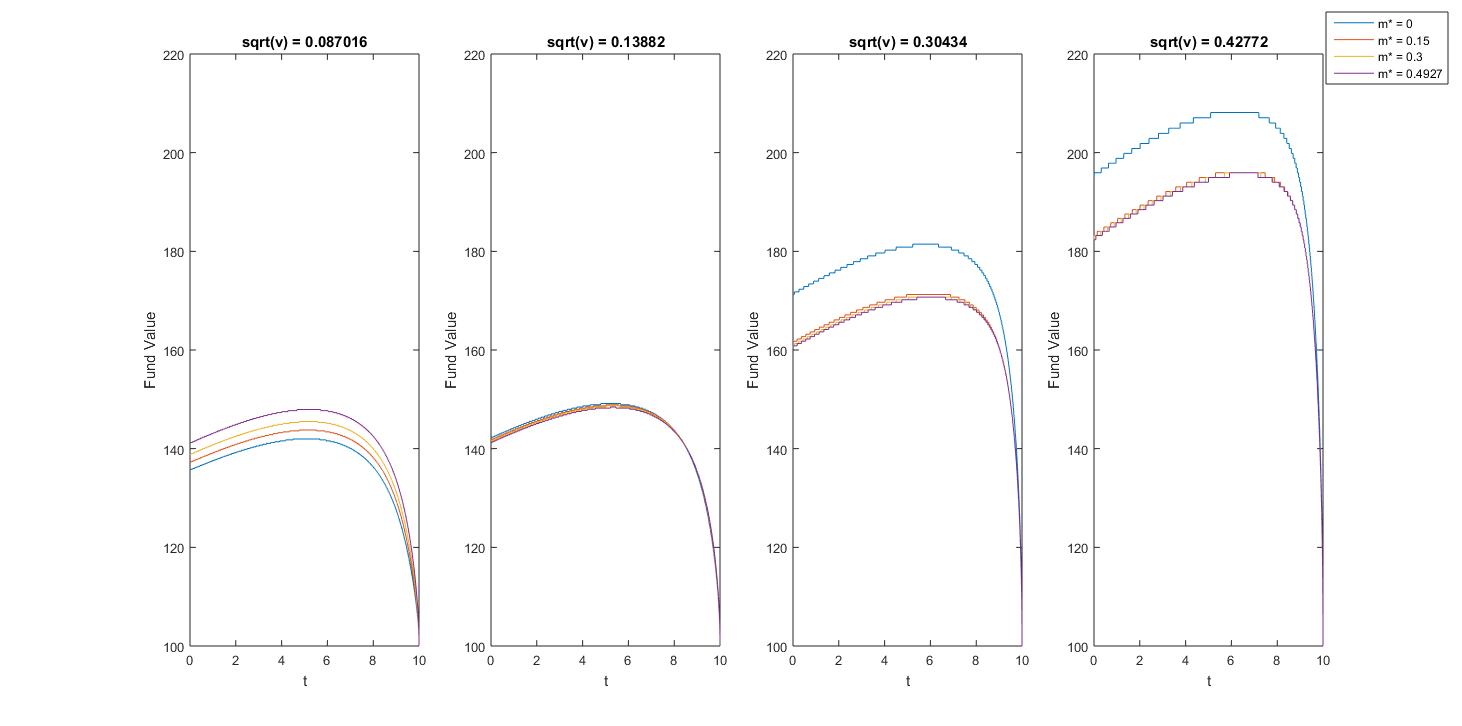}  \\
		\multicolumn{2}{c}{\small{\textsc{(c)} $c_t=\tilde{c}^*+\tilde{m}^*\vix_t$}}\\
		\multicolumn{2}{c}{\includegraphics[scale=0.15]{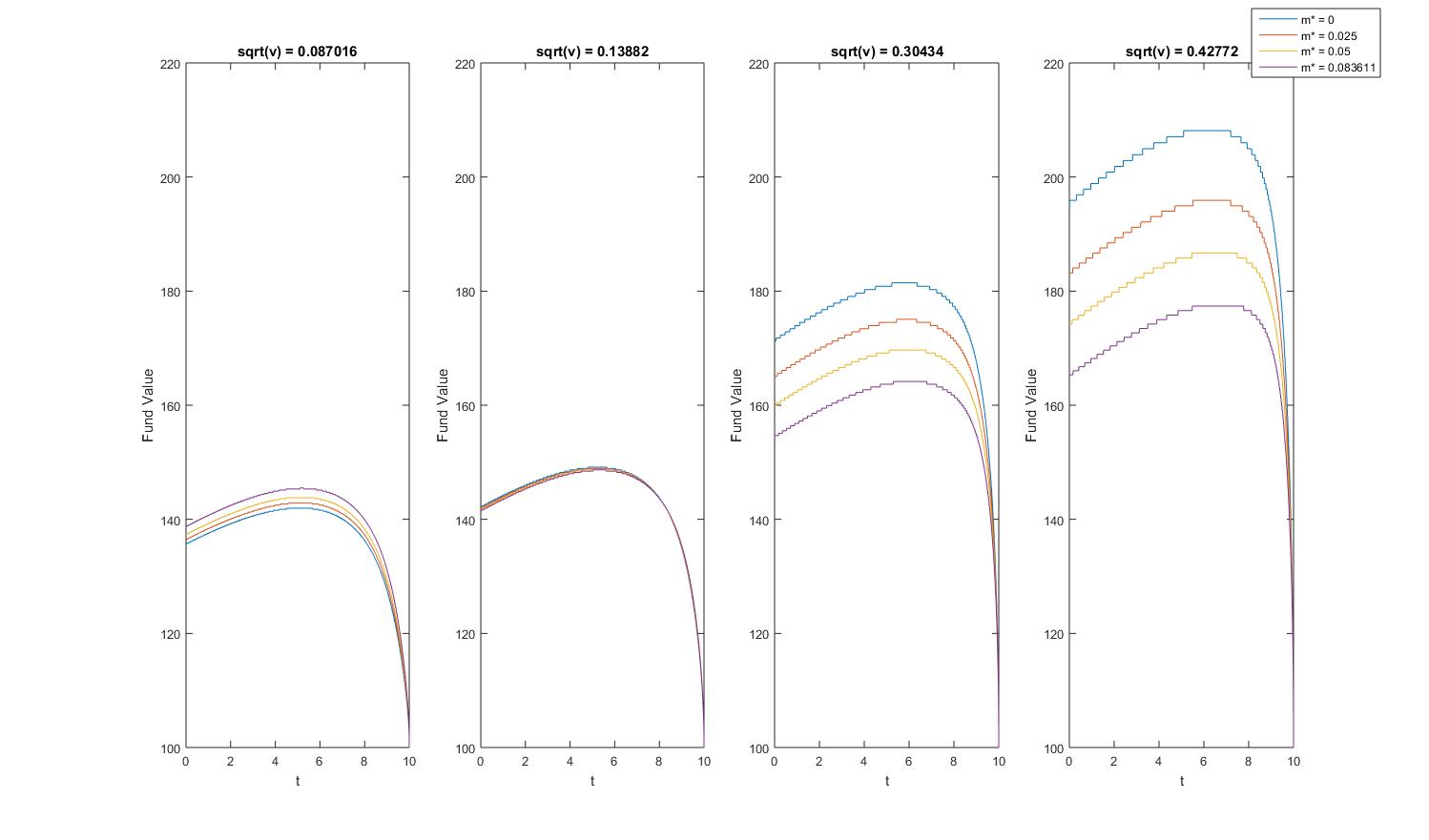}}
	\end{tabular}
	\caption{\small{ The $y$ section of the approximated optimal surrender surface, $f^{(m,N)}_y$, for different volatility levels $\sqrt{y}$ and fair multipliers $\tilde{m}^*$.}}\label{figBoundaryAll}
\end{figure}

For the capped fees, we observe in Figure \ref{figBoundaryAll}\textsc{(b)} that when the volatility is low, capped $\vix^2$-linked fee VA contracts are surrendered for lower fund value than the constant fee ones, as for the two other fee structures. However, when the volatility gets high enough, the cap is reached, and so, the fee paid under the capped structure is the same for all $m^*>0$ (see Table \ref{tblFeeExample}); the capped $\vix^2$-linked fee VA contracts are thus surrendered at almost all the same fund value level (Figure \ref{figBoundaryAll} \textsc{(b)} graphs 3 and 4 when $m^*>0$) . This illustrates again the relation that exists between fees and optimal surrender decisions. VA contracts with higher fee rates are surrendered for lower fund values than contracts with lower fee rates.\\
\\
Similar conclusions can be drawn when $G=F_0e^{\tilde{g} T}$ with $\tilde{g}=2\%$ in the Heston model  and under the 3/2 stochastic volatility model, see Appendix \ref{appendixSupplMaterial}, available online as supplemental material, for details.

\section{Conclusion}\label{conc}
In this paper, we provide a framework based on CTMC approximations to analyze the surrender incentives resulting from VIX-linked fees in variable annuities under general stochastic volatility models. 
Under the assumption that the policyholder will maximize the risk-neutral value of his variable annuity, the pricing of a variable annuity is an optimal stopping problem with a time-discontinuous reward function. 
Under general fee and surrender charge structures, we develop efficient numerical algorithms based on a two-layers CTMC approximation to price variable annuities with and without early surrenders. 
We derive a closed-form analytical formula for the value of a variable annuity without surrender rights and provide a quick and simple way of determining early surrenders value via a recursive algorithm. 
We also present an easy procedure to approximate the optimal surrender surface under the hypothesis that the surrender region is of threshold type. 
Lastly, we observe numerically that $\vix$-linked fees do not significantly affect the value of early surrenders under the selected set of Heston parameters. 
However, numerical examples also reveal that the optimal surrender decision is impacted by VIX-linked fee structures. 
In particular, we observe that the optimal surrender strategy is more stable with respect to volatility changes when the fees are linked to the volatility index. 
The CTMC approximation of the VIX presented in Section \ref{subsectionCTMCapproxVIX} can also be used to price exotic path-dependent options on the $\vix$, which is left as future research.
\section*{Disclosure statement}
The authors report there are no competing interests to declare.
\bibliography{myBib}

  \clearpage
\newpage
\appendix

\section{Proofs of Main Results}

\subsection{Proof of Proposition \ref{propExpectionApprox}}\label{appendixProofPropCondExpApprox}

First recall that
\begin{align*}
	&\prob{P}(\XmN_{t+h}=x_l,\Vm_{t+h}=v_j\big|\XmN_{t}=x_i,\Vm_t=v_k)\\
	& \qquad =\prob{P}(Y^{(m,N)}_{t+h}=(j-1)N+l\big|Y^{(m,N)}_{t}=(k-1)N+i).
\end{align*}
By inspection of the matrix $\mathbf{G}^{(m,N)}$  and since $h$ is small, we have by \eqref{eq_pijh} that
\begin{equation*}
	\prob{P}(\XmN_{t+h}=x_l,\Vm_{t+h}=v_j\big|\XmN_{t}=x_i,\Vm_t=v_k)=\begin{cases}
		q_{kj}h+c_{ii}^{kj}(h) & \text{if $i=l$,$j\neq k$}\\
		\lambda_{il}^kh+c_{il}^{kk}(h) &  \text{if $i\neq l$,$j=k$}\\
		1+(q_{kk}+	\lambda_{ii}^k)h+c_{ii}^{kk}(h) & \text{if $i=l$,$j=k$}\\
		c_{il}^{kj}(h) &\text{if $i\neq l$, $k\neq j$}.
	\end{cases}
\end{equation*}
where the functions $\{c_{il}^{kj}\}$ satisfy ${\lim_{h\rightarrow 0}\frac{c_{il}^{kj}(h)}{h}=0}$ for ${i,l\in\{1,2,\ldots,N\}}$ and ${k,j\in\{1,2,\ldots,m\}}$.\\
\\
From the last equality, we observe that the regime-switching CTMC cannot change regime ($\Vm$) and state ($\XmN$) simultaneously over small time intervals.\\
\\
It follows that
\begin{eqnarray*}
	& &\e\left[\phi\left(t+h,\XmN_{t+h},\Vm_{t+h}\right)|\XmN_t=x_i,\Vm_t=v_k\right]\\
	& &\quad=\sum_{\substack{j=1\\j\neq k}}\phi(t+h,x_i,v_j)(q_{kj}h+c_{ii}^{kj}(h))+\sum_{\substack{l=1\\l\neq i}}\phi(t+h,x_l,v_k)(\lambda_{il}^kh+c_{il}^{kk}(h))\\
	& &\quad\quad\quad +\phi(t+h,x_i,v_k)\left(1+(q_{kk}+\lambda_{ii}^k)h+c_{ii}^{kk}(h)\right)+\sum_{\substack{j=1\\j\neq k}}^m\sum_{\substack{l=1\\l\neq i}}^N \phi(t+h,x_l,v_j)c_{il}^{kj}(h)\\
	& &\quad = \phi(t+h,x_i,v_k)+\sum_{j=1}^m \phi(t+h,x_i,v_j)q_{kj}h+\sum_{l=1}^N\phi(t+h,x_l,v_k)\lambda_{il}^kh+c_{N}(h),
\end{eqnarray*}
\begin{sloppypar}
since for ${i,l\in\{1,2,\ldots,N\}}$ and ${j,k\in\{1,2,\ldots,m\}}$, ${\lim_{h\rightarrow 0}\frac{\phi(t+h,x_l,v_j) c_{il}^{kj}(h)}{h}=0}$.\\
\end{sloppypar}
Also observe that
\begin{eqnarray*}
	& &\e\left[\phi(t+h,\XmN_{t+h},v_j)|\Vm_t=\Vm_{t+h}=v_j,\XmN_t=x_i\right]\\
	& &\quad=\sum_{l=1}^N \phi(t+h,x_l,v_j)\prob{P}\left(\XmN_{t+h}=x_l|\Vm_t=\Vm_{t+h}=v_j,\XmN_t=x_i\right)\\
	& &\quad=\sum_{\substack{l=1\\l\neq i}}^N \phi(t+h,x_l,v_j)(\lambda_{il}^jh+c_{il}^{j}(h))+\phi(t+h,x_i,v_j)(1+\lambda_{ii}^jh+c_{ii}^{j}(h))\\
	& &\quad=\phi(t+h,x_i,v_j)+\sum_{l=1}^N \phi(t+h,x_l,v_j)(\lambda_{il}^jh+c_{il}^{j}(h)).
\end{eqnarray*}

Hence, we have that
\begin{eqnarray*}
	& &\sum_{j=1}^m\e\left[\phi(t+h,\XmN_{t+h},v_j)|\Vm_t=\Vm_{t+h}=v_j,\XmN_t=x_i\right]\prob{P}\left(\Vm_{t+h}=v_j |\Vm_{t}=v_{k}\right)\\
	& &\quad = \sum_{j=1}^m \left(\phi(t+h,x_i,v_j)+\sum_{l=1}^N\phi(t+h,x_l,v_j)(\lambda_{il}^jh+c_{il}^j(h))\right)\left(q_{kj}h+c_{kj}(h)\right)\\
	& &\quad\quad\quad\quad +\phi(t+h,x_i,v_k)+\sum_{l=1}^N\phi(t+h,x_l,v_k)(\lambda_{il}^kh+c_{il}^{k}(h))\\
	&=&\phi(t+h,x_i,v_k)+\sum_{j=1}^m \phi(t+h,x_i,v_j)q_{kj}h+\sum_{l=1}^N\phi(t+h,x_l,v_k)\lambda_{il}^kh+c_{m}(h).
\end{eqnarray*}
where the second line follows from
$$\prob{P}\left(\Vm_{t+h}=v_j |\Vm_{t}=v_{k}\right)=\begin{cases}
	1+q_{kk}h+c_{kk}(h)&\text{if $j=k$}\\
	q_{kj}h+c_{kj}(h)&\text{if $j\neq k$}.
\end{cases}$$
The results follows from setting $c(h)=c_{m}(h)-c_{N}(h)$.

\hfill\qed

\subsection{Proof of Proposition \ref{propBermAmConv}}\label{Proof-4.5}

In order to prove Proposition \ref{propBermAmConv}, we first need to introduce some additional notation.\\
\\
Suppose that $t=0$. For $0\leq s\leq T$, define the discounted reward process of the original contract by
$$Z_s=e^{-rs}\varphi(s,F_s,V_s),$$
and the discounted reward process of the modified contract with $M$ possible exercise dates by
$$\tilde{Z}^{(M)}_s=Z_s\ind_{\{s\in\mathcal{H}_M\}}.$$
Let the process $\lbrace J^{(M)}_t\rbrace_{ 0\leq t\leq T}$ be defined by
\begin{equation}
J^{(M)}_t=\esssup_{\tau\in \mathcal{T}_{\Delta_M},\,\tau\geq t} \e_t[Z_{\tau}],\label{eqJ}
\end{equation}
and the process $\lbrace \tilde{J}^{(M)}_t\rbrace_{0\leq t\leq T}$ be defined by 
\begin{equation}
\tilde{J}^{(M)}_t=\esssup_{\tau\in \mathcal{T}_{t,T}}\e_t[\tilde{Z}_{\tau}^{(M)}].\label{eqtildeJ}
\end{equation}
 As explained in \cite{schweizer2002bermudan}, since the Bermudan contract cannot be exercised outside the region of permitted exercise times and $Z$ is non-negative, we expect the value of the Bermudan contract with the payoff process $Z$ to be equivalent to an American contract with the payoff process $\tilde{Z}^{(M)}$. That is, we expect $J^{(M)}$ to be equal to $\tilde{J}^{(M)}$, almost surely. 
This idea is formalized in \cite{schweizer2002bermudan}, Proposition 3 which is restated in the following lemma. 
\begin{lemma}[\cite{schweizer2002bermudan}, Proposition 3]\label{lemJequal}
	Fix $M\in\mathbb{N}$ and let $J^{(M)}$ and $\tilde{J}^{(M)}$ be defined as in \eqref{eqJ} and \eqref{eqtildeJ}, respectively, then $J^{(M)}=\tilde{J}^{(M)}$, almost surely.
\end{lemma}

We can now prove Proposition \ref{propBermAmConv}.
\begin{proof}[Proof of Proposition \ref{propBermAmConv}]
	The following is inspired by the proof of Proposition 6 of \cite{bassan2002optimal}.\\
	\\
	Without loss of generality, suppose that $t=0$. 
	We showed in Lemma \ref{lemJequal} that $J_t^{(M)}=\tilde{J}_t^{(M)}$, almost surely. Hence, the Bermudan contract in \eqref{eqBermVA} may be expressed as follows:
	$$ b_M(0,x,y)=\sup_{\tau \in \mathcal{T}_{0,T}} \e_{x,y}[\tilde{Z}^{(M)}_{\tau}],$$
 Now notice that for each $t$, $\{\tilde{Z}_t^{(M)}\}_{M\in \mathbb{N}}$ is an increasing sequence of random variables converging to $Z_t$ pointwisely, that is, for each $t\geq 0$, $\tilde{Z_t}^{(M)}(\omega)\uparrow Z_t(\omega)$ for all $\omega\in\Omega$. Thus, given $F_0=x$ and $V_0=y$ and  by using the monotone convergence theorem (\cite{williams1991probability}, Theorem 5.3), we obtain
	\begin{eqnarray*}
		v(0,x,y)&=&\sup_{\tau \in \mathcal{T}_{0,T}} \e[Z_{\tau}]=\sup_{\tau \in \mathcal{T}_{0,T}} \e[\sup_{M\in\mathbb{N}} \tilde{Z}^{(M)}_{\tau}]
		=\sup_{\tau \in \mathcal{T}_{0,T}} \sup_{M\in\mathbb{N}} \e[ \tilde{Z}^{(M)}_{\tau}]\\
		&=&  \sup_{M\in\mathbb{N}} \sup_{\tau \in \mathcal{T}_{0,T}} \e[ \tilde{Z}^{(M)}_{\tau}]
	  = \sup_{M\in\mathbb{N}} b_M(0,x,y)=\lim_{M\rightarrow\infty} b_M(0,x,y).
	\end{eqnarray*}
\end{proof}
\newpage
\setcounter{page}{1}
\section{Supplemental Material}\label{appendixSupplMaterial}
This document provides supplemental material to Analysis of VIX-Linked Fee Incentives in Variable Annuities via Continuous-Time Markov Chain Approximation.

\subsection{Accuracy of the VA price and the Approximated Optimal Surrender Strategy}
 The purpose of this subsection is to analyze numerically the accuracy of Algorithm \ref{algoOptSurrSurf} to approximate the optimal surrender strategy. In the context of American put options under the Black-Scholes setting, \cite{lamberton1993convergence} considers a similar approximation for the exercise boundary when the underlying diffusion process is approximated by a Markov chain. In particular, under this setting, \cite{lamberton1993convergence} proves theoretically the convergence of the approximated exercise boundary (also called the critical price) towards the exact exercise boundary in their Theorem 2.1. 
 \\
 In order to analyze the accuracy of our methodology against a known benchmark, we use our algorithms in a problem that has been done by others. More precisely, we compare the approach of Algorithm \ref{algoOptSurrSurf} (also called the CTMC Bermudian approximation and denoted below by ``CTMC-Berm'') for approximating the optimal surrender surface and the value of a variable annuity with early surrenders to the ones obtained using the approach of \cite{bernard2014optimal}, the ``Benchmark''. The methodology proposed by \cite{bernard2014optimal} is based on the integral equation method of \cite{kim1990analytic} under the Black-Scholes model, and cannot be easily extended to more general models. Thus, in order to compare with their method, we assume that the volatility of the index is constant so that
$$\diff S_t, =rS_t \diff t+ \sigma S_t \diff W_t,\quad t\geq 0,$$
where $\sigma>0$ is the volatility. We also suppose that $c(x,y)=\tilde{c}>0$, for all $(x,y)\in\reals_+\times\mathcal{S}_V$, and we use $g(t,y)=1$ for all $(t,y)\in[0,T]\times\mathcal{S}_V$. Hence, the VA account dynamics is given by
$$\diff F_t, =(r-\tilde{c})F_t \diff t+ \sigma F_t \diff W_t,\quad 0 \leq t\leq T.$$
We approximate the logarithm of the fund process $\tilde{X}_t=\ln F_t$, $0\leq t\leq T$ by using  CTMC approximation method. The CTMC approach for one-dimensional diffusion processes works in the exact same way as the one described in Subsection \ref{subSectCTMCvariance} for the approximation of $V$, see \cite{cui2019continuous} for details.  The number of state-space grids and time steps are set to $N=5,000$ and $M=T\times 500$, respectively. We use the grid of  \cite{tavellapricing}, as discussed in Section \ref{subsecMarketVACTMCparam}, with  a non-uniformity parameter $\tilde{\alpha}_{\tilde{X}}=5$. The grid upper ($\tilde{x}_N$) and lower bounds ($\tilde{x}_1$) of the approximated process $\tilde{X}$ are set about the mean of $\tilde{X}_t$ at $t=T/2$ as follows:
$\tilde{x}_N=\bar{\mu}(t)+\gamma\bar{\sigma}(t)$ and $\tilde{x}_1=\bar{\mu}(t)-\gamma\bar{\sigma}(t)$
where $\bar{\mu}(t)$ is the conditional mean of $\tilde{X}_t$ given $\tilde{X}_0$ and $\bar{\sigma}(t)$, the conditional standard deviation. We use $\gamma=7.2$. Now note that, under the Black-Scholes setting, we have
$$\bar{\mu}(t)=\tilde{X}_0+\left(r-\frac{\sigma^2}{2}\right)t,\quad\textrm{ and }\quad\bar{\sigma}(t)=\sigma\sqrt{t}.$$

As a benchmark, we use the value of the variable annuity with early surrenders and the optimal surrender boundary obtained using the method of \cite{bernard2014optimal} with $1500$ steps per year. Other parameters are $G=100$, $r=0.03$ and $T=15$. We tested our method for different volatility levels $\sigma\in\{0.1,\,0.2,\,0.3,\,0.4\}$ and initial values of the VA account $F_0\in\{90,100\}$. The fair fee $\tilde{c}^*$ is calculated such that the expected discounted value of the maturity benefit equals the value of the VA account at inception, that is $F_0=\e\left[e^{-rT}\max(G, F_T)\right]$. We compare to the benchmark the value of the variable annuity with early surrenders ``VA with ES'' and the optimal surrender boundary (``Opt. Surr. Bound.''). The results are reported in Table \ref{tblConvSurrBound_BS}. The column ``Rel. Err.'' documents the relative errors whereas ``Max Re. Err'' documents the maximum relative errors\footnote{Given a benchmark value $y$ and its approximation $y_{approx}$, the relative error is defined as $|y-y_{approx}|/|y|$ for $y\neq 0$; whereas the maximum relative error is the largest relative error over a sample of benchmark values and their approximations.}. 
 \begin{table}[h!]
	\begin{tabular}{l cc ccc c c}
		\toprule
		   &   \multirow{2}{*}{$\sigma$}& \multirow{2}{*}{$\tilde{c}^*$ ($\%$)} &\multicolumn{3}{c}{\textbf{VA with ES}} & &\textbf{Opt. Surr. Bound.}\\
		 \cmidrule{4-6} \cmidrule{8-8} 
		& & & CTMC-Berm & Bernard et al. (2014) & Rel. Err.& & Max Rel. Err.\\
		\midrule
		\multirow{4}{*}{$F_0=100$} &0.1&0.1374 &100.851600 & 100.851748 &1.47e-6 & & 2.60e-3\\
		  &0.2&0.9094 &104.400379 & 100.401287 &8.70e{-6} & & 5.20e{-3}\\
		  &0.3&1.9277 &108.577366 & 108.579001 &1.51e{-5} & &7.78e{-3}\\
		  &0.4&2.9415 &112.823616 & 112.826112 &2.21e{-5} & &1.04e{-2}\\
		\midrule
	    \multirow{4}{*}{$F_0=90$} &0.1&0.2641 &91.28494 & 91.285171 &2.53e-6 & & 2.60e-3\\
		&0.2&1.3062 &94.989758 & 94.990712 &1.01e{-5} & & 5.20e{-3}\\
		&0.3&2.5571 &99.012022 & 99.013806 &1.80e{-5} & &7.78e{-3}\\
		&0.4&3.7631 &103.022547 & 103.025197 &2.57e{-5} & &1.03e{-2}\\
		\bottomrule
	\end{tabular}
	\caption{Relative errors of a variable annuities with ES and maximum relative errors of the optimal surrender boundaries.} 
	\label{tblConvSurrBound_BS}
\end{table} 
\\
\\
From Table \ref{tblConvSurrBound_BS}, we see that the CTMC Bermudian approximation achieves a high level of accuracy across all volatility levels and initial value of VA account $F_0$.

\subsection{Accuracy and efficiency of the CTMC approximation for the VIX index}
We assess the accuracy of the CTMC approximation for the volatility index for the 3/2 and Heston models. Under a Heston-type model, the VIX has a closed-form expression given by
\begin{equation}\label{eqVIXsquared}
	\vix^2_t=B+ A V_t
\end{equation}
with $A=\frac{1-e^{-\kappa\tau}}{\kappa\tau}$ and $B=\frac{\theta(\kappa\tau-1+e^{-\kappa\tau})}{\kappa\tau}$, see \cite{ZhuZang2007} for details.\\
\\
For the 3/2 model, a closed-form expression for the $\vix$ may be found in \cite{carr2007new}, Theorem 4. However, as pointed out by \cite{drimus2012options}, the integral that appears in the analytical formula is difficult to implement and is not suited for fast and accurate numerical approximation. For this reason, benchmark results are obtained via Monte Carlo simulation\footnote{We simulate 500K paths (plus 500K antithetic variables) using Milstein discretization scheme and we use 5,000 time steps. } under 3/2 model.\\
\\
The market and CTMC parameters are set to the ones of Subsection \ref{subsecMarketVACTMCparam} for the Heston model, except that we set $m=1,000$ (rather than $m=50$). For the 3/2 model, we use the parameters reported in Section \ref{subsectParam32}. The results are presented in Table \ref{tblVIXapproxAccuracy}. The column ``$V_t$'' (resp. $1/V_t$)  displays the initial value of the variance at time $t\geq 0$ for the Heston model (resp. the 3/2 model), the column ``CTMC'' shows the result of the VIX approximation using Algorithm \ref{algoCTMCapproxVIX} with $n=1,000$ time steps, and the column ``Benchmark'' reports the benchmark value calculated either by using the closed-form formula \eqref{eqVIXsquared} (for the Heston model) or by Monte Carlo simulation (for the 3/2 model). The relative error are documented in column ``Rel. err''.

\begin{table}[h!]
	\begin{subtable}[c]{0.49\linewidth}
		\centering
		\scalebox{0.98}{
			\begin{tabular}{c|ccc}
		\hline
		$\mathbf{V_t}$ & \textbf{CTMC} & \textbf{Benchmark} & \textbf{Rel. err.} \\
		\hline \hline
		 0.01&	11.1077\%&	11.1068\%&	8.10e-05\\
		0.02&	14.6829\%&	14.6824\%&	3.41e-05\\
		0.04&	20.0000\%&	20.0000\%&	0.00e+00\\
		0.06&	24.1746\%&	24.1749\%&	1.24e-05\\
		0.09&	29.3433\%&	29.3439\%&	2.04e-05\\
		\hline
	\end{tabular}}
		\subcaption{Heston model}
	\end{subtable}
	\begin{subtable}[c]{0.49\linewidth}
		\centering
		\scalebox{0.98}{
			\begin{tabular}{c|ccc}
		\hline
		$1/\mathbf{V_t}$ & \textbf{CTMC} & \textbf{Benchmark}  & \textbf{Rel. err.} \\
		\hline\hline
	    0.01	&11.2202\%&	11.2203\%6&	8.91e-06\\
		0.02&	15.7016\%&	15.7020\%&	2.55e-05\\
		0.04&	21.4892\%&	21.4901\%&	4.19e-05\\
		0.06&	25.3096\%&	25.3112\%&	6.32e-05\\
		0.09&	29.2708\%&	29.2735\%5&	9.22e-05\\
		
		\hline
	\end{tabular}}
		\subcaption{3/2 model}
	\end{subtable}
	\caption{Accuracy of the CTMC- VIX approximation,  Algorithm \ref{algoCTMCapproxVIX}}
	\label{tblVIXapproxAccuracy}
\end{table} 

The results of Table \ref{tblVIXapproxAccuracy} confirm the high accuracy of the CTMC-VIX approximation for both models. It is worth mentioning that, when using Algorithm \ref{algoCTMCapproxVIX}, we obtain simultaneously the value of the $\vix_t^{(m),k}$ for all $v_k\in\mathcal{S}_V^{(m)}$ within less than a 0.1 second for both models. The value of the CTMC approximation given a particular value for $V_t$ is then interpolated between the appropriate grid points. This further increases the efficiency of our algorithm.

\subsection{VA prices accuracy and computation time under the Heston model}\label{appendixCTMCresult}
Table \ref{tblFairValueSet1N100} shows the value of a variable annuity with and without surrender rights (``VA with ES'' and ``VA without ES'', respectively) under the Heston model with the Uncapped $\vix^2$-linked fees using $N=100$, $N=1,100$ and $M=252\times 10$; and using $N=2,000$ and $M=500\times 10$. All other market, VA and CTMC parameters are the same as in Subsection \ref{subsecMarketVACTMCparam}.
\begin{table}[h!]
	\centering
	\resizebox{\textwidth}{!}{\begin{tabular}{ll| cccc|c}
		\hline
		 & $\mathbf{\tilde{m}^*= }$ & $\mathbf{0.0000}$ & $\mathbf{0.1500}$ & $\mathbf{0.3000}$ & $\mathbf{0.4345}$ & \\
		& $\mathbf{c^*\,\,=  }$ & $\mathbf{1.5338\%}$  & $\mathbf{1.0036\%}$ & $\mathbf{0.4741\%}$ & $\mathbf{0.000\%}$ & \textbf{Computation Time (sec.)}\\
		\hline
		\hline
		\multirow{3}{*}{\textbf{VA without ES}} & $N\,\,=2,000$ & $100.00090$ & $100.00091$& $100.00092$& $100.00093$ & $2,500$\\
		& $N\,\,=1,100$ & $100.00094$&$100.00095$ & $100.00096$& $100.00097$ & $400$\\
		 & $N\,\,=100$ & $100.00750$ & $100.00769$& $100.00789$& $100.00807$ & $0.85$\\
		 \hline
		\multirow{3}{*}{\textbf{VA with ES}} & $N\,\,=2,000$&$103.01785$ & $103.00823$ & $103.00330$ & $103.00367$ &$7,100$\\
		  & $N\,\,=1,100$ &$103.01743$ & $103.00788$ & $103.00300$ & $103.00341$ & $1,600$\\
		  & $N\,\,=100$ &$103.0162$ & $103.00676$ & $103.00190$ & $103.00228$ & $54$\\
		\hline 
		\multirow{3}{*}{\textbf{ES value}} & $N=2,000$ & $3.01695$ & $3.00732$ & $3.00238$ & $3.00274$ &N/A\\
	   	&$N\,\,=1,100$ & $3.01649$ & $3.00693$ & $3.00204$ & $3.00243$ & N/A\\ 
		 & $N\,\,=100$ & $3.00862$ & $2.99907$ & $2.99401$ & $2.99421$ & N/A\\
		\hline
	\end{tabular}}
	\caption{Variable annuity with and without early surrenders using CTMC Approximation with $N=100$ and $1,100,$ with  $M=252\times 10$; and  $N=2,000$ with $M=500\times 10$.}
	\label{tblFairValueSet1N100}
\end{table} 
\\
\\
All the numerical results and computation times in Table \ref{tblFairValueSet1N100} are performed using Equation \ref{eq_CTMC_VA} and Algorithm \ref{algoVAsurrenders} combine with the Expokit procedures of \cite{sidje1998expokit}, function \textit{expv}. Since fair fees are calibrated at inception using the exact pricing formula of \cite{cui2017}, the benchmark value for the VA without surrender rights is $F_0=100$. When $N=100$, accurate VA prices are obtained extremely fast (within $54$ sec. for the VA with surrender rights and less than a second for the VA without surrender rights). The absolute error is around $10^{-3}$ for both values\footnote{The benchmark for the VA with early surrenders is the approximated CTMC value obtained using $N=2,000$ and $M=500$.}.\\
\\
Similar results are obtained with the Fast Algorithms (\ref{algoVA_WOsurrendersFast}, \ref{algoVA_WITHsurrendersFast} and \ref{algoOptSurrSurfFast}). The results are reported in Table \ref{tblFairValueSet1N100Fast}. The values of VA with and without surrender rights are calculated simultaneously, as per Remark \ref{rmkFastAlgoEfficiency}. So, the computation times in column ``Computation Time (sec.)'' are for both prices.

\begin{table}[h!]
	\centering
	\resizebox{\textwidth}{!}{\begin{tabular}{ll| cccc|c}
			\hline
			& $\mathbf{\tilde{m}^*= }$ & $\mathbf{0.0000}$ & $\mathbf{0.1500}$ & $\mathbf{0.3000}$ & $\mathbf{0.4345}$ & \\
			& $\mathbf{c^*\,\,=  }$ & $\mathbf{1.5338\%}$  & $\mathbf{1.0036\%}$ & $\mathbf{0.4741\%}$ & $\mathbf{0.000\%}$ & \textbf{Computation Time (sec.)}\\
			\hline
			\hline
			\multirow{3}{*}{\textbf{VA without ES}} & $N\,\,=2,000$ & $99.99615$&$99.99611$ & $99.99607$& $99.99603$ & N/A \\
			& $N\,\,=1,100$ & $99.99619$&$99.99615$ & $99.99611$& $99.99607$ & N/A\\
			& $N\,\,=100$ & $100.00276$ & $100.00290$& $100.00304$& $100.00317$ & N/A\\
			\hline
			\multirow{3}{*}{\textbf{VA with ES}} & $N\,\,=2,000$&$103.02361$ & $103.01360$ & $103.00815$ & $103.00789$ &$4,405$\\
			& $N\,\,=1,100$ &$103.02360$ & $103.01359$ & $103.00814$ & $103.00788$ & $1,244$\\
			& $N\,\,=100$ &$103.02237$ & $103.01251$ & $103.00715$ & $103.00680$ & $9.60$\\
			\hline 
			\multirow{3}{*}{\textbf{ES value}} & $N=2,000$ & $3.02745$ & $3.01748$ & $3.01208$ & $3.01186$ &N/A\\
			&$N\,\,=1,100$ & $3.02741$ & $3.01744$ & $3.01203$ & $3.01181$ & N/A\\ 
			& $N\,\,=100$ & $3.01961$ & $3.00961$ & $3.00411$ & $3.00363$ & N/A\\
			\hline
	\end{tabular}}
	\caption{Variable annuity with and without early surrenders (ES) using CTMC Approximation Fast Algorithms with $N=100$ and $1,100,$ with  $M=500\times 10$; and  $N=2,000$ with $M=500\times 10$.}
	\label{tblFairValueSet1N100Fast}
\end{table} 

By comparing the two tables, we note that the Fast Algorithms provide very accurate results extremely fast\footnote{Recall that the computation times in Table \ref{tblFairValueSet1N100Fast} for the Fast Algorithms are for the value of the VAs with and without ES simultaneously.} compared to the original Algorithms 

\subsection{ Other Numerical Analysis of VIX-linked Fee Incentives in the Heston Model}\label{appendixNumAnalysisHestonG}
In this appendix, we show the numerical results obtained under the Heston model for the three fee structures (Uncapped $\vix^2$, Capped $\vix^2$ and Uncapped $\vix$), detailed in Subsection \ref{subsectFeeStruct}, when the guaranteed amount is set to $G=F_0 e^{\tilde{g} T}$ with $\tilde{g}=2\%$ (rather than $G=F_0$). The conclusions are similar to the ones detailed in Subsection \ref{subsectNumAnalHeston}.\\
\\
Unless stated otherwise, in this section, all market, VA and CTMC parameters are the same as in Subsection \ref{subsecMarketVACTMCparam}.
	 except for $G=F_0 e ^{\tilde{g}T}$ with $\tilde{g}=2\%$ and $M=T\times 252$.\\
	 \\
Fair fee parameters $(\tilde{c}^*, \tilde{m}^*)$ are presented in Table \ref{tblFairFeeGmodif}, whereas Table \ref{tblESvalueGmodif} shows the approximated values of early surrenders. When $c_t=\tilde{c}^*+\tilde{m}^*\vix^2$, the fair parameters are obtained using the exact formula of \cite{cui2017}. For the other fee structures, the fair parameters are obtained using CTMC approximation with $N=100$, $m=50$ and $M=T \times 252$ to accelerate the computation time.
\begin{table}[h!]
	\begin{subtable}[c]{0.495\linewidth}
		\centering
		\scalebox{0.95}{
				\begin{tabular}{c cccc}
		\hline
		 & \multicolumn{4}{c}{$c_t=\tilde{c}^*+\tilde{m}^*\vix^2_t$}\\
			\hline
		$\tilde{m}^*$ & $0.0000$ & $0.3000$ & $0.8000$ & $1.1313$\\
	
		$\tilde{c}^*$ & $3.82542\%$  & $2.80632\%$ & $1.11545\%$ & $0.0000\%$\\
		\hline
		\hline
		& \multicolumn{4}{c}{$c_t=\min\{\tilde{c}^*+\tilde{m}^*\vix^2_t,K\}$, $K=4.5\%$}\\
		\hline
		$\tilde{m}^*$ & $0.0000$ & $0.3000$ & $0.8000$ & $1.4144$\\
		
		$\tilde{c}^*$ & $3.82542\%$  & $2.83600\%$ & $1.45530\%$ & $0.0000\%$\\
		\hline
		\hline
		& \multicolumn{4}{c}{$c_t=\tilde{c}^*+\tilde{m}^*\vix_t$}\\
		\hline
		$\tilde{m}^*$ & $0.0000$ & $0.0750$ & $0.1250$ & $0.2128$\\
		
		$\tilde{c}^*$ & $3.82542\%$  & $2.47570\%$ & $1.57660\%$ & $0.00000\%$\\
		\hline
	\end{tabular}}
		\subcaption{Fair fee vectors $(\tilde{c}^*,\tilde{m}^*)$}\label{tblFairFeeGmodif}
	\end{subtable}
	\begin{subtable}[c]{0.495\linewidth}
		\centering
		\scalebox{0.95}{
			\begin{tabular}{c cccc}
		\hline
		& \multicolumn{4}{c}{$c_t=\tilde{c}^*+\tilde{m}^*\vix^2_t$}\\
		\hline
		$\tilde{m}^*$ & $0.0000$ & $0.3000$ & $0.8000$ & $1.1313$\\
		
	    ES Value & $5.00216$  & $5.01219$ & $5.03998$ & $5.06681$\\
		\hline
		\hline
		& \multicolumn{4}{c}{$c_t=\min\{\tilde{c}^*+\tilde{m}^*\vix^2_t,K\}$, $K=4.5\%$}\\
		\hline
		$\tilde{m}^*$ & $0.0000$ & $0.3000$ & $0.8000$ & $1.4144$\\
		
		ES Value  & $5.00216$  & $5.00415$ & $4.98201$ & $4.95968$\\
		\hline
		\hline
		& \multicolumn{4}{c}{$c_t=\tilde{c}^*+\tilde{m}^*\vix_t$}\\
		\hline
		$\tilde{m}^*$ & $0.0000$ & $0.0750$ & $0.1250$ & $0.2128$\\
		
     	ES Value  & $5.00216$  & $5.00351$ & $5.00569$ & $5.01254$\\
		\hline
	\end{tabular}}
		\subcaption{Early surrender values (``ES'' values).}\label{tblESvalueGmodif}
	\end{subtable}
	\caption{Fair fee vectors and approximated early surrender values under Heston model when $G=F_0e^{\tilde{g} T}$, $\tilde{g}=2\%$}
\end{table} 
Figures \ref{figBoundary3D_Gmodif2}, \ref{figBoundary_Gmodif2} and \ref{figBoundaryAll_Gmodif2} display the approximated optimal surrender surfaces under the Heston model for the three fee structures when $G=F_0 e ^{\tilde{g}T}$ with $\tilde{g}=2\%$. 
\begin{figure}[!t]
	\centering
	\begin{tabular}{cccc}
	    \multicolumn{4}{c}{\small{$c_t=\tilde{c}^*+\tilde{m}^*\vix^2_t$}}\\
		\includegraphics[scale=0.2]{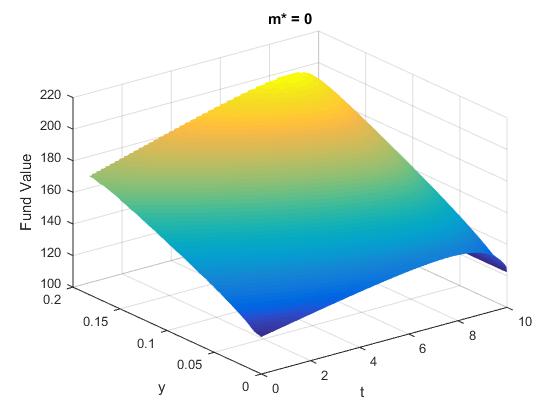}  & \includegraphics[scale=0.2]{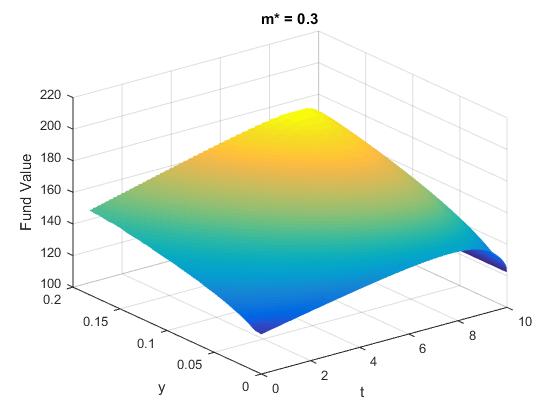}   & \includegraphics[scale=0.2]{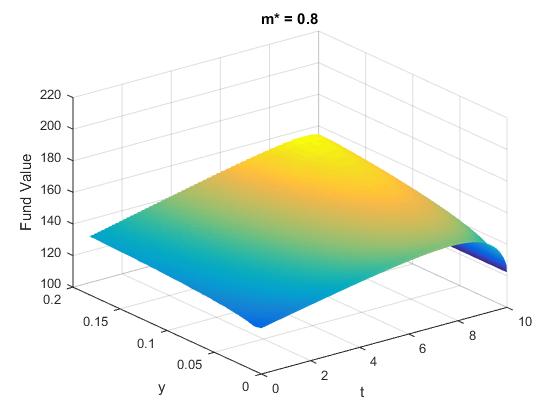}  & \includegraphics[scale=0.2]{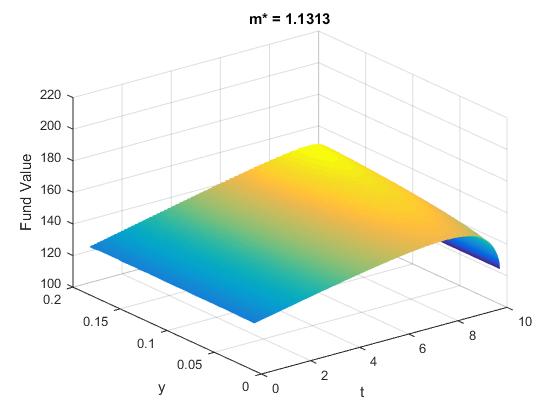} \\
		\multicolumn{4}{c}{\small{$c_t=\min\{\tilde{c}^*+\tilde{m}^*\vix^2_t,K\}$, $K=4.5\%$}}\\
		\includegraphics[scale=0.2]{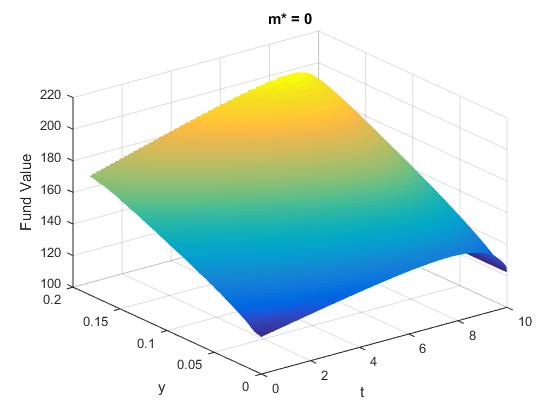} & \includegraphics[scale=0.2]{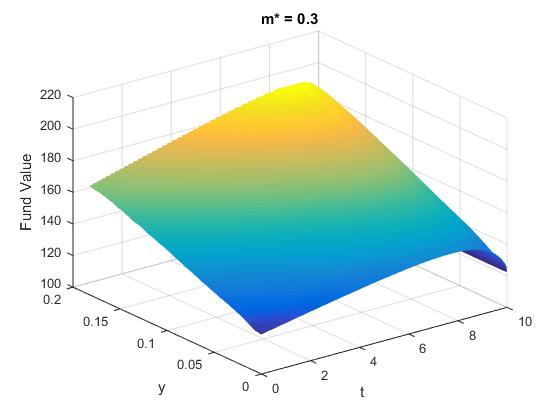}&	\includegraphics[scale=0.2]{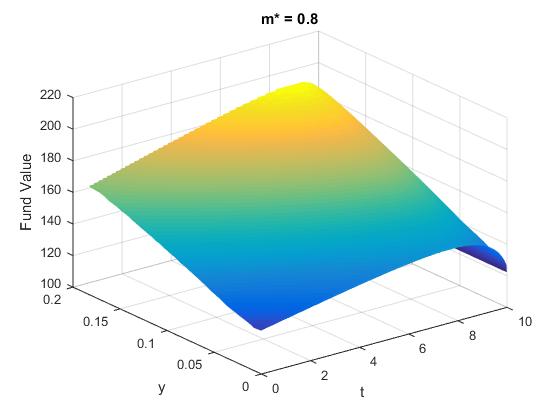} &	\includegraphics[scale=0.2]{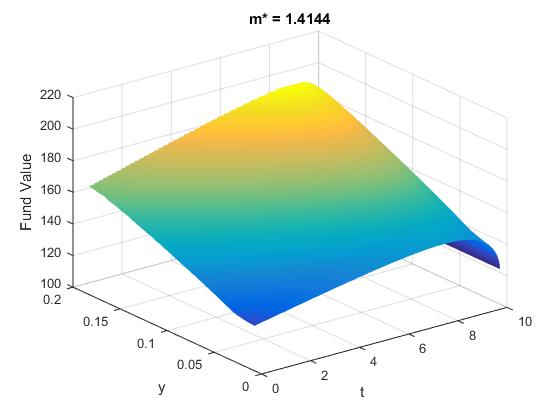} \\
		 \multicolumn{4}{c}{\small{$c_t=\tilde{c}^*+\tilde{m}^*\vix_t$}}\\
		\includegraphics[scale=0.2]{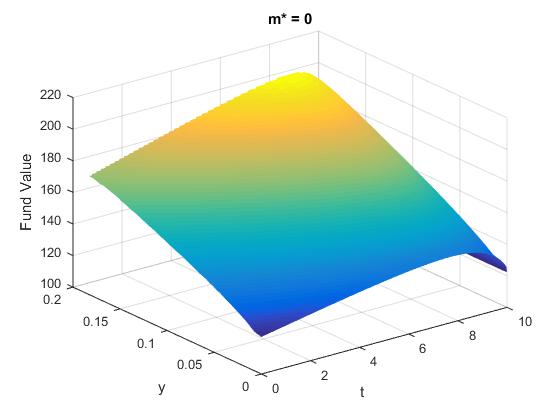} &\includegraphics[scale=0.2]{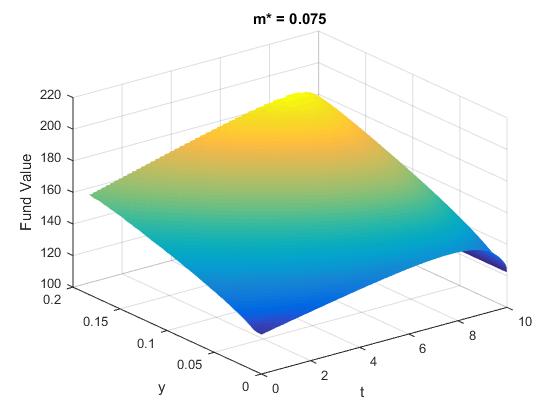}  &\includegraphics[scale=0.2]{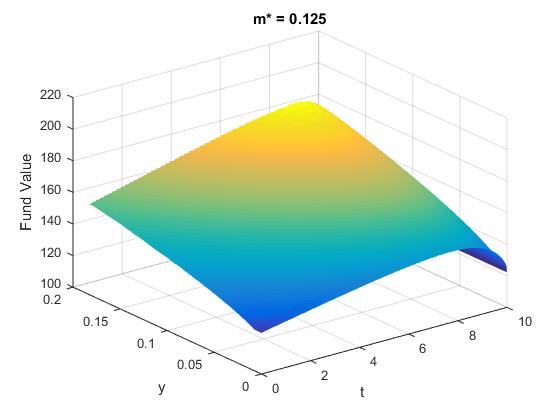}  &
		\includegraphics[scale=0.2]{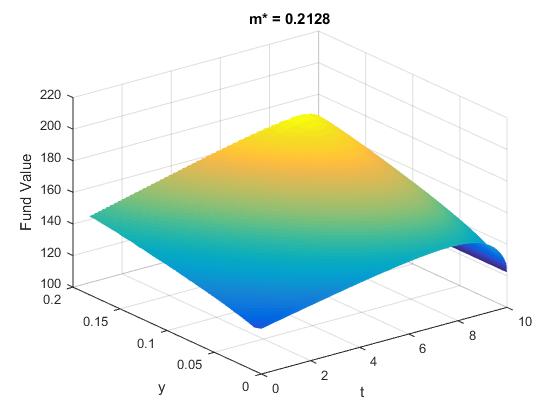}  
    \end{tabular}
	\caption{\small{Approximated optimal surrender surfaces for different values of fair multiplier $\tilde{m}^*$ under the Heston model, $G=F_0 e^{\tilde{g} T}$ with $\tilde{g}=2\%$.}}\label{figBoundary3D_Gmodif2}
\end{figure}
\newpage
\begin{figure}[!t]
	\centering
	\begin{tabular}{cccc}
		\multicolumn{4}{c}{\small{$c_t=\tilde{c}^*+\tilde{m}^*\vix^2_t$}}\\
		\includegraphics[scale=0.2]{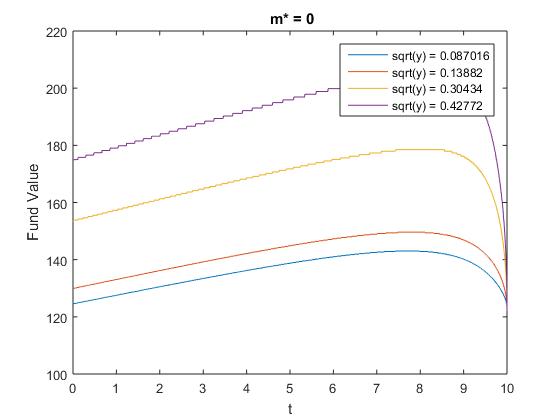}   &\includegraphics[scale=0.2]{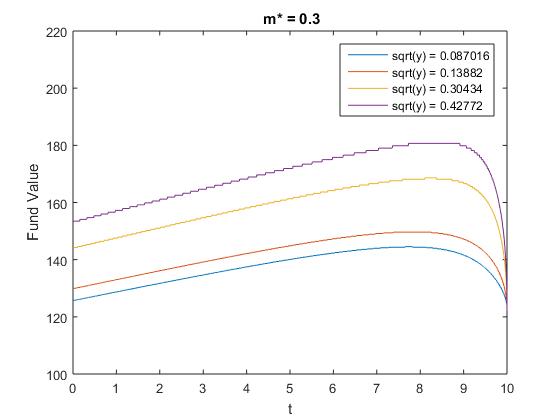}   & \includegraphics[scale=0.2]{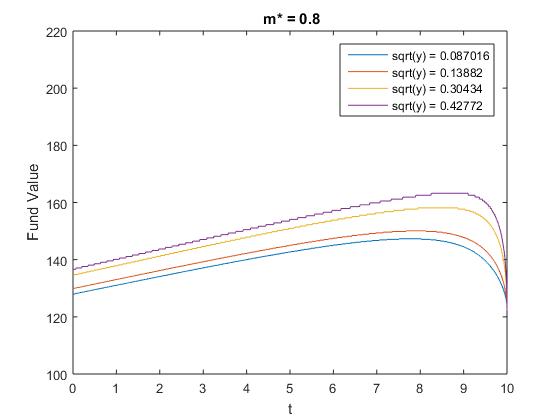}  & \includegraphics[scale=0.2]{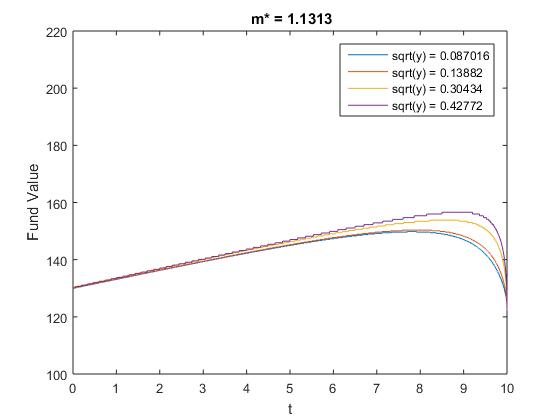} \\
		\multicolumn{4}{c}{\small{$c_t=\min\{\tilde{c}^*+\tilde{m}^*\vix^2_t,K\}$, $K=4.5\%$}}\\
			\includegraphics[scale=0.2]{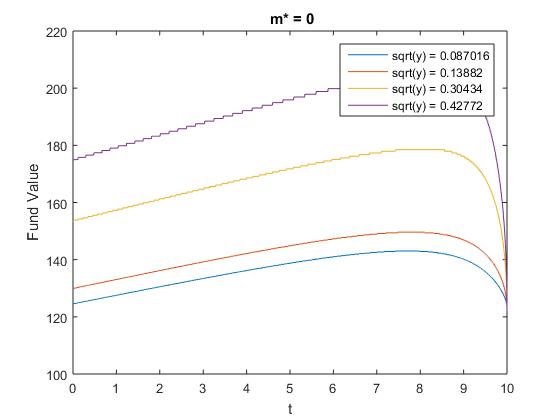}  & 	\includegraphics[scale=0.2]{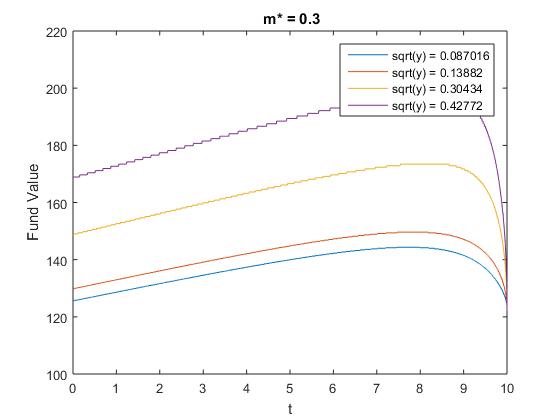} &		\includegraphics[scale=0.2]{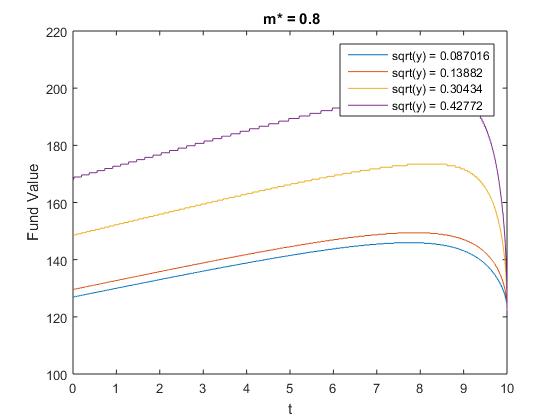}   &		\includegraphics[scale=0.2]{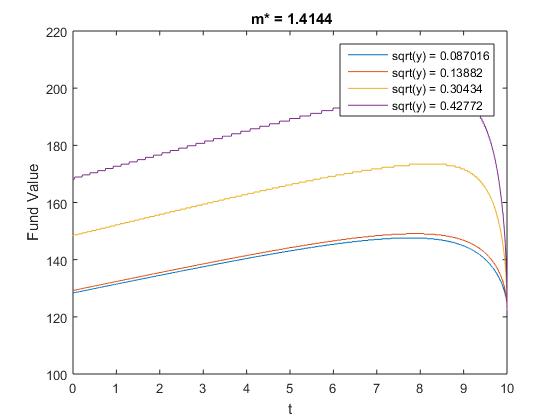}   \\
		\multicolumn{4}{c}{\small{$c_t=\tilde{c}^*+\tilde{m}^*\vix_t$}}\\
		\includegraphics[scale=0.2]{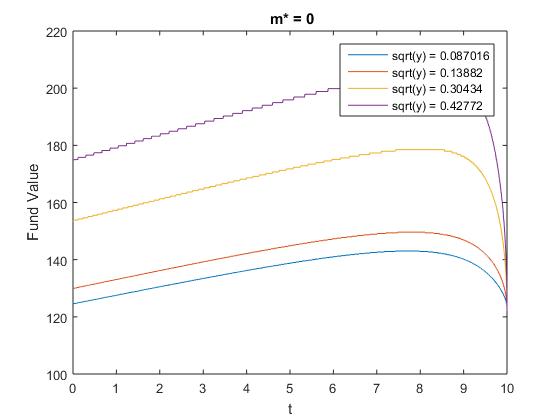}   &\includegraphics[scale=0.2]{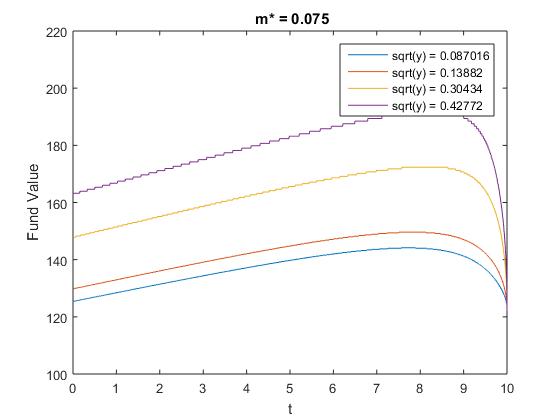}    &\includegraphics[scale=0.2]{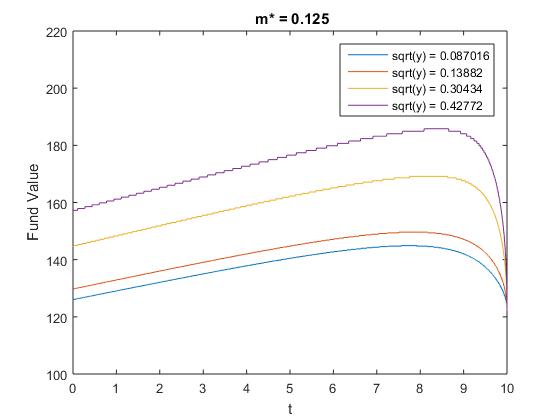}    &
		\includegraphics[scale=0.2]{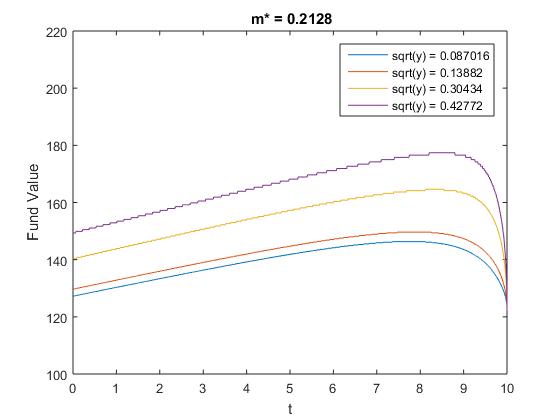}   
	\end{tabular}
	\caption{\small{ The $y$ section of the approximated optimal surrender surface under the Heston model, $f^{(m,N)}_y$, for different volatility levels $\sqrt{y}$ and fair multipliers $\tilde{m}^*$, $G=F_0 e^{\tilde{g} T}$ with $\tilde{g}=2\%$.}}\label{figBoundary_Gmodif2}
\end{figure}
\begin{figure}[!t]
	\centering
	\begin{tabular}{cc}
		\small{$c_t=\tilde{c}^*+\tilde{m}^*\vix^2_t$} & \small{$c_t=\min\{\tilde{c}^*+\tilde{m}^*\vix^2_t,K\}$, $K=4.5\%$} \\
		\includegraphics[scale=0.25]{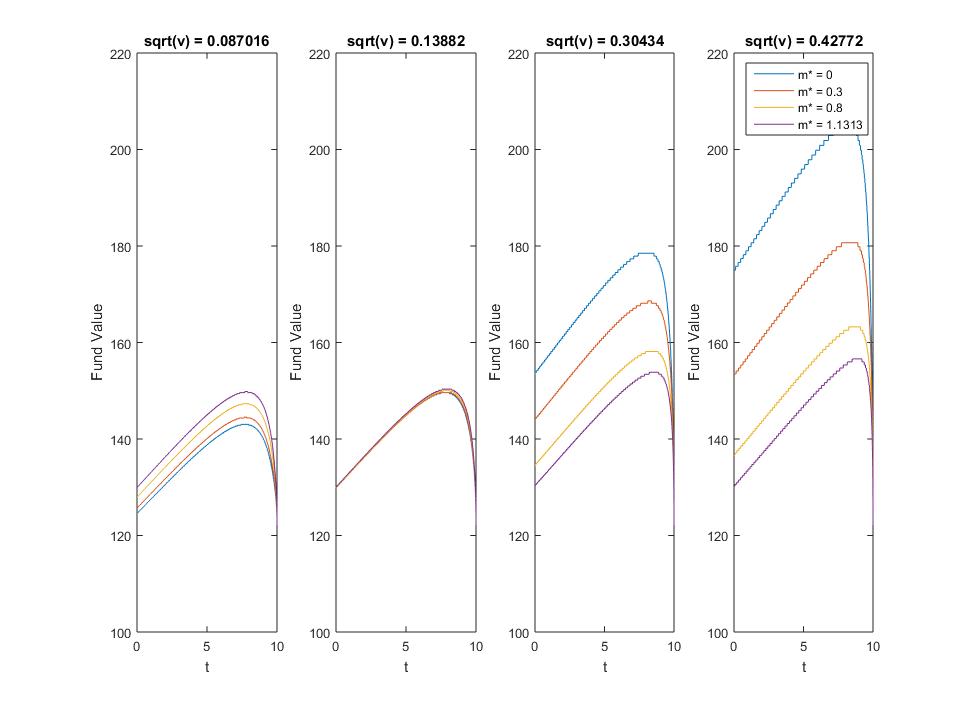}   &\includegraphics[scale=0.25]{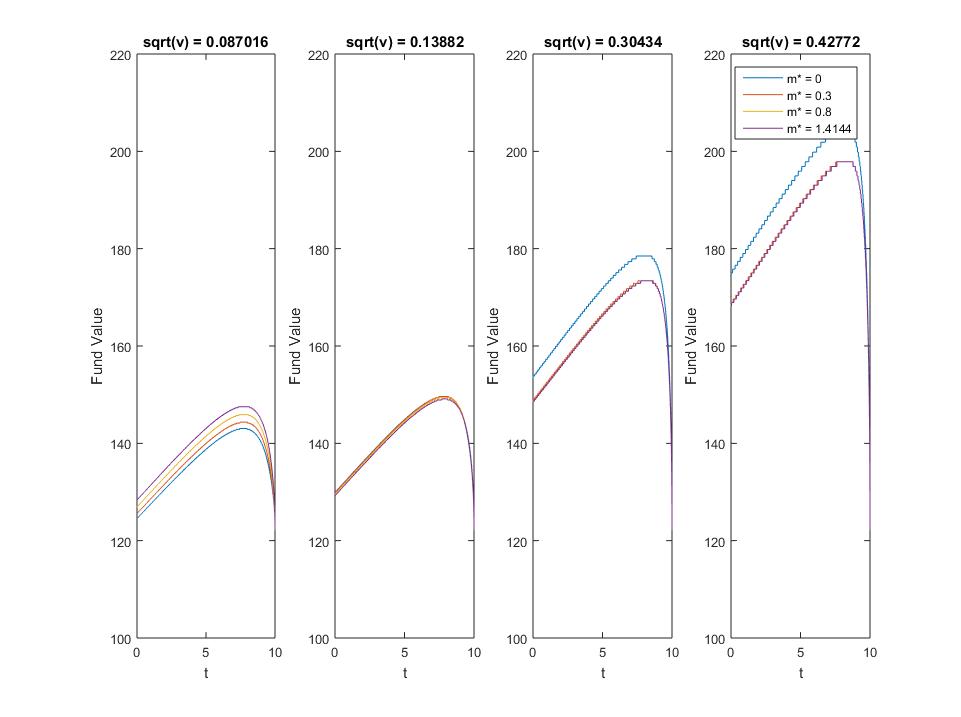}  \\
		\multicolumn{2}{c}{\small{$c_t=\tilde{c}^*+\tilde{m}^*\vix_t$}}\\
		  \multicolumn{2}{c}{\includegraphics[scale=0.25]{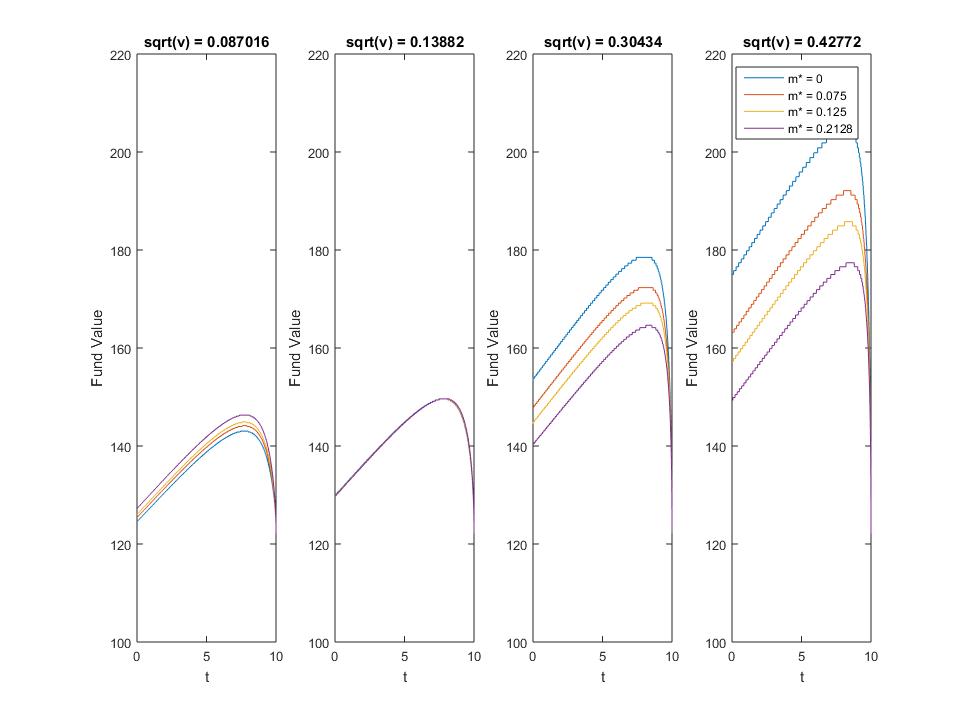}}
	\end{tabular}
	\caption{\small{ The $y$ section of the approximated optimal surrender surface under the Heston model, $f^{(m,N)}_y$, for different volatility levels $\sqrt{y}$ and fair multipliers $\tilde{m}^*$, $G=F_0 e^{\tilde{g} T}$ with $\tilde{g}=2\%$.}}\label{figBoundaryAll_Gmodif2}
\end{figure}

\newpage
\newpage
\subsection{Numerical Analysis under 3/2 Model}\label{subsectNumAnal32}
In this section, we analyse numerically the impact of VIX-linked fee incentives under the 3/2 model.
\subsubsection{Market, VA and CTMC Parameters}\label{subsectParam32}

In order for the results under the 3/2 model to be comparable to the ones obtained under the Heston model, the market parameters ($\theta$, $\kappa$ and $\sigma$ for the 3/2 model) are calibrated to at-the-money call options\footnote{We consider 4 options with maturity $T=0.5,\,2.5,\,5$ and $10$. We used $S_0=K=100$ and a dividend yield of $1.5338\%$.} priced using the Heston model with the market parameters in Table \ref{tblMarketParam}. The initial value of the variance is set to $V_0=0.03$ for the Heston model and $V_0=1/0.03$ for the 3/2 model. The correlation $\rho=-0.75$ and the risk-free rate $r=0.03$ are assumed to be the same under both models. The resulting market parameters are presented in Table \ref{tblMarketParam32} and the model dynamics is given in Table \ref{tblExSVmodels}\\
\begin{table}[h!]
	\begin{tabular}{ccccccc}
		\hline
		  Parameter   & ${V_0}$ & ${\kappa}$ & ${\theta}$ & ${\sigma}$ & $\mathbf{\rho}$ & ${r}$\\
		Value    & $1/0.03$ & $5.7488$ & $46.1326$ & $15.4320$ & $-0.75$ & $0.03$\\ 
		\hline
	\end{tabular}
	\caption{Market parameters for 3/2 model}
	\label{tblMarketParam32}
\end{table}  

Unless stated otherwise, all numerical experiments in this section are performed using the CTMC parameters listed in Table \ref{tblCTMCparam32}. Note that for the 3/2 model, the VIX is approximated using Algorithm \ref{algoCTMCapproxVIX} with $n=1,000$ time steps, whereas in the Heston model, the volatility index is calculated using a closed-form formula.
\begin{table}[h!]
	\begin{tabular}{ccccccccccc}
		\hline
		 Parameter & ${m}$ & ${N}$ & ${v_1}$ & ${v_m}$ & $\tilde{\alpha}_v$& $x_1$ & $x_N$ & $\tilde{\alpha}_X$ & $M$ &n\\
		Value &$1,000$ & $1,000$ & $V_0/200$ & $8 v_0$ & $0.5764$ & $-X_0$ & $2 X_0$ &$2/100$ &$5,000$ &1,000   \\ 
		\hline
	\end{tabular}
	\caption{CTMC Parameters for 3/2 Model}
	\label{tblCTMCparam32}
\end{table}

\subsubsection{Fee Structures and Fair Fee Parameters}
As for the Heston model, we consider three fee structures: the Uncapped $\vix^2$, the Capped $\vix^2$ and the Uncapped $\vix$, see Section \ref{subsectFeeStruct} for details. The fair parameters are calibrated using a CTMC approximation with $N=100$ (all other CTMC parameters are the same as in Table \ref{tblCTMCparam32}), to reduce the computation time. Table \ref{tblFairFee32} presents the fair fee vectors  $(\tilde{c}^*, \tilde{m}^*)$  obtained under the 3/2 model. Note how close the fair fee vectors obtained under the 3/2 model (Table \ref{tblFairFee32}) are from the ones obtained under the Heston model (Table \ref{tblFairFeeHeston}).\\

\begin{table}[h!]
			\begin{tabular}{c cccc}
				\hline
				& \multicolumn{4}{c}{$c_t=\tilde{c}^*+\tilde{m}^*\vix^2_t$}\\
				\hline
				$\tilde{m}^*$ & $0.0000$ & $0.15000$ & $0.3000$ & $0.4235$\\
				
				$\tilde{c}^*$ & $1.5273\%$  & $0.9858\%$ & $0.4449\%$ & $0.0000\%$\\
				\hline
				\hline
				& \multicolumn{4}{c}{$c_t=\min\{\tilde{c}^*+\tilde{m}^*\vix^2_t,K\}$, $K=2\%$}\\
				\hline
				$\tilde{m}^*$ & $0.0000$ & $0.1500$ & $0.3000$ & $0.5791$\\
				
				$\tilde{c}^*$ & $1.5273\%$  & $1.0335\%$ & $0.6351\%$ & $0.0000\%$\\
				\hline
				\hline
				& \multicolumn{4}{c}{$c_t=\tilde{c}^*+\tilde{m}^*\vix_t$}\\
				\hline
				$\tilde{m}^*$ & $0.0000$ & $0.0250$ & $0.0500$ & $0.0846$\\
				
				$\tilde{c}^*$ & $1.5273\%$  & $1.0757\%$ & $0.6241\%$ & $0.0000\%$\\
				\hline
		\end{tabular}
	\caption{Fair fee vectors $(\tilde{c}^*,\tilde{m}^*)$ under 3/2 model}
	\label{tblFairFee32}
\end{table} 

\subsubsection{Effect of VIX-Linked Fees on Surrender Incentives}
Recall from Table \ref{tblExSVmodels} that under the 3/2 model, the dynamics of the index and the variance processes are given by\footnote{Note that in this formulation of the 3/2 model, the process $V$ represents the inverse of the variance process. It is also common to see the 3/2 model expressed in terms of its variance, see for instance \cite{drimus2012options}.}
\begin{equation}
	\begin{array}{ll}
		\diff S_t  &= r S_t \diff t+\frac{1}{\sqrt{V_t}} S_t \diff W_t^{(1)},\\
		\diff V_t & =\kappa(\theta -V_t)\diff t+\sigma\sqrt{V_t}\diff W^{(2)}_t,\label{eqEDS_F32_Q}
	\end{array}
\end{equation}
where $W=(W^{(1)}, W^{(2)})^T$ is a bi-dimensional correlated Brownian motion under $\prob{Q}$ and such that $[W^{(1)},\,W^{(2)}]_t=\rho t$ with $\rho\in[-1,0]$\footnote{The parameter $\rho$ is assumed to be non-positive for the martingale measure to exist, see Footnote \ref{footnoteMGmeasure32} for details.}, and $\kappa,\theta, \sigma>0$ with $2\kappa\theta>\sigma^2$.\\
\\
Now from Lemma \ref{lemmaDecoupleBM}, we find that $\gamma(x)=\frac{\ln(x)\rho}{\sigma}$. Thus, given a certain fee process $\{c_t\}_{0\leq t\leq T}$ (see Subsection \ref{subsectFeeStruct} for examples), the dynamics of the auxiliary process is obtained as follows
\begin{equation}
	\begin{array}{ll}
		\diff X_t  &= \mu_X(X_t,Y_t) \diff t+\sigma_X(Y_t)\diff W^{*}_t,\\
		\diff V_t&=\mu_V(V_t)\diff t+\sigma_V(V_t)\diff W^{(2)}_t,\label{eqEDS_X_Q32}
	\end{array}
\end{equation} 
where $\mu_X(X_t,V_t)=r-c_t+\frac{\rho\kappa}{\theta}+\frac{1}{V_t}\left(\frac{\rho\sigma}{2}-\frac{1}{2}-\frac{\rho\kappa\theta}{\sigma}\right)$, and $\sigma_X(V_t)=\sqrt{\frac{1-\rho^2}{V_t}}$, $0\leq t\leq T$.\\
\\
When the fee is tied to the VIX, the form of the fee function is not known explicitly at this point under the 3/2 model, since the the VIX does not have a closed-form expression (see footnote \ref{footnoteApproxVIX32} for details). However, as illustrated in \eqref{eqEDS_Xm}, when the inverse variance process $V$ is replaced by its CTMC approximation $V^{(m)}$, the auxiliary process $X$ becomes a RS diffusion $X^{(m)}$ with the following dynamics:
\begin{equation}
	\begin{array}{ll}
		\diff X_t^{(m)}  & = r-c_t^{(m)}+\frac{\rho\kappa}{\theta}+\frac{1}{V_t^{(m)}}\left(\frac{\rho\sigma}{2}-\frac{1}{2}-\frac{\rho\kappa\theta}{\sigma}\right) \diff t+\sigma_X(V_t^{(m)})\diff W^{*}_t,\quad t\geq 0,\label{eqEDS_X_Q32_V2}
	\end{array}
\end{equation} 
where  $c_t^{(m)}=c(X_t^{(m)},V_t^{(m)})$ is the CTMC approximation of the fee process.\\
\\
Recall that $\vix^{(m)}$ is the CTMC approximation of the volatility index, see Proposition \ref{propCTMCapproxVIX} and Algorithm \ref{algoCTMCapproxVIX}. The three fee structures exposed in Subsection \ref{subsectFeeStruct} can then be approximated using  $\vix^{(m)}$ as shown in Table \ref{tblFeeProcessCTMCapprox}.
\begin{table}[h!]
	\begin{tabular}{|c|c|}
		\hline
		\textbf{Fee Structure} & $c_t^{(m)}$, $0\leq t\leq T$\\
		\hline\hline
		\textbf{Uncapped $\vix^2$}  & $c_t^{(m)}=\tilde{c}+\tilde{m}(\vix_t^{(m)})^2$\\
		\textbf{Capped $\vix^2$}  & $c_t^{(m)}=\min\{K,\tilde{c}+\tilde{m}(\vix_t ^{(m)})^2$\}\\
		\textbf{Uncapped $\vix$}  & $c_t^{(m)}=\tilde{c}+\tilde{m}\vix_t ^{(m)}$\\
		\hline
	\end{tabular}
	\caption{CTMC approximation of the VIX-linked fee process}
	\label{tblFeeProcessCTMCapprox}
\end{table}  
\\
\\
The second layer CTMC approximation of Subsection \ref{Chap2sectionCTMC_RS} is then applied to the regime-switching diffusion \eqref{eqEDS_X_Q32_V2} with the approximated fee processes listed in Table \ref{tblFeeProcessCTMCapprox}.\\
\\
  Using the CTMC technique outlined in Section \ref{sectionCTMCapprox}; and the market, the variable annuity, and the CTMC parameters of Subsection \ref{subsectParam32}, we performed the valuation of a variable annuity with and without early surrenders. The results are similar to the ones obtained under the Heston model and are summarized below, confirming again that fees tied to the volatility index do not significantly affect the value of surrender incentives.\\
 \\
The value of early surrenders (``ES values'') under the 3/2 model are presented in Table \ref{tblESvalue32}. 
\begin{table}[h!]
	\begin{tabular}{c cccc}
		\hline
		& \multicolumn{4}{c}{$c_t=\tilde{c}^*+\tilde{m}^*\vix^2_t$}\\
		\hline
		$\tilde{m}^*$ &  $0.0000$ & $0.1500$ & $0.3000$ & $0.4235$\\
		
		ES Value & $2.91799$  & $2.91126$ & $2.90875$ & $2.91078$\\
		\hline
		\hline
		& \multicolumn{4}{c}{$c_t=\min\{\tilde{c}^*+\tilde{m}^*\vix^2_t,K\}$, $K=2\%$}\\
		\hline
		$\tilde{m}^*$ & $0.0000$ & $0.1500$ & $0.3000$ & $0.5791$\\
		
		ES Value  & $2.91799$  & $2.90594$ & $2.89435$ & $2.88076$\\
		\hline
		\hline
		& \multicolumn{4}{c}{$c_t=\tilde{c}^*+\tilde{m}^*\vix_t$}\\
		\hline
		$\tilde{m}^*$ & $0.0000$ & $0.0250$ & $0.0500$ & $0.0846$\\
		
		ES Value  & $2.9180$  & $2.9120$ & $2.9068$ & $2.9012$\\
		\hline
	\end{tabular}
	\caption{Approximated early surrender values (ES values) under the 3/2 model.}
	\label{tblESvalue32}
\end{table}  

The approximated optimal surrender surfaces for the three fee structures under the 3/2 model are displayed in Figures \ref{figBoundary3D_32}, \ref{figBoundary_32} and \ref{figBoundaryAll_32}. Note that in order for the Figures under the Heston and 3/2 models to be comparable, the $y$-axis under the 3/2 model represents the variance of the fund, that is $\frac{1}{V_t}$,  $0\leq t\leq T$ (recall that $V$ is the inverse variance in  \eqref{eqEDS_F32_Q}).
\begin{figure}[h!]
	\centering
	\begin{tabular}{cccc}
		\multicolumn{4}{c}{\small{$c_t=\tilde{c}^*+\tilde{m}^*\vix^2_t$}}\\
		\includegraphics[scale=0.2]{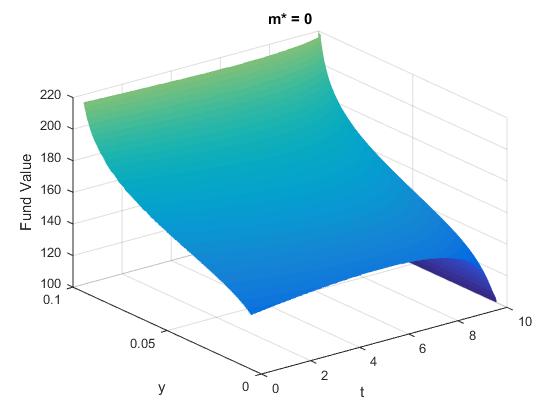}  & \includegraphics[scale=0.2]{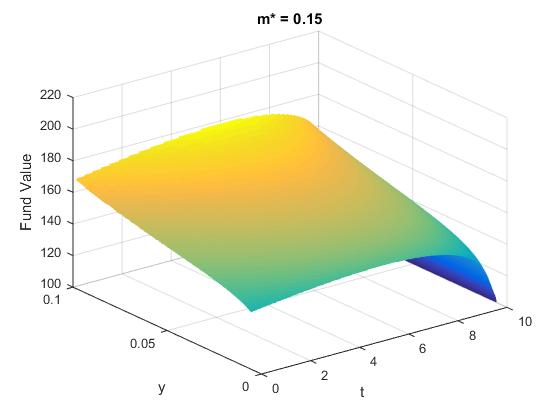}   & \includegraphics[scale=0.2]{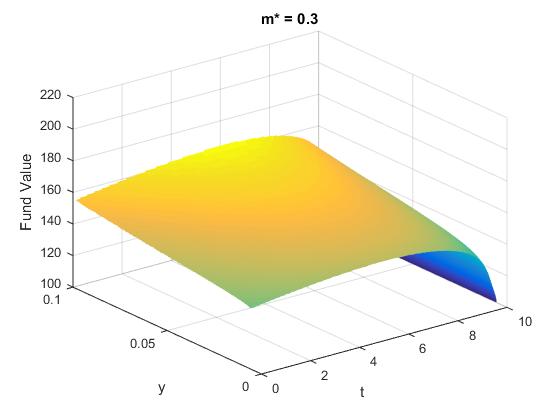}  & \includegraphics[scale=0.2]{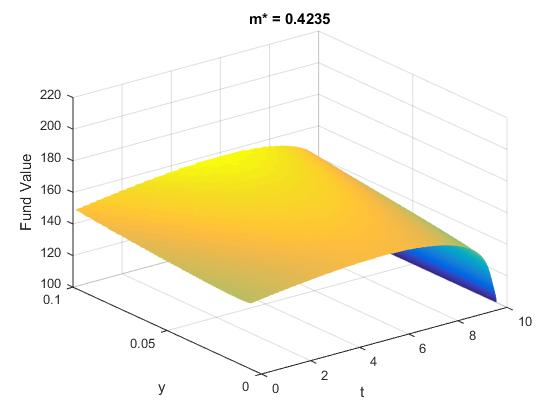} \\
		\multicolumn{4}{c}{\small{$c_t=\min\{\tilde{c}^*+\tilde{m}^*\vix^2_t,K\}$, $K=2\%$}}\\
		\includegraphics[scale=0.2]{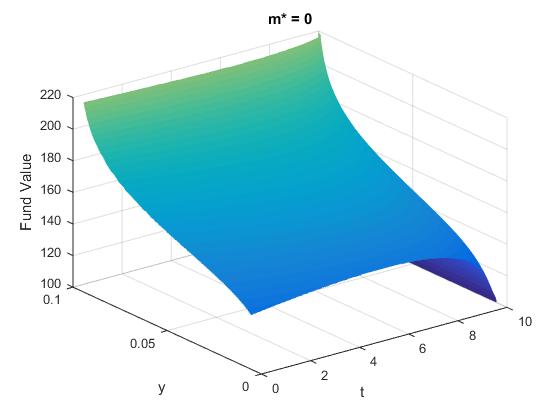} & \includegraphics[scale=0.2]{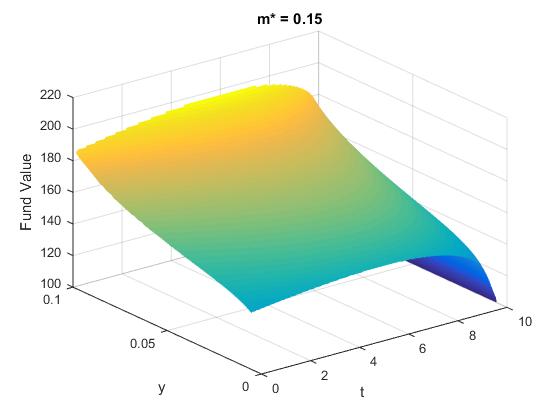}&	\includegraphics[scale=0.2]{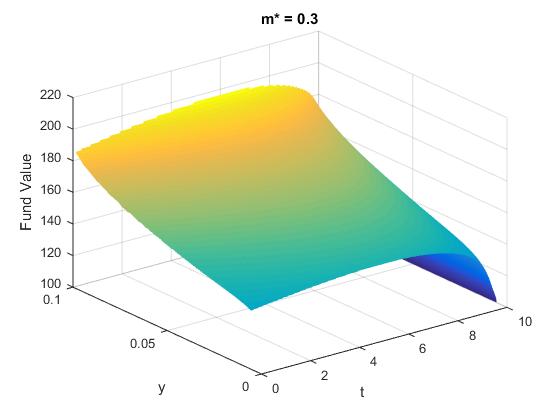} &	\includegraphics[scale=0.2]{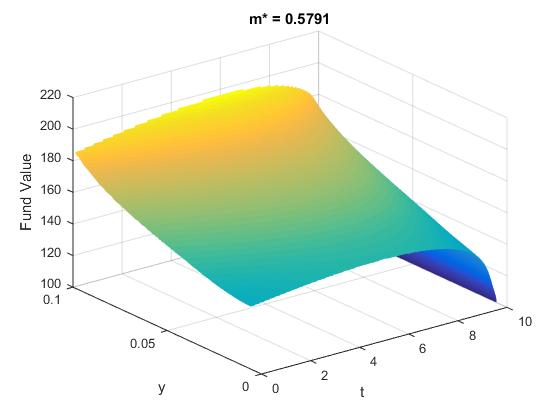} \\
		\multicolumn{4}{c}{\small{$c_t=\tilde{c}^*+\tilde{m}^*\vix_t$}}\\
		\includegraphics[scale=0.2]{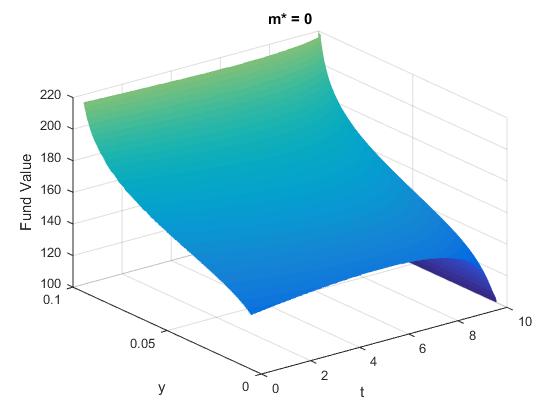} &\includegraphics[scale=0.2]{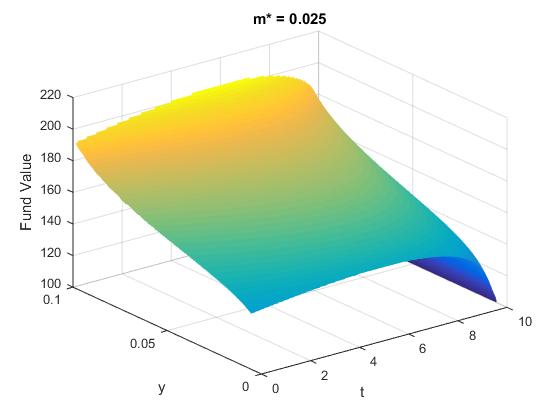}  &\includegraphics[scale=0.2]{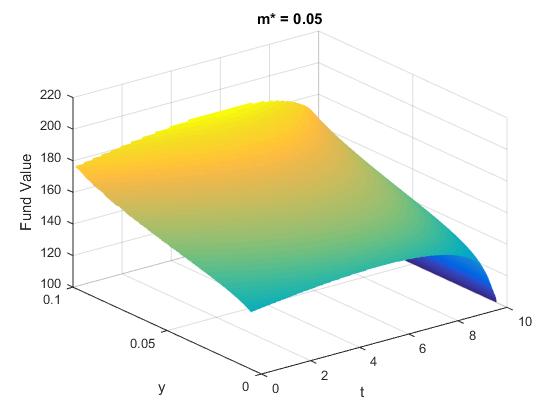}  &
		\includegraphics[scale=0.2]{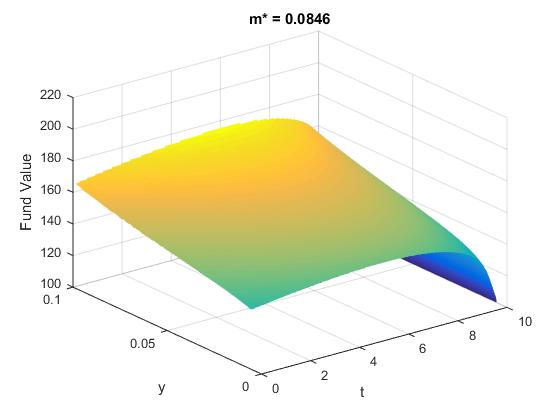}  
	\end{tabular}
	\caption{\small{Approximated optimal surrender surface for different values of fair multiplier $\tilde{m}^*$ under the 3/2 model where $x$-axis represents the time and the $y$-axis the variance.}}\label{figBoundary3D_32}
\end{figure}
As noted previously, we observe that the optimal surrender surface is increasing with the volatility. However, when $\tilde{m}^*=0$, the surrender surface increases much faster than under the Heston model. We also note that $\vix$-linked fees help to neutralize the effect of the volatility on the optimal surrender decision, confirming again the relation existing between the fees and the optimal surrender strategies; that is, VA contracts with high fees are surrendered at lower fund values than VA contracts with low fees. 
\begin{figure}[!t]
	\centering
	\begin{tabular}{cccc}
		\multicolumn{4}{c}{\small{$c_t=\tilde{c}^*+\tilde{m}^*\vix^2_t$}}\\
		\includegraphics[scale=0.2]{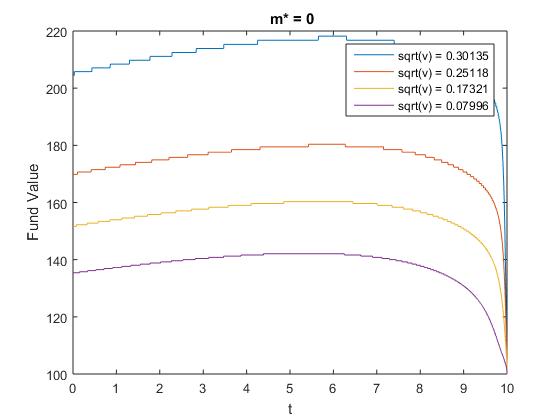}   &\includegraphics[scale=0.2]{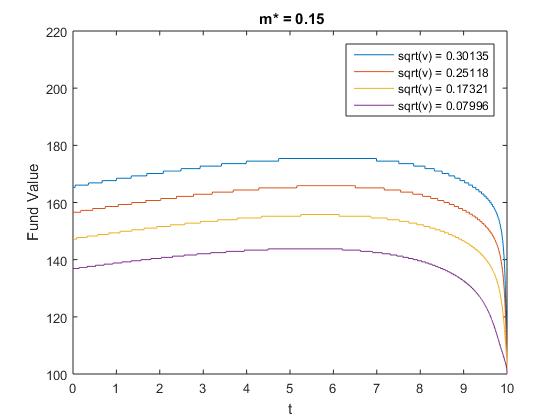}   & \includegraphics[scale=0.2]{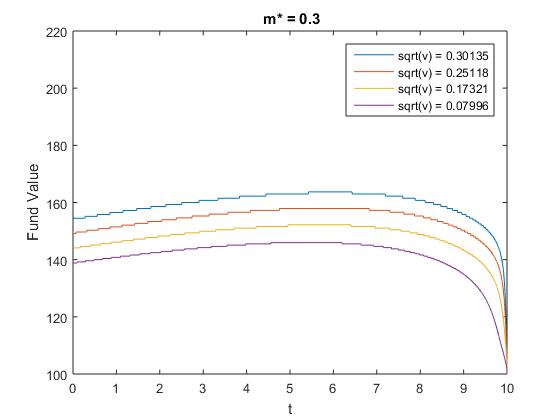}  & \includegraphics[scale=0.2]{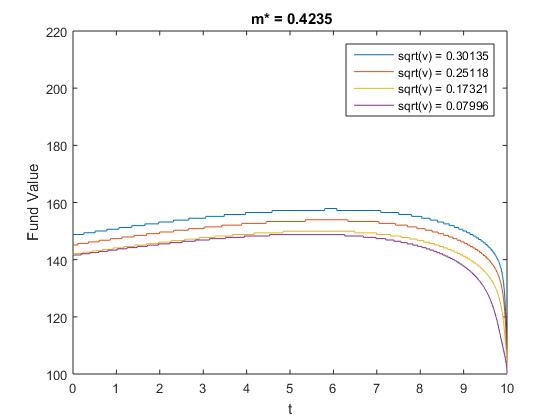} \\
		\multicolumn{4}{c}{\small{$c_t=\min\{\tilde{c}^*+\tilde{m}^*\vix^2_t,K\}$, $K=2\%$}}\\
		\includegraphics[scale=0.2]{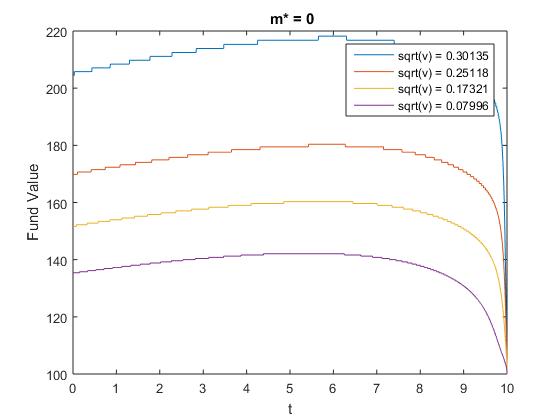}  & 	\includegraphics[scale=0.2]{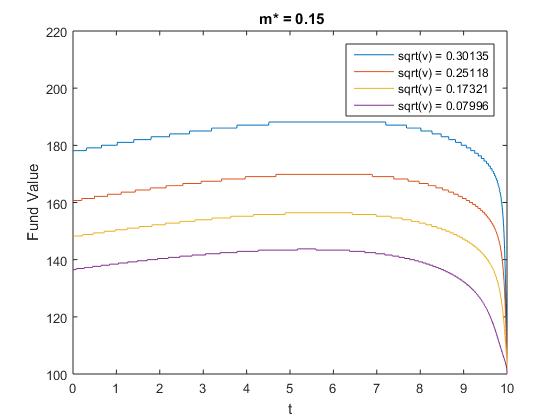} &		\includegraphics[scale=0.2]{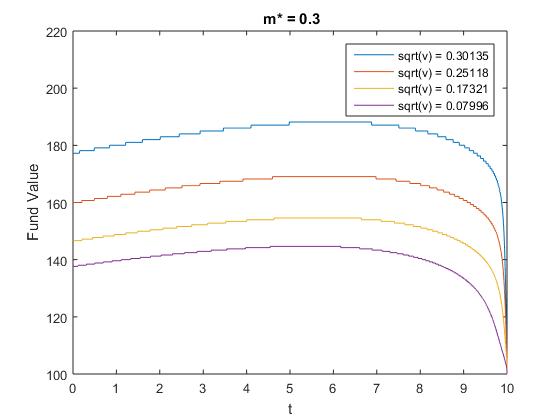}   &		\includegraphics[scale=0.2]{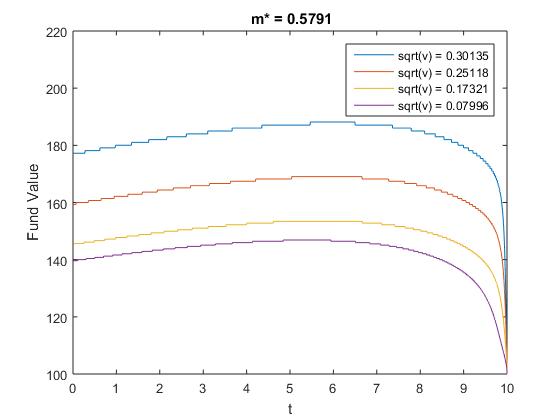}   \\
		\multicolumn{4}{c}{\small{$c_t=\tilde{c}^*+\tilde{m}^*\vix_t$}}\\
		\includegraphics[scale=0.2]{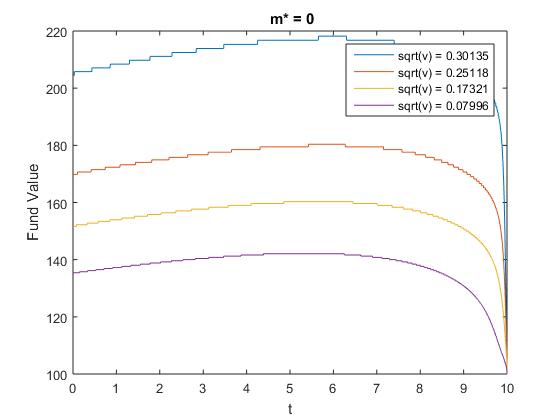}   &\includegraphics[scale=0.2]{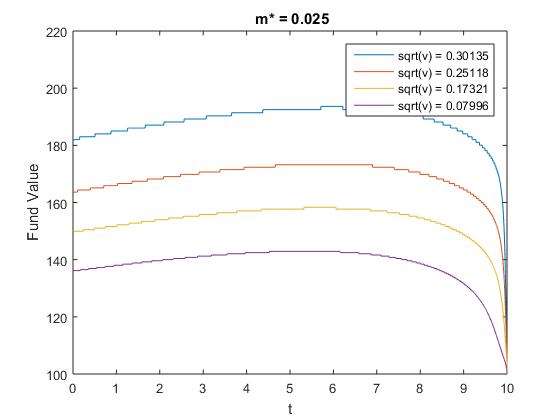}    &\includegraphics[scale=0.2]{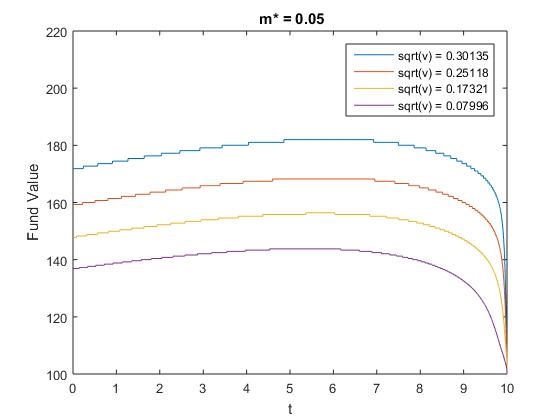}    &
		\includegraphics[scale=0.2]{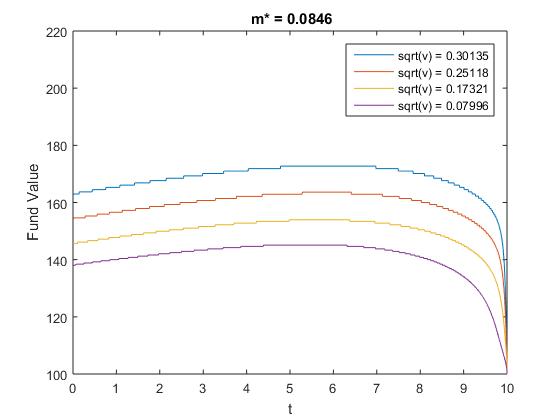}   
	\end{tabular}
	\caption{\small{ The $y$ section of the approximated optimal surrender surface, $f^{(m,N)}_y$, for different volatility levels $\sqrt{y}$ and fair multipliers $\tilde{m}^*$.}}\label{figBoundary_32}
\end{figure}
\newpage
\begin{figure}[!t]
	\centering
	\begin{tabular}{cc}
		\small{$c_t=\tilde{c}^*+\tilde{m}^*\vix^2_t$} & \small{$c_t=\min\{\tilde{c}^*+\tilde{m}^*\vix^2_t,K\}$, $K=2\%$} \\
		\includegraphics[scale=0.15]{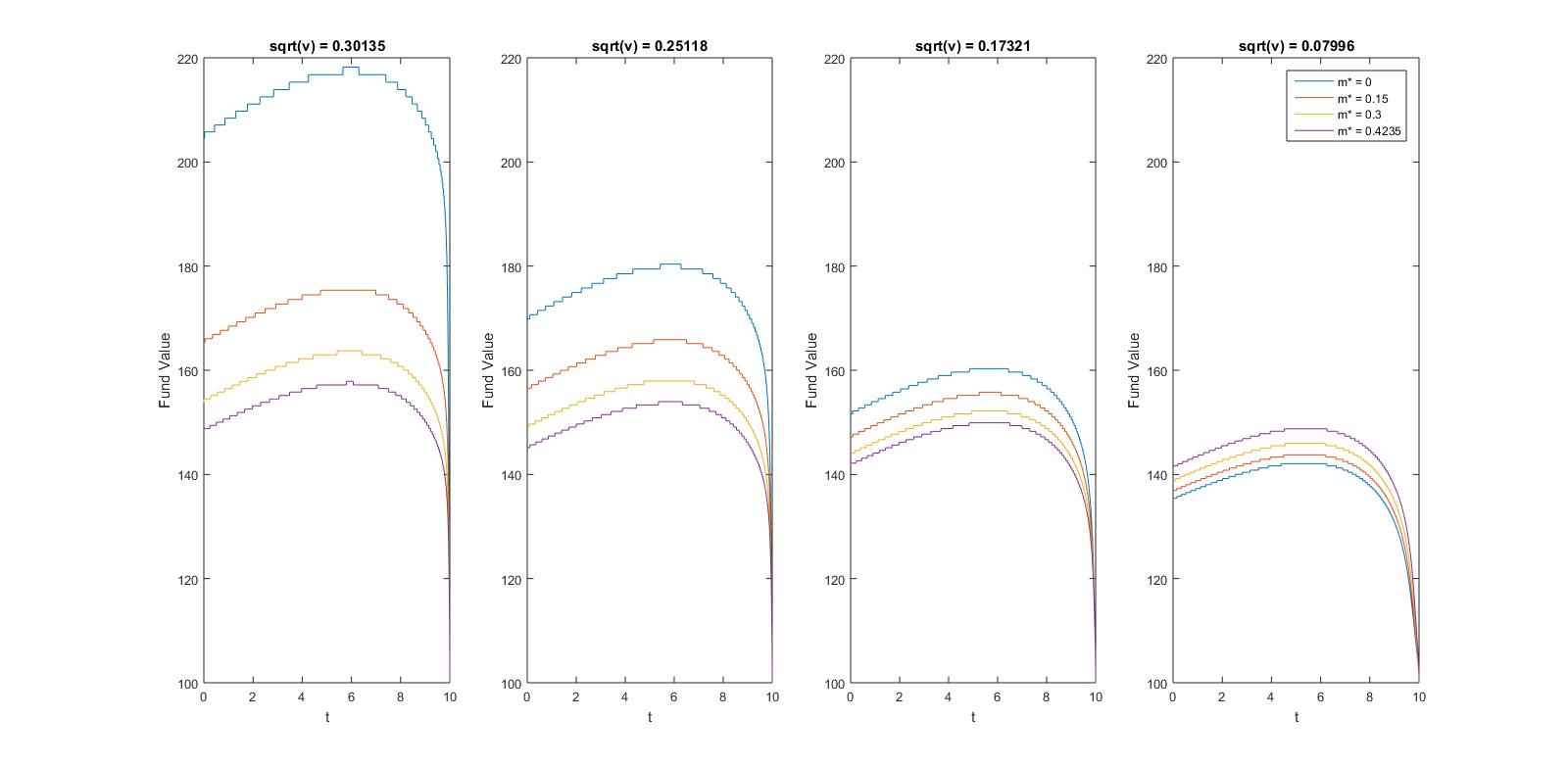}   &\includegraphics[scale=0.15]{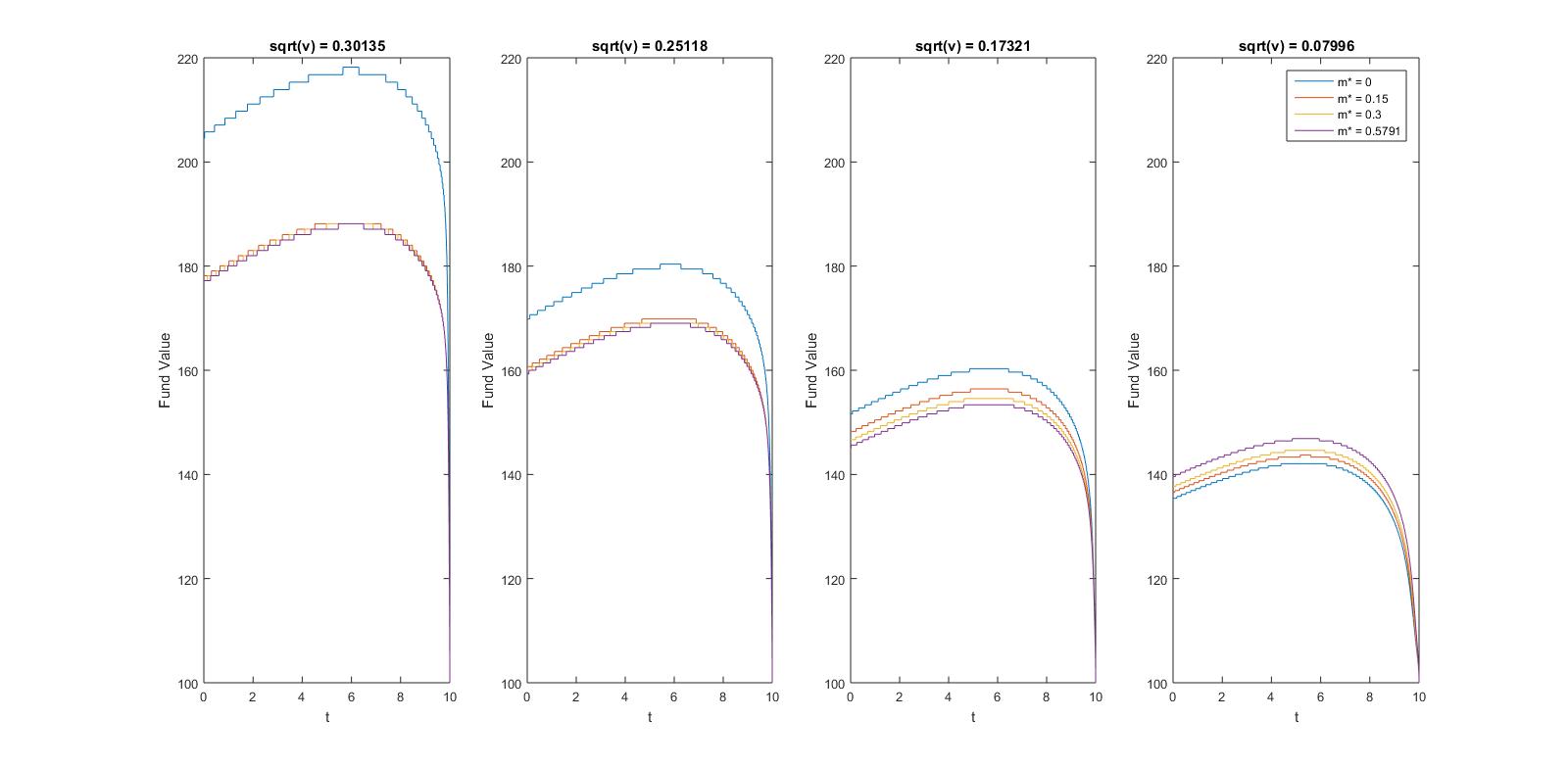}  \\
		\multicolumn{2}{c}{\small{$c_t=\tilde{c}^*+\tilde{m}^*\vix_t$}}\\
		\multicolumn{2}{c}{\includegraphics[scale=0.15]{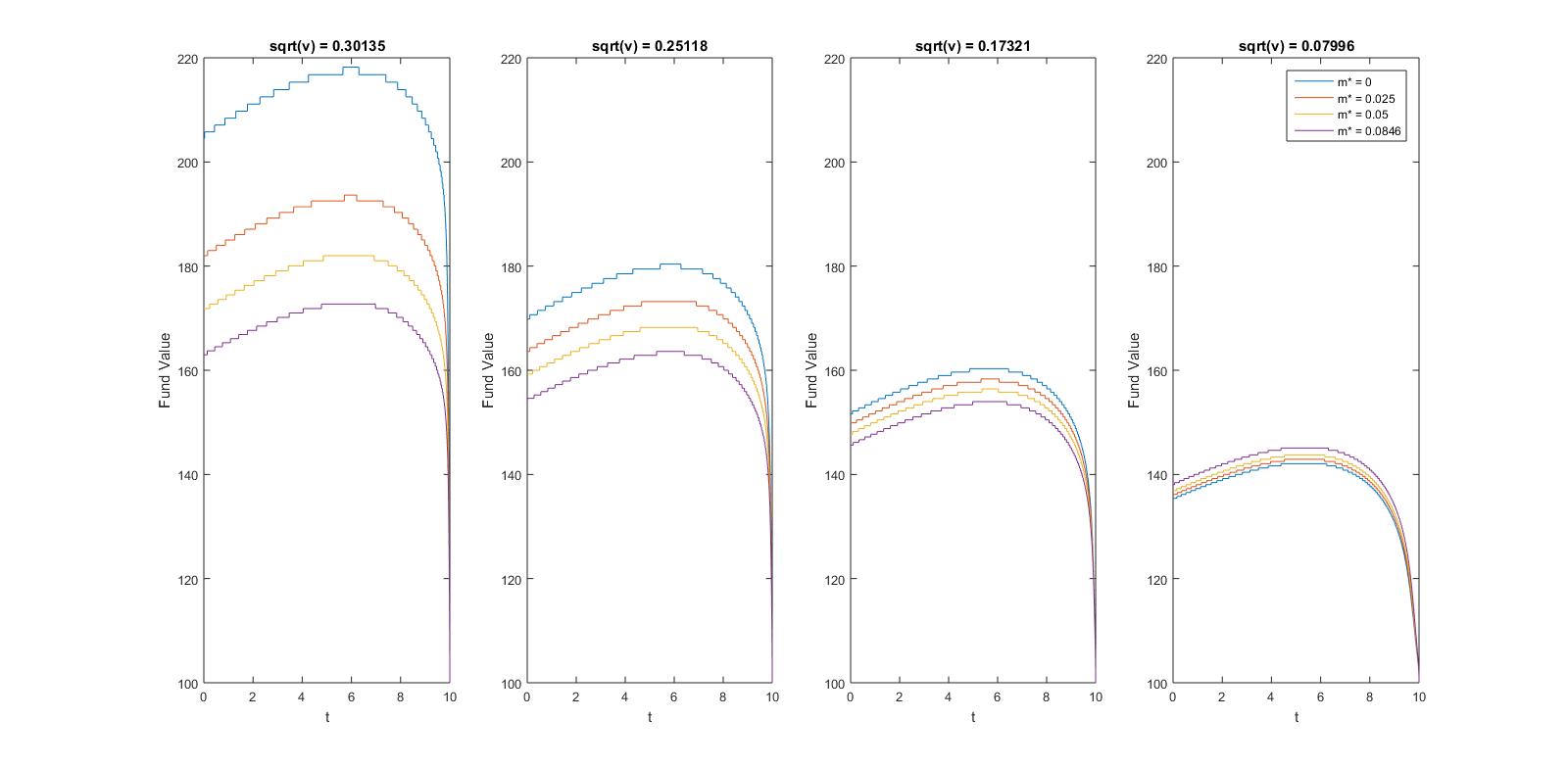}}
	\end{tabular}
	\caption{\small{ The $y$ section of the approximated optimal surrender surface, $f^{(m,N)}_y$, for different volatility levels $\sqrt{y}$ and fair multipliers $\tilde{m}^*$.}}\label{figBoundaryAll_32}
\end{figure}

\clearpage

	%
	%
\end{document}